\documentclass[acmsmall]{acmart}

\acmJournal{PACMPL}
\acmVolume{1}
\acmNumber{OOPSLA} 
\acmArticle{1}
\acmYear{2022}
\acmMonth{12}
\acmDOI{} 
\startPage{1}

\setcopyright{none}

\usepackage{booktabs}   
\usepackage{subcaption} 
\usepackage{thm-restate}
\usepackage{xcolor}
\usepackage{subcaption}
\usepackage{bbm}
\usepackage{listings}
\usepackage[inline]{enumitem}

\newtheorem{theorem}{Theorem}

\newtheorem{lemma}{Lemma}
\theoremstyle{definition}
\newtheorem{definition}{Definition}
\newtheorem{example}{Example}

\makeatletter
\newcommand{\xRightarrow}[2][]{\ext@arrow 0359\Rightarrowfill@{#1}{#2}}
\makeatother

\usepackage{varwidth}
\usepackage{enumitem}

\usepackage{amsmath}
\usepackage{graphicx}
\usepackage{algorithm}
\usepackage{algorithmicx}
\usepackage[noend]{algpseudocode}
\usepackage{color}
\usepackage{xspace}

\newcommand{\dom}{\operatorname{\mathbf{dom}}}

\usepackage{tikz}
\usetikzlibrary{arrows,backgrounds,decorations.pathmorphing,automata,shapes,positioning,chains,calc}
\usepackage[inference]{semantic}

\usepackage{marginfix} 

\newcommand{\OMIT}[1]{}

\usepackage{cutwin}

\newcommand{\wtms}{\textsc{TMS2-ra}\xspace}
\newcommand{\tmlra}{\textsc{TML-ra}\xspace}
\newcommand{\tmlsc}{\textsc{TML-sc}\xspace}
\newcommand{\glb}{{\it glb}}

\def\eg{e.g.,\xspace}
\def\ie{i.e.,\xspace}
\def\etal{{et al}\xspace}
\def\cf{cf.\xspace}

\usepackage{mathpartir}

\usepackage{pgf}
\usepackage{pgfplots}
\pgfplotsset{compat=1.17}

\usepackage{subcaption}
\newcommand{\linefill}{\cleaders\hbox{$\smash{\mkern-2mu\mathord-\mkern-2mu}$}\hfill\vphantom{\lower1pt\hbox{$\rightarrow$}}}  
  
\newcommand{\transi}[2]{\mathrel{\lower1pt\hbox{$\mathrel-_{\vphantom{#2}}\mkern-8mu\stackrel{#1}{\linefill_{\vphantom{#2}}}\mkern-11mu\rightarrow_{#2}$}}}
\newcommand{\trans}[1]{\transi{#1}{{}}}
\newcommand{\eseq}[1]{\langle~\rangle}

\newcounter{sarrow}
\newcommand\strans[1]{%
  \mathrel{\raisebox{0.1em}{
    \stepcounter{sarrow}%
    \!\!\!\!
    \begin{tikzpicture}
      \node[inner sep=.5ex] (\thesarrow) {$\scriptstyle #1$};
      \path[draw,<-,decorate,line width=0.25mm,
      decoration={zigzag,amplitude=0.7pt,segment length=1.2mm,pre=lineto,pre length=4pt}] 
      (\thesarrow.south east) -- (\thesarrow.south west);
    \end{tikzpicture}%
  }}}

\newcommand{\nat}{\mathbb{N}}
\newcommand{\noteq}{\neq}

\newcommand{\Var}{\mathit{Var}}
\newcommand{\Reg}{\mathit{Reg}}

\newcommand{\LComm}{\mathit{LCom}}
\newcommand{\prog}{\Pi}

\newcommand{\pc}{{\it pc}}

\newcommand{\stepgoto}[2]{#1\ {\bf goto}\ #2 }
\newcommand{\ifgoto}[3]{\kwif\ #1\ {\bf goto}\ #2\ {\bf elseto}\ #3 }

\newcommand{\txbegin}[1]{{\tt TxBegin}(#1)}
\newcommand{\txread}[2]{{\tt TxRead}(#2, #1)}
\newcommand{\txwrite}[2]{{\tt TxWrite}(#1, #2)}
\newcommand{\txend}{{\tt TxEnd}}

\newcommand{\Label}{\mathit{Label}}
\newcommand{\GVar}{{\it GVar}}
\newcommand{\LVar}{\mathit{LVar}}
\newcommand{\Val}{\mathit{Val}}

\newcommand{\R}{\mathsf{R}}
\newcommand{\A}{\mathsf{A}}

\newcommand{\refeq}[1]{(\ref{#1})}

\newcommand{\True}{{\it true}}
\newcommand{\False}{{\it false}}

\newcommand{\var}{\mathtt{var}}
\newcommand{\valu}{\mathtt{val}}

\newcommand{\imp}{\Rightarrow}

\newcommand{\last}{\mathit{last}}

\newcommand{\kwlet}{\textsf{\textbf{let}}}
\newcommand{\kwin}{\textsf{\textbf{in}}}

\newcommand{\pre}{\textsf{\textbf{pre}}}
\newcommand{\eff}{\textsf{\textbf{eff}}}

\newcommand{\kwcas}{\textsf{\textbf{CAS}}}
\newcommand{\kwswap}{\textsf{\textbf{swap}}}
\newcommand{\kwskip}{\bot}
\newcommand{\kwdo}{\textsf{\textbf{do}}}
\newcommand{\kwwhile}{\textsf{\textbf{while}}}
\newcommand{\kwend}{\textsf{\textbf{end}}}
\newcommand{\kwif}{\textsf{\textbf{if}}}
\newcommand{\kwthen}{\textsf{\textbf{then}}}
\newcommand{\kwelse}{\textsf{\textbf{else}}}
\newcommand{\kwreturn}{\textsf{\textbf{return}}}
\newcommand{\kwthread}{\textsf{\textbf{thread}}}
\newcommand{\kwuntil}{\textsf{\textbf{until}}}

\algnewcommand\Swap{\kwswap}
\algnewcommand\Skip{\kwskip}
\algnewcommand\Thread{\kwthread}

\algrenewcommand\algorithmicend{\kwend}
\algrenewcommand\algorithmicdo{\kwdo}
\algrenewcommand\algorithmicwhile{\kwwhile}
\algrenewcommand\algorithmicfor{\textsf{\textbf{for}}}
\algrenewcommand\algorithmicforall{\textsf{\textbf{for all}}}
\algrenewcommand\algorithmicloop{\textsf{\textbf{loop}}}
\algrenewcommand\algorithmicrepeat{\textsf{\textbf{repeat}}}
\algrenewcommand\algorithmicuntil{\textsf{\textbf{until}}}
\algrenewcommand\algorithmicprocedure{\textsf{\textbf{procedure}}}
\algrenewcommand\algorithmicfunction{\textsf{\textbf{function}}}
\algrenewcommand\algorithmicif{\kwif}
\algrenewcommand\algorithmicthen{\kwthen}
\algrenewcommand\algorithmicelse{\kwelse}
\algrenewcommand\algorithmicreturn{\kwreturn}

\algblockdefx{Thread}{EndThread}%
[1]{\kwthread \xspace #1}%
{\algorithmicend}

\algblockdefx{MyWhile}{EndMyWhile}%
[1]{\kwwhile \xspace #1}%
{\algorithmicend}

\makeatletter
\ifthenelse{\equal{\ALG@noend}{t}}%
  {\algtext*{EndMyWhile}}
  {}%
\makeatother

\algblockdefx{MyUntil}{EndMyUntil}%
[1]{\kwuntil \xspace #1}%
{\algorithmicend}
\makeatletter
\ifthenelse{\equal{\ALG@noend}{t}}%
  {\algtext*{EndMyUntil}}
  {}%
\makeatother

\makeatletter
\ifthenelse{\equal{\ALG@noend}{t}}%
  {\algtext*{EndThread}}
  {}%
\makeatother

\newcommand{\orangearrow}{%
  \(\tikz[baseline=-0.5ex]{\draw (0,0) edge [draw, dotted, orange!70!black, line width=0.4mm,-stealth]
    (0.6,0);} \)}

\newcommand{\xxval}{{\sf val}}
\newcommand{\xxloc}{{\sf loc}}
\newcommand{\xxtst}{{\sf tst}}
\newcommand{\tst}{\mathtt{tst}}

\newcommand{\tidn}[1]{\tid{#1}}

\newcommand{\tid}{\tau}
\newcommand{\loc}{x}
\newcommand{\val}{v}
\newcommand{\locb}{y}
\newcommand{\valb}{u}

\newcommand{\ts}{{\it q}}
\renewcommand{\wr}{{\it w}}

\newcommand{\eqdef}{\mathrel{\hat{=}}}

\newcommand{\rel}{{\it Rel}}
\newcommand{\acq}{{\it Acq}}

\newcommand{\wrset}{{\it WS}}
\newcommand{\rdset}{{\it RS}}
\newcommand{\seenIdxs}{\mathtt{seenIdxs}}
\newcommand{\beginIdx}{\mathtt{beginIdx}}
\newcommand{\synctype}{\mathtt{synctype}}
\newcommand{\rset}{\mathtt{rdSet}}
\newcommand{\wset}{\mathtt{wrSet}}
\newcommand{\regSet}{\mathtt{regs}}
\newcommand{\status}{\mathtt{status}}

\newcommand{\isAcq}{\mathtt{isAcq}}
\newcommand{\memories}{\mathit{M}}
\newcommand{\mmview}{\mathtt{V}}
\newcommand{\mrel}{\mathit{S}}

\newcommand{\vismem}{\mathit{OM}}

\newcommand{\txn}{\mathtt{txn}}

\newcommand{\lst}{\mathit{lst}}
\newcommand{\als}{\mathit{als}}

\newcommand{\txid}{t}

\newcommand{\txnview}{{\it V}}

\newcommand{\logic}{\textsc{TARO}\xspace}

\newcommand{\assert}[1]{
  {\color{blue}
    \left\{
      \begin{array}[c]{@{}l@{}}
        #1
      \end{array}
    \right\}
  }
}

\newcommand{\fresh}{\mathit{fresh}}

\newcommand{\Prog}{{\it Prog}}

\newcommand{\AComm}{{\it ACom}}
\newcommand{\Exp}{{\it Exp}}
\newcommand{\BExp}{{\it BExp}}

\newcommand{\Init}{\mathbf{Init}}

\newcommand{\asgn}{\ensuremath{:=}}

\newcommand{\rat}{\mathbb{Q}}

\newcommand{\Tr}{{\it Tr_{SF}}}
\newcommand{\ct}{{\it ct}}
\newcommand{\at}{{\it at}}

\newcommand{\view}{\mathit{View}}
\newcommand{\tview}{{\tt tview}}
\newcommand{\txview}{{\tt txview}}
\newcommand{\relst}{{\tt rel}}

\newcommand{\ls}{\mathit{ls}}

\newcommand{\mview}{{\tt mview}}

\newcommand{\writes}{\mathtt{writes}}
\newcommand{\covered}{\mathtt{cvd}}

\newcommand{\Write}{{\it Write}}
\newcommand{\View}{{\it View}}

\newcommand{\OW}{{\it OW}}

\newcommand{\init}{{\it init}}

\newcommand{\maxWr}{{\it maxWr}}

\newcommand{\reffig}[1]{Fig.~\ref{#1}}
\newcommand{\refthm}[1]{Theorem~\ref{#1}}
\newcommand{\reflem}[1]{Lem\-ma~\ref{#1}}

\newcommand{\refsec}[1]{\S\ref{#1}}

\newcommand{\refdef}[1]{Definition~\ref{#1}}

\tikzset{
    mo/.style={dashed,->,>=stealth,thick,black!20!purple},
    hb/.style={solid,->,>=stealth,thick,blue},
    sw/.style={solid,->,>=stealth,thick,black!50!green},
    rf/.style={dashed,->,>=stealth,thick,black!50!green},
    fr/.style={dashed,->,>=stealth,thick,red}
 }

\begin{document}

\title[Implementing and Verifying Release-Acquire Transactional
  Memory (Extended Version)]{Implementing and Verifying Release-Acquire Transactional
  Memory (Extended Version)}


\author{Sadegh Dalvandi}
\email{m.dalvandi@surrey.ac.uk}
\orcid{0000-0001-8813-780X}
\affiliation{%
  \institution{University of Surrey}
  \city{Guildford}
  \country{UK}
}

\author{Brijesh Dongol}
\email{b.dongol@surrey.ac.uk}
\orcid{0000-0003-0446-3507}
\affiliation{%
  \institution{University of Surrey}
  \city{Guildford}
  \country{UK}
}
\begin{abstract}
    Transactional memory (TM) is an intensively studied
    synchronisation paradigm with many proposed implementations in
    software and hardware, and combinations thereof.
    However, TM under relaxed memory, \eg C11 (the 2011 C/C++
    standard) is still poorly understood, lacking rigorous foundations
    that support verifiable implementations.  This paper addresses
    this gap by developing \wtms, a relaxed operational TM
    specification. We integrate \wtms with RC11 (the repaired C11
    memory model that disallows load-buffering) to provide a formal
    semantics for TM libraries and their clients. We develop a logic,
    \logic, for {\em verifying} client programs that use \wtms for
    synchronisation. We also show how \wtms can be {\em implemented}
    by a C11 library, \tmlra, that uses relaxed and release-acquire
    atomics, yet guarantees the synchronisation properties required by
    \wtms. We {\em benchmark} \tmlra and show that it outperforms its
    sequentially consistent counterpart in the STAMP
    benchmarks. Finally, we use a simulation-based verification
    technique to {\em prove correctness} of \tmlra.  Our entire
    development is supported by the Isabelle/HOL proof
    assistant. 
  \end{abstract}



\maketitle

\section{Introduction}

The advent and proliferation of architectures implementing relaxed
memory models has resulted in many new challenges in the development
of concurrent programs. In the context of the C/C++ relaxed memory
model defined by C11\footnote{C11 refers to the 2011 ISO specification
  of C/C++. }, over a decade's worth of research has resulted in
rigorous semantic
foundations~\cite{DBLP:conf/popl/BattyDW16,DBLP:conf/popl/BattyOSSW11,DBLP:conf/pldi/LeeCPCHLV20,DBLP:conf/popl/KangHLVD17,DBLP:conf/pldi/LahavVKHD17,Batty2020},
and more recently, logics for reasoning about the correctness of
concurrent
programs~\cite{DBLP:journals/corr/abs-2108-01418,DBLP:journals/corr/abs-2004-02983,DBLP:conf/ppopp/DohertyDWD19,ECOOP20,DBLP:conf/ecoop/KaiserDDLV17,DBLP:conf/popl/KangHLVD17,DBLP:conf/esop/DokoV17,DBLP:conf/pdp/HeVQF16,DBLP:conf/icalp/LahavV15,vafeiadis2013relaxed}. These
works have provided the background necessary to develop high-level
abstractions and concurrency libraries over relaxed-memory
architectures. Recent works have included reimplementations of
concurrent data
structures~\cite{DBLP:conf/ppopp/DalvandiD21,DBLP:journals/pacmpl/RaadDRLV19,DBLP:conf/esop/KrishnaEEJ20,DBLP:journals/pacmpl/EmmiE19,DongolJRA18},
including those with relaxed specifications that aim to exploit the
additional behaviours allowed by relaxed memory.

Our aim for this paper is to implement and verify synchronisation
abstractions, fine-tuned for C11, in the form of {\em transactional
  memory (TM)} libraries, which provide reusable foundations for
high-performance, yet easy to manage concurrency
control~\cite{DBLP:conf/isca/HerlihyM93,DBLP:journals/dc/ShavitT97,2010Guerraoui}. Implementations
include those in software (as STM libraries) and hardware (Intel-RTM
and Armv9). Other variations include hybrid TM that combine software
and hardware TMs and implementations that are natively supported by
the compiler (\eg the continued C++ TM Lite development). In addition
to supporting general-purpose concurrency, TM has also been used to
develop transactional concurrent objects and data
structures~\cite{DBLP:conf/opodis/AssaMGKS21,DBLP:conf/ppopp/AssaMGKS20,DBLP:conf/podc/BronsonCCO10,DBLP:journals/pacmpl/LesaniXKBCPZ22}.
Intel’s persistent memory development kit (PMDK)~\cite{Scargall2020}
extensively promotes the transactional paradigm (though multi-threaded
transactions are not directly supported by PMDK's transactions). These
prior works have assumed SC transactions, i.e., that transactional
access provide the same guarantees as sequentially consistent memory.
Our focus is the verification of STMs implemented as a programming
language library with \emph{relaxed}, \emph{release}, \emph{acquire}
and \emph{release-acquire} accesses providing a pathway towards
simplified development of transactional objects (including concurrent
data structures) for relaxed memory.

TM implementations provide fine-grained interleaving (for efficiency)
that execute with an {\em illusion of atomicity} (for correctness). A
completed transaction may be committed or aborted so that all or none
of its effects are externally visible. TM implementations are designed
to satisfy a variety of correctness conditions such as (strict)
serialisability, opacity, and snapshot isolation, which restrict
ordering possibilities of completed transactions. TM has been
extensively studied for sequentially consistent (SC)
architectures~\cite{DBLP:journals/tc/Lamport79}, but implementations
over relaxed memory are limited. 

Prior works on relaxed memory transactions~(\eg
\cite{DBLP:conf/pldi/ChongSW18,DongolJR18}) have focussed on
foundations of {\em hardware transactions} and their interaction with
relaxed memory models, \eg the expected isolation guarantees,
reordering possibilities etc. The work of Chong
\etal~\cite{DBLP:conf/pldi/ChongSW18} also provides for semantics of
native C++ transactions. However, native TM support in C++ is still in
a state of flux~\cite{DBLP:journals/taco/ZardoshtiZBSS19,Spear2020}
and the underlying designs have changed since the original works by
Chong \etal~\cite{DBLP:conf/pldi/ChongSW18}. Moreover, these semantics
are presented in an axiomatic (aka declarative) style, which cannot be
used to verify TM implementations, where we require {\em operational}
descriptions of correctness. Therefore, our point of departure is a
separate set of works on TM specifications, in particular the TMS2
specification~\cite{DBLP:journals/fac/DohertyGLM13}, which has been
used extensively as a TM specification for standard (\ie SC)
architectures.

More recent works have taken steps towards C++ implementations,
including native support of TM within
C++~\cite{DBLP:journals/taco/ZardoshtiZBSS19} and STMs implemented
using C++ relaxed memory~\cite{DBLP:conf/podc/RodriguezS20}.  However,
\citet{DBLP:journals/taco/ZardoshtiZBSS19} do not describe
interactions with the C11 relaxed memory model, while Rodriguez and
Spear~\cite{DBLP:conf/podc/RodriguezS20} focus on data race freedom
and privatisation guarantees. Neither of these works have a formal
semantics, nor are they supported by a verification methodology. (See
\refsec{sec:related-work} for a more comprehensive survey of related
works.)

Our work addresses several gaps in the current state-of-the-art of
transactions for C/C++. We work with RC11, i.e., the {\em repaired
  C11} memory model~\cite{DBLP:conf/pldi/LahavVKHD17}. The RC11 memory
model disallows program-order and reads-from cycles, and hence
disallows load-buffering behaviour. This restriction greatly
simplifies reasoning and variants of RC11 are supported by a number of
different
logics~\cite{DBLP:conf/icalp/LahavV15,DBLP:journals/corr/abs-2004-02983,DBLP:conf/ppopp/DalvandiD21,ECOOP20,DBLP:conf/ecoop/KaiserDDLV17,DBLP:conf/pldi/DangJCNMKD22}. Logics
that address the full C11 memory model (allowing load buffering) have
also been developed, but proofs in these logics are limited to small
litmus
tests~\cite{DBLP:journals/corr/abs-2108-01418,DBLP:conf/esop/SvendsenPDLV18}.

We develop:
\begin{enumerate*}[label=(\roman*)]
\item a reusable {\em specification} of TM that provides well-defined
  guarantees to those developing client programs;
\item techniques for {\em verifying client programs} in C11 that use
  such TM abstractions;
\item {\em implementations} of TM in C11, including
  their rigorous verification; and
\item {\em mechanisation} of the verification described above in the
  theorem prover Isabelle/HOL.
\end{enumerate*}
We discuss these contributions in more detail below.

\paragraph{Correctness specifications.}  To enable verification,
we start with the \emph{TMS2}
specification~\cite{DBLP:journals/fac/DohertyGLM13}. TMS2 implies the
{\em TMS1} specification, which is known to be both necessary and
sufficient for {\em observational refinement} (of client
programs)~\cite{AttiyaGHR18}. The main difference between TMS1 and
TMS2 is that TMS1 allows aborted transactions to observe different
serialization orders~\cite{LLM12}. In contrast, TMS2, like
\emph{opacity}~\cite{2010Guerraoui}, ensures strict serializability of
the committed transactions and furthermore that aborted transactions
are consistent with the serialisation order. Although more restrictive
than TMS1, TMS2 has been shown to be a robust correctness condition
that is useful in practice, providing a specification for a number of TM
implementations under
SC~\cite{DBLP:journals/fac/DerrickDDSTW18,DBLP:conf/opodis/DohertyDDSW16,DBLP:conf/forte/ArmstrongDD17,DBLP:conf/forte/ArmstrongD17}.

Under relaxed memory, the TMS2 specification is inadequate since it
does not provide any of the client-side guarantees required by relaxed
memory
libraries~\cite{DBLP:conf/esop/RaadLV18,DBLP:journals/pacmpl/RaadDRLV19,DBLP:conf/ppopp/DalvandiD21,DongolJRA18}. Such
client-side guarantees are required under relaxed memory since writes
in one thread are not guaranteed to be propagated to other threads unless
the library is properly synchronised (\cf the message passing litmus
tests~\cite{AlglaveMT14}).

Our first contribution is the adaptation of TMS2 to address this
issue. In particular, our specification, \wtms, provides a flexible
meaning of correctness, allowing a client to specify {\em relaxed},
{\em releasing}, {\em acquiring} and {\em release-acquiring}
transactions (see \refsec{sec:release-acquire-tm}), mimicking the
memory annotations of C11
atomics~\cite{DBLP:conf/popl/BattyOSSW11}. This
provides greater flexibility in TM design; we develop a model in which
these different types of transactions co-exist within the same TM
system.

\paragraph{Client verification.}  Our second contribution (see
\refsec{sec:logic:-logic-release}) is a verification technique for
relaxed-memory client programs that use \wtms. In particular, we prove
correctness of several variations of the message passing litmus test,
synchronised through \wtms transactions, to show that \wtms behaves as
expected. In particular, we show how different client-side guarantees
are achieved depending on the type of synchronisation guarantee
(relaxed, releasing or acquiring) guaranteed by the transaction in
question.

Our verification framework includes a new logic, \logic, capable of
efficiently reasoning about the {\em views} of a client
programs~\cite{DBLP:journals/corr/abs-2004-02983,DBLP:conf/ecoop/KaiserDDLV17}. This
means that the correctness of programs can be established using a
standard Owicki-Gries reasoning
framework~\cite{ECOOP20,DBLP:journals/acta/OwickiG76,DBLP:conf/ppopp/DalvandiD21}.

\paragraph{Implementation, benchmarking and verification.} Our
third contribution is the implementation and full verification of an
STM algorithm that uses C11 relaxed/release-acquire atomics and
implements \wtms. Our implementation is an adaptation of Dalessandro
\etal's {\em Transactional Mutex Lock
  (TML)}~\cite{DBLP:conf/europar/DalessandroDSSS10}, which presents a
simple mechanism for synchronising transactions optimised for
read-heavy workloads. TML is synchronised using a single global lock,
and allows multiple concurrent read-only transactions, but at most one
writing transaction, \ie a writing transaction causes all other
concurrent transactions to abort.

Interestingly, our adapted algorithm, which we call \tmlra, allows
more concurrency than TML by exploiting the parallelism afforded by
relaxed and release-acquire C11 atomics. Moreover, a writing
transaction does not force other read-only transactions to abort,
allowing greater read/write parallelism (see \refsec{sec:tml-c11}). We
show that this theoretical speedup manifests in real implementations
and \tmlra outperforms its SC counterpart in all STAMP benchmarks (see
\refsec{sec:benchmarking}).

We use a simulation-based verification method for the C11 memory
model~\cite{DBLP:conf/ppopp/DalvandiD21} to prove correctness of
\tmlra. This proof establishes a refinement between \tmlra and \wtms,
which ensures that all observable behaviours of \tmlra are observable
behaviours of \wtms.  
Thus, if a client program $C$ is proved correct when it uses \wtms,
then $C$ will also be correct if we replace calls to \wtms in $C$ by
calls to \tmlra.

\paragraph{Mechanisation.} Our fourth contribution is the
mechanisation of all proofs presented in the paper in the Isabelle/HOL
proof assistant (available as supplementary material). This includes
the operational semantics of C11 integrated with \wtms, soundness of
all \logic rules, the use of \logic to prove several client programs
that use \wtms, and finally the proof of simulation between \wtms and
\tmlra.\footnote{Our development may be found in~\cite{Artifact}.}

\paragraph{Overview.} This paper is structured as follows. We describe
our requirements for relaxed and release-acquire transactions in
\refsec{sec:abstr-object-semant}. We formalise this semantics in
\refsec{sec:release-acquire-tm} via the \wtms specification, and
describe its integration with a view-based semantics for RC11 with
release-acquire atomics~\cite{ECOOP20}. In \refsec{sec:tml-c11}, we
provide an examplar implementation and benchmarking results for
\tmlra. In \refsec{sec:logic:-logic-release}, we present our logic for
reasoning about release-acquire transactional memory, which provides a
method of reasoning about client programs that use the \wtms
specification. Finally, in \refsec{sec:cont-refin-b}, we present a
proof of correctness of \tmlra via refinement w.r.t. \wtms.

\section{Transactional guarantees in C11}
\label{sec:abstr-object-semant}

\newcommand{\sflag}{{\it sflag}}
\newcommand{\syncC}{{\it syncC}}
\renewcommand{\view}{{\it view}}
\newcommand{\highlight}[1]{\colorbox{blue!15}{#1}}
\newcommand{\raext}[1]{\colorbox{black!8}{#1}}
\newcommand{\TS}{{\it TS}}
\newcommand{\TID}{{\it TId}}
\renewcommand{\Var}{{\it Loc}}
\newcommand{\seq}{{\it seq}}
\setlength{\fboxsep}{1pt}




A TM specification in a relaxed memory setting has two distinct sets of
goals. The first set must guarantee the expected behaviours of
transactions, \eg serializability, opacity etc. The second must
provide client-side guarantees, \eg release-acquire synchronisation,
observational refinement etc. We consider both in our \wtms
specification (see \reffig{fig:wtms}).

\subsection{Release-acquire synchronisation}
\label{sec:release-acquire-c11}

\begin{figure}[t]
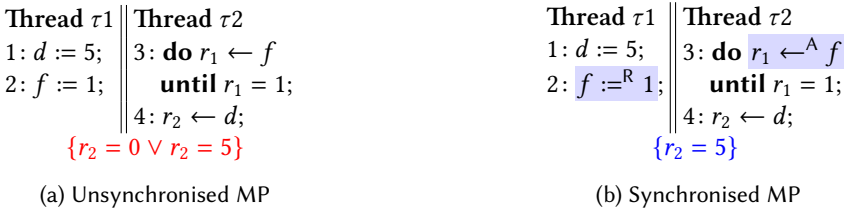

  \begin{minipage}[b]{0.48\columnwidth}
      \begin{center} 
  $\begin{array}{@{}l@{~}||@{~}l}
     {\bf Thread }\ \tidn{1}
     & {\bf Thread }\ \tidn{2}\\
     1\!: d := 5;  & 
                       3\!:  \kwdo\ r_1 \gets f \\
          
     2\!: f := 1; & \ \ \ \ \kwuntil\ r_1 = 1;  \\ 
     & 4\!:  r_2 \gets d; \\
     \end{array}$

   {\color{red} $\{r_2 = 0 \lor r_2=5\}$} 
 \end{center}
 \subcaption[caption]{Unsynchronised MP}
 \label{fig:po-message-bad}
  \end{minipage}
  \hfill
  \begin{minipage}[b]{0.48\columnwidth}
  \begin{center} 
   $\begin{array}{l@{~}||@{~}l}
      \begin{array}[t]{@{}l@{}}
        {\bf Thread }\ \tidn{1}\\
        1\!: d := 5; \\
        2\!:\highlight{$f :=^{\sf R} 1$};\\
      \end{array}
      & 
      \begin{array}[t]{@{}l@{}}
        {\bf Thread }\ \tidn{2}\\
        3\!: \kwdo\ \highlight{$r_1 \gets^{\sf A} f$} \ \\
        \ \ \ \ \kwuntil\ r_1 = 1;  \\ 
        4\!: r_2 \gets d; \\
      \end{array}
   \end{array}$
   
   {\color{blue} $\{ r_2=5\}$}  \qquad \qquad       
 \end{center}
 \subcaption[caption]{Synchronised MP}
 \label{fig:po-message}

\end{minipage}


\vspace{-2pt}
\caption[caption]{Message passing (MP) in C11 
}
\Description{Message passing (MP) in C11}

\end{figure}

Prior to detailing the design choices of \wtms, we recap the basics of
release-acquire synchronisation in C11, including a recently developed
timestamp-based operational semantics, which is the semantics assumed
by \wtms.

The fragment of C11 we focus on is the RC11-RAR fragment. The first
``R'' denotes the \emph{repairing
  model}~\cite{DBLP:conf/pldi/LahavVKHD17}, which precludes `thin-air'
behaviour by disallowing memory operations within a thread to be
reordered. The ``RAR'' refers to the fact that the model includes
\emph{release-acquire} as well as \emph{relaxed}
atomics~\cite{DBLP:journals/darts/DalvandiDDW20,DBLP:conf/ppopp/DohertyDWD19}.
\footnote{Note that extending this model to include other types of C11
  synchronisation (\eg SC fences) and relaxations that allow
  intra-thread ordering is
  possible~\cite{DBLP:journals/corr/abs-2108-01418}, but these
  extended models are not so interesting for the purposes of this
  paper, and the additional complexity that they induce detracts from
  our main contributions. } For the remainder of this paper, we simply
write C11 to refer to RC11-RAR.

We explain the main ideas behind release-acquire synchronisation using
the message passing (MP) litmus test in Figs.~\ref{fig:po-message-bad}
and \ref{fig:po-message}. It comprises two shared variables: $d$ (for
data) and $f$ (for a flag), both of which are initially $0$.  Under
SC, the postcondition of the program is $r_2 = 5$ because the loop in
thread~$\tidn{2}$ only terminates after $f$ has been updated to $1$ in
thread $\tidn{1}$, which in turn happens after $d$ is set to
$5$. Therefore, the only possible value of $d$ that thread $\tidn{2}$
can read is $5$.

However, in \reffig{fig:po-message-bad}, all read/write accesses of
$d$ and $f$ are {\em relaxed}, and hence the program can only
establish the weaker postcondition $r_2 = 0 \lor r_2 = 5$ since it is
possible for thread~$\tidn{2}$ to read $0$ for {\it d} at line 4. In
particular, reading $1$ for $f$ does not guarantee that
thread~$\tidn{2}$ will read $5$ for $d$.

This anomaly is corrected in \reffig{fig:po-message} where the
\highlight{highlighted} code depicts the necessary changes. In
particular, we introduce a {\em release} annotation (line 2) as well
as an {\em acquire} annotation (line 3), which together induces a
\emph{happens-before} relation if the read of $f$ at line 3 reads from
the write at line~2 (see \cite{DBLP:conf/popl/BattyOSSW11}). This in
turn ensures that thread~$\tidn{2}$ sees the most recent write to $d$
at line~1. We explain how relaxed accesses and release-acquire
synchronisation is formalised by the operational semantics in
\refsec{sec:views}.

\subsection{Transactional message passing}
\label{sec:wtms}

\begin{figure*}[t]
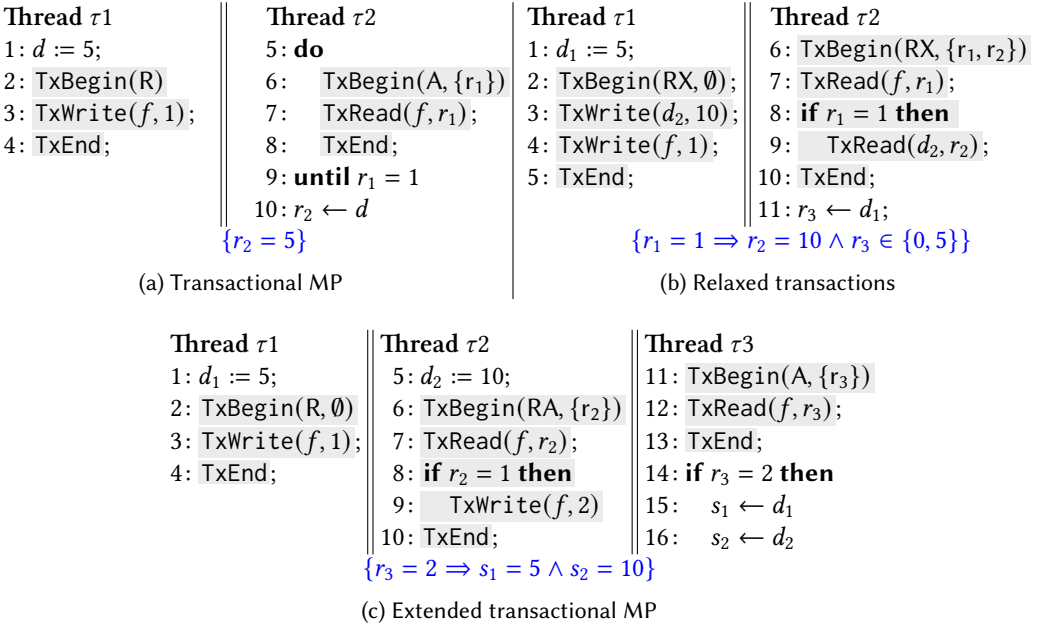

  \begin{minipage}[b]{0.48\columnwidth}
   \scalebox{1.0}{$\begin{array}{l@{\quad}||@{\quad}l}
      \begin{array}[t]{@{}l@{}}
        {\bf Thread }\ \tidn{1}\\
        1\!: d := 5; \\
        2\!:\raext{${\tt TxBegin}(\sf R)$} \\
        3\!:\raext{${\tt TxWrite}(f, 1)$};\\
        4\!:\raext{${\tt TxEnd}$};\\
      \end{array}
      & 
      \begin{array}[t]{@{}r@{~}l@{}}
        \multicolumn{2}{l}{{\bf Thread }\ \tidn{2}}\\
        5\!:& \kwdo \\
        6\!:& \quad \raext{${\tt TxBegin}(\sf A, \{r_1\})$} \\
        7\!:& \quad \raext{${\tt TxRead}(f, r_1)$};\\
        8\!:& \quad \raext{${\tt TxEnd}$};\\
        9\!:& \kwuntil\ r_1 = 1 \\
        10\!:& r_2 \gets d\\
      \end{array}
   \end{array}$}
   
 \hfill {\color{blue} \qquad $\{ r_2=5\}$} \hfill{}
 \subcaption[caption]{Transactional MP 
 }
 \label{fig:mp-tm}
\end{minipage}
\hfill \vline \hfill
\begin{minipage}[b]{0.48\columnwidth}
   \scalebox{1.0}{$\begin{array}{@{}l@{~~~~}||@{~~~~}l@{}}
      \begin{array}[t]{@{}l@{}}
        {\bf Thread }\ \tidn{1}\\
        1\!: d_1 := 5; \\
        2\!:\raext{${\tt TxBegin}(\sf RX,\emptyset)$}; \\
        3\!:\raext{${\tt TxWrite}(d_2,10)$}; \\
        4\!:\raext{${\tt TxWrite}(f, 1)$};\\
        5\!:\raext{${\tt TxEnd}$};\\
      \end{array}
      & 
      \begin{array}[t]{@{}r@{~}l@{}}
        \multicolumn{2}{l}{{\bf Thread }\ \tidn{2}}\\
        6\!:& \raext{${\tt TxBegin}(\sf RX, \{r_1, r_2\})$} \\
        7\!:& \raext{${\tt TxRead}(f, r_1)$};\\
        8\!:& \raext{\kwif\ $r_1 = 1$\ \kwthen\ }\\
        9\!:& \raext{\quad${\tt TxRead}(d_2, r_2)$};\\
        10\!:& \raext{${\tt TxEnd}$};\\
        11\!:& r_3\gets d_1; 
      \end{array}
   \end{array}$}
   
 \hfill {\color{blue} \qquad $\{r_1 = 1 \imp r_2=10 \wedge r_3 \in \{0,5\}\}$} \hfill{}
 \subcaption[caption]{Relaxed transactions} 
 \label{fig:mp-tm-2}
\end{minipage}
\bigskip

\smallskip

\hfill
 \begin{minipage}[b]{0.65\columnwidth}
   \scalebox{1.0}{$\begin{array}{@{}l@{~~~~}||@{~~~~}l@{~~~~}||@{~~~~}l}
      \begin{array}[t]{@{}l@{}}
        {\bf Thread }\ \tidn{1}\\
        1\!: d_1 := 5; \\
        2\!:\raext{${\tt TxBegin}(\sf R, \emptyset)$} \\
        3\!:\raext{${\tt TxWrite}(f, 1)$};\\
        4\!:\raext{${\tt TxEnd}$};\\
      \end{array}
      & 
      \begin{array}[t]{@{}r@{~}l@{}}
        \multicolumn{2}{@{}l}{{\bf Thread }\ \tidn{2}}\\
        5\!:& d_2 := 10; \\
        6\!:& \raext{${\tt TxBegin}(\sf RA, \{r_2\})$} \\
        7\!: & \raext{${\tt TxRead}(f, r_2)$};\\
        8\!:& \raext{$\kwif\ r_2 = 1\ \kwthen$} \\
        9\!:& \raext{$\quad {\tt TxWrite}(f, 2)$} \\
        10\!:&\raext{${\tt TxEnd}$};\\
      \end{array}
      & 
      \begin{array}[t]{@{}r@{~}l@{}}
        \multicolumn{2}{@{}l}{{\bf Thread }\ \tidn{3}}\\
        11\!:& \raext{${\tt TxBegin}(\sf A, \{r_3\})$} \\
        12\!: & \raext{${\tt TxRead}(f, r_3)$};\\
        13\!:& \raext{${\tt TxEnd}$};\\
        14\!:& \kwif\ r_3 = 2\  \kwthen \\
        15\!:& \ \ \ \ s_1 \gets d_1\\
        16\!:& \ \  \ \ s_2 \gets d_2\\
      \end{array}
   \end{array}$}
   
 \hfill {\color{blue} $\{r_3 = 2 \imp s_1=5 \wedge s_2=10\}$}\hfill {}

 \subcaption[caption]{Extended transactional
   MP 
 }
 \label{fig:ra-trans-mem} 
\end{minipage}
\hfill{}


\caption{Transactional memory client interactions}
\Description{Transactional memory client interactions}
\label{fig:trans-client}
\end{figure*}

We now describe the guarantees provided by our transactional model in
the context of a client program. Like standard reads and writes in
C11, we allow transactions to be combined with a synchronising
annotation, which may be one of relaxed ({\sf RX}), releasing ({\sf
  R}), acquiring ({\sf A}), or release-acquiring ({\sf RA}). These
annotations dictate whether or not transaction induces a client-side
happens before.  In particular, client-side happens-before is induced
from thread $\tidn{1}$ to thread $\tidn{2}$ if
\begin{enumerate*}[label=(\roman*)]
\item a read in transaction $t_2$ executed by $\tidn{2}$ reads-from a write
  in transaction $t_1$ executed by $\tidn{1}$,
\item $t_1$ contains a release annotation (either {\sf R} or {\sf RA}), and 
\item $t_2$ contains an acquire annotation (either {\sf A} or {\sf
    RA}).
\end{enumerate*}
We illustrate the implications of these annotations via the examples
in Fig.~\ref{fig:trans-client}, where the \raext{highlights} are used
to identify the transactions. We assume that a client provides a
transaction with a set of registers that it may use when the
transaction begins (see \refsec{sec:release-acquire-tm}).

\reffig{fig:mp-tm} describes a transactional variation of MP. Thread
$\tidn{1}$ comprises a (non-transactional) relaxed write on $d$ followed by
a transactional write of the flag, $f$. Thread $\tidn{2}$ contains a
transactional read of $f$ within a loop that terminates if $\tidn{2}$ reads
$1$ for $f$. After the loop terminates, $\tidn{2}$ performs a
(non-transactional) relaxed read of $d$. In this example, like
\reffig{fig:po-message}, the release and acquire annotations induce a
happens-before relation from $\tidn{1}$ to $\tidn{2}$ and hence ensure that the
read of $d$ in $\tidn{2}$ does not return the stale value, $0$.

\reffig{fig:mp-tm-2} describes a program that uses a relaxed
transaction. The postcondition of the program considers the case where
the transaction in $\tidn{1}$ occurs before the transaction in
$\tidn{2}$ since the antecedent assumes that $r_1 = 1$, \ie the read
of $f$ at line~7 reads the write of $f$ at line~4. In this example,
both transactions are relaxed, and hence, the ordering of transactions
above does not induce a happens before from $\tidn{1}$ to
$\tidn{2}$. Thus, the read of $d_1$ at line~11 is not guaranteed to
see the write of $d_1$ at line~1, \ie the final value of $r_3$ is
either $0$ or $5$. However, since the write and read of $d_2$ occurs
within the transactions of $\tidn{1}$ and $\tidn{2}$, respectively, if
$r_1=1$, then $\tidn{2}$ is guaranteed to read $10$ for $d_2$.

Finally, \reffig{fig:ra-trans-mem} demonstrates a program with an {\sf
  RA} transaction. The antecedent of the program's postcondition
implies that the transaction in $\tidn{3}$ occurs after the
transaction in $\tidn{2}$, which in turn occurs after the transaction
in $\tidn{1}$. Here, the transaction annotations ensure that the
writes to $d_1$ and $d_2$ (at lines 1 and 5) performed by the {\em
  client} are seen by the client reads at lines 15 and 16. This is
because the transaction in $\tidn{2}$ (annotated by {\sf RA}) is
guaranteed to synchronise with the transaction in $\tidn{1}$
(annotated by {\sf R}), and similarly, the transaction in $\tidn{2}$
(annotated by {\sf A}) is guaranteed to synchronise with the
transaction in $\tidn{2}$ (annotated by {\sf RA}). Note that if the
transaction in $\tidn{2}$ was only releasing, then $\tidn{1}$ and
$\tidn{2}$ would not synchronise, and the read at line 15 may return
either $0$ or~$5$. Yet, the read at line 16 would still be guaranteed
to return $10$ for $d_2$ since $\tidn{2}$ and $\tidn{3}$
synchronise. If the transaction in $\tidn{2}$ was only acquiring, then
$\tidn{2}$ and $\tidn{3}$ would not synchronise. In this case,
although $\tidn{1}$ and $\tidn{2}$ have synchronised, neither of the
reads at lines~15~and~16 are guaranteed to return the new writes at
lines~1 and~5.

Deciding a transaction's synchronisation flag ultimately comes down to
the needs of a client program, much like {\tt memory\_order}
parameters on {\tt atomic\_compare\_exchange} instructions in
C11~\cite{CMP-EX}. Client programs that require message passing
through transactions would use release-acquire, while others may only
require relaxed annotations.


\newcommand{\OV}{{\it OV}}
\newcommand{\regset}{{\it regSet}}

\section{Release-acquire TM specification}
\label{sec:release-acquire-tm}

With the basic requirements for release-acquire and transactional
synchronisation in place, we work towards a formal TM
specification. Our specification will be closely tied to an
operational semantics for C11 with timestamped writes and per-thread
views~\cite{DBLP:journals/corr/PodkopaevSN16,Dolan:2018:LDRF,DBLP:conf/ecoop/KaiserDDLV17,DBLP:conf/popl/KangHLVD17,ECOOP20}
(see \refsec{sec:views}). We integrate this model with a TM
specification in \refsec{sec:wtms-1}.


\subsection{View-based operational semantics}
\label{sec:views}
As discussed above, in our model, the C11 relaxed memory state is
formalised by \emph{timestamped writes}. Instead of mapping each
location to a value, 
the state contains a set of writes
$\writes \subseteq \Write$,
where $\Write=\Var \times \Val \times \TS$ represents a write to a
location $\Var$ with value $\Val$ and $TS \eqdef \mathbb{Q}$ is the
set of possible timestamps. If $w \in \Write$ and
$w = (\loc, \val, \ts)$, then we let $\xxloc(w) \eqdef \loc$,
$\xxval(w) \eqdef \val$, $\xxtst(w) \eqdef \ts$, be the functions that
extract the location, value and timestamp of $w$, respectively.

A {\em view} is a mapping from a location to a write of that location,
\ie $\View \eqdef \Var \to \Write$.
To define the allowable reads by each thread to each location, the
state also records a {\em thread view} for each thread define by a
function
\[
  \tview: \TID \to \View
\]
where $\TID \eqdef \nat$ is the set of thread identifiers.  A thread
may read from \emph{any} write whose timestamp is no smaller than the
thread's current view. Thus, the {\em observable values} ($\OV$), \ie
the set of values that thread $\tid$ can read for location $\loc$ is
  \begin{align*}
    \OW_\tid(x) & \eqdef \left\{\wr \in \writes 
                  \begin{array}{@{~}|l@{}}
                    \xxloc(\wr) = x \wedge {} 
                    \xxtst(\wr) \geq \xxtst(\tview_\tid(\loc))
                  \end{array}\right\}
    \\
    \OV_\tid(x) & \eqdef \{\xxval(\wr) \mid \wr \in \OW_\tid(x)\}
  \end{align*}
  A write may be introduced at any timestamp greater than the thread's
  current view (with a caveat that ensures atomicity of
  read-modify-writes, see \cite{DBLP:conf/ppopp/DohertyDWD19,ECOOP20}
  for details).

Finally, to formalise release-acquire synchronisation, a state in the
timestamp model also includes a notion of a {\em modification view},
\[
  \mview: \Write \to \View
\]
which is a function that records the thread view of the executing
thread when a new write is introduced to $\writes$. In particular, if
thread $\tid$ introduces a new write $\wr$ to $\writes$ and
$\tview_\tid$ is updated to $view$ in this new state, then $\mview$ is
also updated so that $\mview_\wr = view$ in the new state. This
information is used to update thread views in case release-acquire
synchronisation occurs.

Formally, when threads synchronise, a new view is calculated using
$\otimes$, which is defined as follows. Given $V_1, V_2 \in \View$, we
have
\[
V_1 \otimes V_2  \eqdef \lambda \loc 
.\ 
\begin{array}[t]{@{}l@{}}
  \textbf{if}~\xxtst(V_2(\loc)) \leq \xxtst(V_1(\loc))\ 
  \textbf{then}\ V_1(\loc)\ \textbf{else}\  V_2(\loc)
\end{array}
\]
which constructs a new view by taking the write with the larger
timestamp for each location $\loc$.

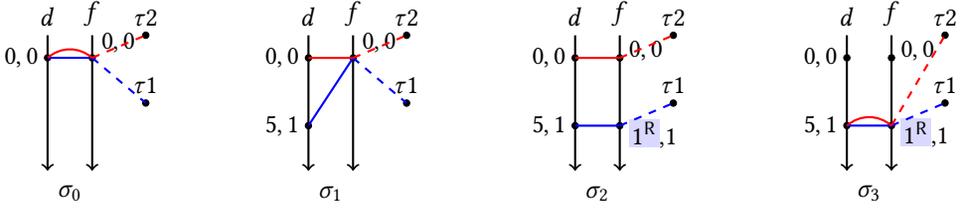
\begin{figure}[t]
  \begin{minipage}[t]{0.48\columnwidth}
  \begin{minipage}[b]{0.48\columnwidth}
    \begin{center}
      \scalebox{0.99}{
    \begin{tikzpicture}[scale=0.6]
      \draw[thick,->] (0,4) -- (0,1); \draw[thick,->] (1,4) --
      (1,1);
    
      \coordinate (d0) at (0,3.5); \coordinate (d5) at (0,2);
      \coordinate (f0) at (1,3.5); \coordinate (f1) at (1,1.5);
      \coordinate (sigma) at (0.5,0.5);
      \node at (sigma) {$\sigma_0$} ;

      \coordinate (T2-1) at (2.2,4);
      \coordinate (T1-1) at (2.2,2.5);
            
      \draw (0,4) node[above] {$d$}; \draw (1,4) node[above] {$f$};
      \draw (T2-1) node[above] {$\tidn{2}$};
      \draw (T1-1) node[above] {$\tidn{1}$};

      \filldraw [black] (d0) circle (2pt) node[left, black]
      {$0, 0$};
     
      \filldraw [black] (f0) circle (2pt) node[right, black,yshift=2mm]
      {$0, 0$};

      \filldraw [black] (T1-1) circle (2pt);
    
      \filldraw [black] (T2-1) circle (2pt);

      \path [draw, blue, thick] (d0) edge 
      (f0) (f0) edge [dashed] (T1-1);
      \path [draw, thick, red] (d0) edge
      [bend left=40] (f0) (f0) edge [dashed] (T2-1);

      \end{tikzpicture}}
  \end{center}
    \end{minipage}       
    \hfill
  \begin{minipage}[b]{0.48\columnwidth} 
    \begin{center}
    \scalebox{0.99}{
    \begin{tikzpicture}[scale=0.6]
      \draw[thick,->] (0,4) -- (0,1); \draw[thick,->] (1,4) --
      (1,1);
    
      \coordinate (d0) at (0,3.5); \coordinate (d5) at (0,2);
      \coordinate (f0) at (1,3.5); \coordinate (f1) at (1,1.5);
      \coordinate (sigma) at (0.5,0.5);
      \node at (sigma) {$\sigma_1$} ;

      \coordinate (T2-1) at (2.2,4);
      \coordinate (T1-1) at (2.2,2.5);
            
      \draw (0,4) node[above] {$d$}; \draw (1,4) node[above] {$f$};
      \draw (T2-1) node[above] {$\tidn{2}$};
      \draw (T1-1) node[above] {$\tidn{1}$};

      \filldraw [black] (d0) circle (2pt) node[left, black]
      {$0, 0$};
     
      \filldraw [black] (f0) circle (2pt) node[right, black,yshift=2mm]
      {$0, 0$};

      \filldraw [black] (d5) circle (2pt) node[left,
        black] {$5,1$};

      \filldraw [black] (T1-1) circle (2pt);
    
      \filldraw [black] (T2-1) circle (2pt);

      \path [draw, blue, thick] (d5) -- (f0) (f0) edge
        [dashed] (T1-1);
      \path [draw, thick, red] (d0) -- (f0) (f0) edge
        [dashed] (T2-1);

      \end{tikzpicture}}
  \end{center}
    \end{minipage}
  \end{minipage}
  \hfill
  \begin{minipage}[t]{0.5\columnwidth}
      \begin{minipage}[b]{0.48\columnwidth}
        \begin{center}
          \scalebox{0.99}{
    \begin{tikzpicture}[scale=0.6]
      \draw[thick,->] (0,4) -- (0,1); \draw[thick,->] (1,4) --
      (1,1);

      \coordinate (d0) at (0,3.5); \coordinate (d5) at (0,2);
      \coordinate (f0) at (1,3.5); \coordinate (f1) at (1,2);

      \coordinate (sigma) at (0.5,0.5);
      \node at (sigma) {$\sigma_2$} ;
      
      \coordinate (T2-1) at (2.2,4);
      \coordinate (T1-1) at (2.2,2.5);

      \draw (0,4) node[above] {$d$}; \draw (1,4) node[above] {$f$};
      \draw (T2-1) node[above] {$\tidn{2}$};
      \draw (T1-1) node[above] {$\tidn{1}$};

      \filldraw [black] (d0) circle (2pt) node[left, black]
      {$0,0$};
     
      \filldraw [black] (f0) circle (2pt) node[right, black,yshift=1mm]
      {$0,0$};

      \filldraw [black] (d5) circle (2pt) node[left,
        black] {$5,1$}; 

       \filldraw [black] (f1) circle (2pt) node[right,
        black,yshift=-1mm] {\highlight{$1^{\sf R}$},1}; 

      \filldraw [black] (T1-1) circle (2pt);
    
      \filldraw [black] (T2-1) circle (2pt);

      \path [draw, blue, thick] (d5) -- (f1) (f1) edge
        [dashed] (T1-1);
      \path [draw, thick, red] (d0) -- (f0) (f0) edge
        [dashed] (T2-1);
      \end{tikzpicture}}
      \end{center}
    \end{minipage}
    \hfill
      \begin{minipage}[b]{0.48\columnwidth}
        \begin{center}
          \scalebox{0.99}{
    \begin{tikzpicture}[scale=0.6]
      \draw[thick,->] (0,4) -- (0,1); \draw[thick,->] (1,4) --
      (1,1);
    
      \coordinate (d0) at (0,3.5); \coordinate (d5) at (0,2);
      \coordinate (f0) at (1,3.5); \coordinate (f1) at (1,2);
      \coordinate (sigma) at (0.5,0.5);
      \node at (sigma) {$\sigma_3$} ;

      \coordinate (T2-1) at (2.2,4);
      \coordinate (T1-1) at (2.2,2.5);
            
      \draw (0,4) node[above] {$d$}; \draw (1,4) node[above] {$f$};
      \draw (T2-1) node[above] {$\tidn{2}$};
      \draw (T1-1) node[above] {$\tidn{1}$};

      \filldraw [black] (d0) circle (2pt) node[left, black]
      {$0,0$};
     
      \filldraw [black] (f0) circle (2pt) node[right, black,yshift=1mm]
      {$0,0$};

      \filldraw [black] (d5) circle (2pt) node[left,
        black] {$5,1$}; 

       \filldraw [black] (f1) circle (2pt) node[right,
        black,yshift=-1mm] {\highlight{$1^{\sf R}$},1}; 

      \filldraw [black] (T1-1) circle (2pt);
    
      \filldraw [black] (T2-1) circle (2pt);

       \path [draw, blue, thick] (d5) edge 
        (f1) (f1) edge [dashed] (T1-1);

\path [draw, thick, red] (d5) edge [bend left=40] (f1) (f1) edge          
       [dashed] (T2-1);

     \end{tikzpicture}}
   \end{center}
   
 \end{minipage}
\end{minipage}

  \caption{Synchronised message passing views}
  \Description{Synchronised message passing views}
 \label{fig:po-message-sync-exec}
\end{figure}

\begin{example}[Synchronised MP]\label{ex:sync-mp}
  Consider \reffig{fig:po-message-sync-exec}, which depicts a possible
  execution of the program in \reffig{fig:po-message}. Each ``$v,i$''
  represents a ``value, timestamp'' pair for the location in
  question. The initial state is $\sigma_0$, where the views of
  threads $\tidn{1}$ and $\tidn{2}$ are both the initial writes. State
  $\sigma_1$ occurs after executing line~1, where the view of
  $\tidn{1}$ is updated to the new write on $d$. Similarly, $\sigma_2$
  occurs after executing line~2. Note that the new write is tagged
  with a release annotation. Moreover, the operational semantics
  guarantees that in $\sigma_2$, we have
  $\sigma_2.\mview_{(f, 1,1)}(d) = (d, 5, 1)$, \ie the modification
  view of the write $(f, 1, 1)$ returns $(d, 5, 1)$ for $d$ (since
  this was the thread view of $\tidn{1}$ for $d$ when the write at
  line 2 occurred).

  Finally, $\sigma_3$ depicts the state after execution of line~$3$,
  where the read returns the value $1$ for $f$. In this case, the
  thread view of $\tidn{2}$ for $f$ is updated to the new read. More
  importantly, due to release-acquire annotations the semantics
  enforces that the thread view of $\tidn{2}$ for $d$ in $\sigma_3$ is
  {\em also} updated to the new modification view, \ie
  $\sigma_2.\mview_{(f, 1, 1)}(d)$. Thus, after state $\sigma_3$,
  $\tidn{2}$ will no longer be able to return the stale value $0$ for
  $d$.
\end{example}
The key difference in execution of the unsynchronised example
(\reffig{fig:po-message-bad}) is that the read at line 3 does not
update $\tview_{\tidn{2}}(d)$. Hence, for the state of
\reffig{fig:po-message-bad} analogous to $\sigma_3$, the view of
$\tidn{2}$ for $d$ will remain at the initial write, allowing it to
return a stale value.


\subsection{TMS2}
\label{sec:tms2}

First, we consider the TMS2 specification, which is our \wtms
specification {\em without} any client-side release-acquire
guarantees. This is given by the unhighlighted components of
\reffig{fig:wtms}, which correspond precisely to the {\em internal
  actions} of
TMS2~\cite{DBLP:journals/fac/DohertyGLM13}.\footnote{\wtms, like TMS2
  is presented as an I/O automaton~\cite{DBLP:books/mk/Lynch96}. For
  simplicity, we eschew the external actions, but they can easily
  be included to formalise the TM interface. 
} Note that each action of \reffig{fig:wtms} is atomic and guarded by
the conditions defined in \pre. If all the conditions in \pre\ hold
the transition is {\em enabled}, and the corresponding action
atomically updates the state according to the assignments and
functions in \eff. If some condition in \pre\ does not hold then the
transition is {\em blocked}. We use $\sqcap$ to denote a
non-deterministic choice (see \cite{DBLP:books/mk/Lynch96} for
details).

TMS2 is a close operational approximation of {\em
  opacity}~\cite{2010Guerraoui}. The differences between TMS2 and
opacity are minor~\cite{LLM12}, and much of the discussion below
applies equally to opacity. 
TMS2 (and opacity) distinguishes between {\em completed} and {\em
  live} transactions, where completed transaction may either be {\em
  committed} or {\em aborted}. TMS2 guarantees the existence of a
total order $\prec$ over {\em all} transactions such that:
\begin{enumerate}[label=(\arabic*)]
\item if transaction $t_1$ executes {\tt TxEnd} before $t_2$ executes
  {\tt TxBegin}, then $t_1 \prec t_2$;
\item for any transaction $t$, if $\prec_{\downarrow t}$ is the strict
  downclosure of~$t$ w.r.t. $\prec$ and $m$ is the memory obtained by
  applying the committed transactions in $\prec_{\downarrow t}$ in
  order, then
  \begin{itemize}
  \item all internal reads in $t$ for a variable $x$ are consistent
    with the last write to $x$ in $t$, and
  \item all external reads of $t$ are consistent with $m$.
\end{itemize}
\end{enumerate}
Note that conditions (1) and (2) together imply strict
serialisability of the transactions. Condition (2) additionally
ensures that no transaction reads from an aborted or live transaction
since all external writes can be explained by the prior writes of
committed transactions only. Moreover, reads of all transactions
(including aborted and live transactions) never return a spurious
value, \ie each non-aborting read can be explained by prior committed
transactions.

The existence of the total order mentioned above is guaranteed by the
TMS2 specification as follows. Each transaction $t$ comprises a local
read set, {\tt rdSet}$_t$, local write set, {\tt wrSet}$_t$, and
variable, {\tt status}$_t$ that is used to model control flow within a
transaction. If the status of $t$ is {\tt NOTSTARTED}, $t$ may
transition to status {\tt READY} if a thread $\tid$ executes {\tt
  TxBegin}$_\tid$. Once ready, $\tid$ may execute some number of {\tt
  TxRead} and {\tt TxWrite} operations, or {\tt TxEnd}, which sets the
status of $t$ (the transaction that $\tid$ is executing) to {\tt
  COMMITTED}. Note that if transaction $t$ is {\tt READY}, it may transition to
status {\tt ABORT} at any time. Moreover, in some circumstances, $t$
may be forced to abort because all other transitions of $t$ are
blocked. For example, if $t$ is a writing transaction and $t$'s read
set is not a subset of $\last(M)$, then $t$ must abort.

To ensure read/write consistency, TMS2 uses a sequence of memories
$M$, where a memory is a mapping from locations to values.  A
transaction $t$ records the earliest memory it can read from by
setting {\tt beginIdx}$_t$ to the last index of $M$ when $t$ executed
{\tt TxBegin}$_{\tid}$. Moreover, each committing writing transaction $t$
constructs a new memory $N = \last(M) \oplus{}${\tt wrSet}$_t$ which
is the memory $\last(M)$ overwritten with the write set of $t$. It
then appends $N$ to the end of $M$ (see {\tt TxEndWR}).

We differentiate between internal reads {\tt TxReadInt} and external
reads {\tt TxReadExt}, by whether the read location $x$ is in the
write set of the executing transaction, $t$. An internal read of $x$
simply returns the value of $x$ in the write set of $t$. An external
read of $x$ non-deterministically picks a memory index $i$. This read
is enabled iff $i$ is a valid index (\ie is between {\tt beginIdx}$_t$
and the last memory index, $|M|-1$) and the read set of $t$ is
consistent with $M_i$ (\ie the memory at index $i$). In case an
external read occurs, the read set is updated and the value read is
returned. This means that all external reads in $t$ are validated with
respect to {\em some} memory snapshot between {\tt beginIdx}$_t$ and
the maximum memory index. Note that it is possible for two different
reads to validate w.r.t. {\em different} memory snapshots.

TMS2 prescribes a {\em lazy} write-back strategy via {\tt TxWrite},
where writes are cached in a local write set until the commit occurs
(as described above). However, as we shall see, this does not preclude
implementations that use {\em eager} write-backs, where writes occur
in memory at the time of writing (see~\cite{DBLP:journals/fac/DerrickDDSTW18}). In fact, the
\tmlra algorithm, our main case study in this paper, {\em is} such an
eager algorithm (see \refsec{sec:tml-c11}).

We split the commit phase into two cases: {\em read-only} (modelled by
{\tt TxEndRO}) and {\em writing} (modelled by {\tt TxEndWR}). Since
all reads are validated at the time of reading, a read-only
transaction can simply commit the transaction. On the other hand, the
writing transaction must ensure its reads are valid w.r.t. the last
memory snapshot. The effect of this transition is to install a new
memory snapshot as described above.

The final component of $t$ is a local set {\tt regs}$_{t}$ that is
used to keep track of the set of registers that the transaction has
written to. A client provides the set of registers to be used by each
transaction when the transaction begins. These registers are set to a
special value $\bot$ when a transaction aborts to ensure that no value
read by $t$ is seen outside $t$.

\begin{figure}[t]
  \centering  \small
  {\tt \begin{tabular}[t]{@{}l@{~~~~~~~~~}l@{}}
         \begin{tabular}[t]{@{}l}
           TxBeginV$_{\tid}(\raext{$\sflag$},\raext{$m$}, $\regset$)$\\
           \begin{tabular}[t]{@{}l@{~}l@{}}
             \pre  & status$_{t}$ = NOTSTARTED  \\
             & txn$_\tid$ = $\bot$  \\
             & \raext{$m$ $\in$ $\vismem_{\tau}$} \\
             \eff  & wrSet$_{t}$:=$\emptyset$ \\
                   & rdSet$_{t}$:=$\emptyset$ \\
                   & \raext{beginIdx$_{t}$:= $m$}  \\
                   & \raext{seenIdxs$_{t}$:=$\emptyset$} \\
                   & \raext{synctype$_{t}$:=$\sflag$} \\
                   & regs$_t$:=$\regset$ \\
                   & txn$_\tid$:=$\txid$ \\
                   & status$_{t}$:=READY \\
           \end{tabular}
           \\\\
                      TxWrite$_{\tid}$($x, v$)\\
           \begin{tabular}[t]{@{}l@{~}l@{}}
             \pre  & status$_{t}$ = READY  \\
                   & txn$_\tid$ = $\txid$  \\
             \eff  & wrSet$_{t}$:=wrSet$_{t}$$\cup \{x \mapsto v\}$
           \end{tabular}
           \\\\
           TxReadInt$_{\tid}$($x, r$)\\
           \begin{tabular}[t]{@{}l@{~}l@{}}
             \pre  & status$_{t}$ = READY  \\
                   & $x \mapsto v \in$ wrSet$_{t}$ \\
                   & $r \in{}$regs$_t$ \\
                   & txn$_\tid$ = $\txid$  \\
             \eff  & $r$:=$\val$ 
           \end{tabular}
\\
         \\
         \begin{tabular}[t]{@{}l}
           TxReadExt$_{\tid}$($x, i, r$)\\
           \begin{tabular}[t]{@{}l@{~}l@{}}
             \pre  & status$_{t}$ = READY  \\
                   & $x \notin \dom$(wrSet$_{t}$) \\
                   & beginIdx$_t \le  i < |M|$ \\
                   & rdSet$_t \subseteq M_i$ \\
                   & txn$_\tid$ = $\txid$  \\
             \eff  & rdSet$_t$:=rdSet$_t \cup \{x \mapsto M_i(x)\}$\\
                   & \raext{\kwif\ $S_i \in \{{\sf R}, {\sf RA}\}$}\\
                   & \raext{\kwthen\ seenIdxs$_t$:= seenIdxs$_t\cup \{i\}$} \\
                   & $r$:=$M_i(x)$ 
           \end{tabular}
         \end{tabular}
           \\
           \\
           TxRead$_{\tid}(x, r)$ = \\
           \quad \begin{tabular}[t]{@{}l@{}}
               TxReadInt$_{\tid}(x) \sqcap{}$ $\bigsqcap_{i}$TxReadExt$_{\tid}(x, i, r)$
             \end{tabular}

         \end{tabular}
         &
         \begin{tabular}[t]{@{}l}
           TxEndRO$_{\tid}$\\
           \begin{tabular}[t]{@{}l@{~}l@{}}
             \pre  & status$_{t}$ = READY  \\
                   & wrSet$_t=\emptyset$ \\
                   & txn$_\tid$ = $\txid$  \\
             \eff  & status$_{t}$:=COMMIT \\
                   & \raext{\kwif\ synctype$_t\in \{{\sf A}, {\sf RA}\} \wedge{}$seenIdxs$_t \neq \emptyset $} \\
                   & \raext{\kwthen}\\
                   & \raext{\ \ \kwlet\ nv${}=\view$(seenIdxs$_t$,V) \kwin} \\
                   & \raext{\ \ $\tview_\tid$:= $\tview_\tid \otimes{}$nv} \\
                   & txn$_{\tid}$:=$\bot$ \\
                  & \raext{$\txview_{\tid}$:= $\max$(seenIdxs$_t$)}
           \end{tabular}
         \\
         \\
           TxEndWR$_{\tid}$\\
           \begin{tabular}[t]{@{}l@{~}l@{}}
             \pre  & status$_{t}$ = READY  \\
                   & wrSet$_t \neq \emptyset$ \\
                   & rdSet$_t \subseteq \last(M)$ \\
                   & txn$_\tid$ = $\txid$  \\
             \eff  & $M$:= $M \cdot (\last(M) \oplus{}$wrSet$_t)$ \\ 
                   & status$_{t}$:=COMMIT \\
                   & \raext{\kwif\ synctype$_t\in \{{\sf A}, {\sf RA}\} \wedge {}$seenIdxs$_t \neq \emptyset $} \\
                   & \raext{\kwthen\ } \\ 
                   & \raext{\ \ \kwlet\ nv${} = \view$(seenIdxs$_t$,V) \kwin} \\
                   & \raext{\ \ \ \ $V$:= $V \cdot{}(\tview_\tid \otimes{}$nv)} \\
                   & \raext{\ \ \ \ $\tview_\tid$:= $\tview_\tid \otimes{}$nv} \\
                   & \raext{\kwelse\ $V$:= $V \cdot{}\tview_\tid$} \\
                   & \raext{$S$:= $S \cdot{}$synctype$_t$} \\
                   & txn$_{\tid}$:=$\bot$ \\
                   & \raext{$\txview_{\tid}$:= $\max$(seenIdxs$_t$)}

           \end{tabular}
         \\
         \\
           Abort$_{\tid}$\\
           \begin{tabular}[t]{@{}l@{~}l@{}}
             \pre  & status$_{t}$ = READY  \\
                   & txn$_\tid$ = $\txid$  \\
             \eff  & $\forall s \in{}$regs$_t$.\,$s$:=$\bot$ \\
                   & txn$_{t}$:=$\bot$  \\
                   & status$_{t}$:=ABORT
           \end{tabular}
           \\
           \\
           TXBegin$_{\tid}(\sflag, \regset)$ = $\bigsqcap_{m}$ TXBeginV$_{\tid}(\sflag, m, \regset) $ \\
             TxEnd$_{\tid}$ = 
             \begin{tabular}[t]{@{}l@{}}
               TxEndRO$_{\tid} \sqcap{}$TxEndWR$_{\tid}$
             \end{tabular}
         \end{tabular}
       \end{tabular}}

\bigskip

\raggedright{
  \begin{tabular}[t]{@{}l@{\qquad}l}
    {\bf where} \\
    $M : \seq (\Var \to \Val)$, initially $M = \langle (\lambda v \in \Var. 0)\rangle $ & 
    \raext{$S : \seq \{\sf RX, R, A, RA\}$, initially $S = \langle \sf RX \rangle$} \\
    \raext{$V : \seq (\Var \to \TS)$, initially $V = \langle (\lambda v \in \Var.0) \rangle$} &
    \raext{$\view({\it Idxs}, {\it Vf}) = \lambda l \in \Var.\ \maxWr \{{\it Vf}_i(l) \mid i \in {\it Idxs} \}  $} \\
    \raext{$\vismem_{\tau} =  \{ n ~|~ n \geq \txview_{\tau} \land n \leq |M|-1 \}  $}
  \end{tabular}}
\caption{\wtms specification: \raext{highlighted} components are
  extensions necessary for client synchronisation for C11
  transactions. We assume that the transactions are executed by thread
  $\tid$. Moreover, let $Q \cdot a$ be the sequence $Q$ appended with
  element $a$ and $f \oplus g$ be the function $f$ overridden by
  function $g$. Finally, let $\maxWr$ be a function that returns the
  write with the largest timestamp in the given set of writes. }
\Description{WTMS specification}
  \label{fig:wtms}
\end{figure}

\subsection{\wtms}
\label{sec:wtms-1}

\begin{figure*}[t]
  \centering
    \begin{minipage}[t]{0.45\linewidth}
      \begin{minipage}[b]{0.48\columnwidth}
        \begin{center}
          \scalebox{0.9}{
    \begin{tikzpicture}[scale=0.6]
      \coordinate (d0) at (0,3.5);
      \coordinate (d5) at (0,2);
      \coordinate (f0) at (-2.5,3.5);
      \coordinate (f1) at (-2.5,2);
      \coordinate (sigma) at (0,1.75);

      \draw[thick,->] (0,4.3) -- (0,1);

      \coordinate (T2-1) at (1.2,4.3);
      \coordinate (T1-1) at (1.2,2.5);
            
      \draw (0,4.3) node[above] {$d$};
      \draw (-2.5,4.3) node[above] {$M$};
      \draw (T2-1) node[above] {$\tidn{2}$};
      \draw (T1-1) node[above] {$\tidn{1}$};

      \filldraw [black] (d0) circle (2pt) node[right, black, yshift=-1mm] {$0,0$};
      \filldraw [black] (d5) circle (2pt) node[right, black, yshift=-1mm] {$5,1$}; 

      \node at (f0) (f0n) {$\{f \mapsto 0\}^{\sf RX}$}; 
      \node at (f1) (f1n) {$\{f \mapsto 1\}^{\sf R}$};

      \filldraw [black] (T1-1) circle (2pt);
    
      \filldraw [black] (T2-1) circle (2pt);

      \path [draw, blue, thick] (d5) edge [dashed] (T1-1);
      \path [draw, thick, red]  (d0) edge [dashed] (T2-1);
      \path [draw, dotted, orange!70!black, line width=0.4mm]
      ($(f1n.east) + (-0.1, 0)$) edge[-stealth] (d5)
      ($(f0n.east) + (-0.2, 0)$) edge[-stealth] (d0) ;
    \end{tikzpicture}}
  
      $\sigma_2$ (after executing line 4)
      \end{center}
    \end{minipage}
    \hfill
      \begin{minipage}[b]{0.48\columnwidth}
        \begin{center}
          \scalebox{0.9}{
            \begin{tikzpicture}[scale=0.6]
              \coordinate (d0) at (0,3.5);
              \coordinate (d5) at (0,2);
              \coordinate (f0) at (-2.5,3.5);
              \coordinate (f1) at (-2.5,2);
              
              \draw[thick,->] (0,4) -- (0,1);
              
              \coordinate (T2-1) at (1.2,4);
              \coordinate (T1-1) at (1.2,2.5);
              
              \draw (0,4) node[above] {$d$};
              \draw (-2.5,4) node[above] {$M$};
              \draw (T2-1) node[above] {$\tidn{2}$};
              \draw (T1-1) node[above] {$\tidn{1}$};
              
              \filldraw [black] (d0) circle (2pt) node[right, black, yshift=-1mm] {$0,0$};
              \filldraw [black] (d5) circle (2pt) node[right, black, yshift=-1mm] {$5,1$}; 
              
              \node at (f0) (f0n) {$\{f \mapsto 0\}^{\sf RX}$}; 
              \node at (f1) (f1n) {$\{f \mapsto 1\}^{\sf R}$}; 
              
              \filldraw [black] (T1-1) circle (2pt);
              
              \filldraw [black] (T2-1) circle (2pt);
              
              \path [draw, blue, thick] (d5) edge [dashed] (T1-1);
              \path [draw, thick, red]  (d5) edge [dashed] (T2-1);
              \path [draw, dotted, orange!70!black, line width=0.4mm]
               ($(f1n.east) + (-0.1, 0)$) edge[-stealth] (d5)
               ($(f0n.east) + (-0.2, 0)$) edge[-stealth] (d0) ;
            \end{tikzpicture}}
           
          $\sigma_3$ (after executing line 8)
        \end{center}
      \end{minipage}
      \subcaption{Views for transactional MP (\reffig{fig:mp-tm})}
      \label{fig:txmp-view}
    \end{minipage}
    \hfill \vline
    \begin{minipage}[t]{0.48\linewidth}
      \begin{minipage}[b]{0.48\columnwidth}
        \begin{center}
          \scalebox{0.9}{
    \begin{tikzpicture}[scale=0.6]
      \coordinate (d0) at (0,3.5);
      \coordinate (d5) at (0,1.7);
      \coordinate (f0) at (-2.5,3.5);
      \coordinate (f1) at (-2.5,1.7);

      \draw[thick,->] (0,4.3) -- (0,1);

      \coordinate (T2-1) at (1.2,4.3);
      \coordinate (T1-1) at (1.2,2.2);
            
      \draw (0,4.3) node[above] {$d_1$};
      \draw (-2.5,4.3) node[above] {$M$};
      \draw (T2-1) node[above] {$\tidn{2}$};
      \draw (T1-1) node[above] {$\tidn{1}$};

      \filldraw [black] (d0) circle (2pt) node[right, black, yshift=-1mm] {$0,0$};
      \filldraw [black] (d5) circle (2pt) node[right, black, yshift=-1mm] {$5,1$}; 

      \node at (f0) (f0n) {$
        \left\{\begin{array}{@{}l@{}}
          f \mapsto 0, \\ d_2 \mapsto 0
        \end{array}\right \}^{\sf RX}
$}; 

      \node at (f1) (f1n) {$
        \left\{\begin{array}{@{}l@{}}
          f \mapsto 1, \\ d_2 \mapsto 10 
        \end{array}\right \}^{\sf RX}
$}; 

      \filldraw [black] (T1-1) circle (2pt);
    
      \filldraw [black] (T2-1) circle (2pt);

      \path [draw, blue, thick] (d5) edge [dashed] (T1-1);
      \path [draw, thick, red]  (d0) edge [dashed] (T2-1);
      \path [draw, dotted, orange!70!black, line width=0.4mm]
      ($(f1n.east) + (-0.5, 0)$) edge[-stealth] (d5)
      ($(f0n.east) + (-0.5, 0)$) edge[-stealth] (d0) ;
    \end{tikzpicture}}
  
      $\sigma_2$ (after executing line 5)
      \end{center}
    \end{minipage}
    \hfill
      \begin{minipage}[b]{0.48\columnwidth}
        \begin{center}
          \scalebox{0.9}{
            \begin{tikzpicture}[scale=0.6]
              \coordinate (d0) at (0,3.5);
              \coordinate (d5) at (0,1.7);
              \coordinate (f0) at (-2.5,3.5);
              \coordinate (f1) at (-2.5,1.7);
              
              \draw[thick,->] (0,4.3) -- (0,1);
              
              \coordinate (T2-1) at (1.2,4);
              \coordinate (T1-1) at (1.2,2.2);
              
              \draw (0,4.3) node[above] {$d_1$};
              \draw (-2.5,4.3) node[above] {$M$};
              \draw (T2-1) node[above] {$\tidn{2}$};
              \draw (T1-1) node[above] {$\tidn{1}$};
              
              \filldraw [black] (d0) circle (2pt) node[right, black, yshift=-1mm] {$0,0$};
              \filldraw [black] (d5) circle (2pt) node[right, black, yshift=-1mm] {$5,1$}; 
              
      \node at (f0) (f0n) {$
        \left\{\begin{array}{@{}l@{}}
          f \mapsto 0, \\ d_2 \mapsto 0
        \end{array}\right \}^{\sf RX}
$}; 

      \node at (f1) (f1n) {$
        \left\{\begin{array}{@{}l@{}}
          f \mapsto 1, \\ d_2 \mapsto 10 
        \end{array}\right \}^{\sf RX}
$}; 
              
              \filldraw [black] (T1-1) circle (2pt);
              
              \filldraw [black] (T2-1) circle (2pt);
              
              \path [draw, blue, thick] (d5) edge [dashed] (T1-1);
              \path [draw, thick, red]  (d0) edge [dashed] (T2-1);
              \path [draw, dotted, orange!70!black, line width=0.4mm]
              ($(f1n.east) + (-0.5, 0)$) edge[-stealth] (d5)
              ($(f0n.east) + (-0.5, 0)$) edge[-stealth] (d0) ;
            \end{tikzpicture}}
          
          $\sigma_3$ (after executing line 10)
        \end{center}
\end{minipage}
\subcaption{Views for relaxed transactions (\reffig{fig:mp-tm-2})}
  \label{fig:txmp2-view}
\end{minipage} \bigskip

\bigskip


\begin{minipage}[t]{\linewidth}
      \begin{minipage}[b]{0.33\columnwidth}
        \begin{center}
          \scalebox{0.85}{
            \begin{tikzpicture}[scale=0.6]
              \coordinate (f0) at (-2.5,3.5);
              \coordinate (f1) at (-2.5,1.7);
              \coordinate (f2) at (-2.5,-0.1);
              \coordinate (d0) at (0,3.5);
              \coordinate (d5) at (0,1.7);
              \coordinate (d20) at (2,3.5);
              
              \draw[thick,->] (0,4) -- (0,-0.3);
              \draw[thick,->] (2,4) -- (2,-0.3);
              
              \coordinate (T3-1) at (4,4);
              \coordinate (T2-1) at (4,2.5);
              \coordinate (T1-1) at (4,1);
              
              \draw (-2.5,4) node[above] {$M$};
              \draw (0,4) node[above] {$d_1$};
              \draw (2,4) node[above] {$d_2$};
              \draw (T3-1) node[above] {$\tidn{3}$};
              \draw (T2-1) node[above] {$\tidn{2}$};
              \draw (T1-1) node[above] {$\tidn{1}$};
              
              \filldraw [black] (d0) circle (2pt) node[right, black, yshift=-1mm] {$0,0$};
              \filldraw [black] (d5) circle (2pt) node[right, black, yshift=-1mm] {$5,1$}; 

              \filldraw [black] (d20) circle (2pt) node[right, black, yshift=-1mm] {$0,0$};

      \node at (f0) (f0n) {$
        \left\{\begin{array}{@{}l@{}}
          f \mapsto 0
        \end{array}\right \}^{\sf RX}
$}; 

      \node at (f1) (f1n) {$
        \left\{\begin{array}{@{}l@{}}
          f \mapsto 1
        \end{array}\right \}^{\sf R}
$}; 

              \filldraw [black] (T1-1) circle (2pt);
              \filldraw [black] (T2-1) circle (2pt);
              \filldraw [black] (T3-1) circle (2pt);              
              \path [draw, green!70!black, thick]
              (d0) -- (d20) (d20) edge [dashed] (T3-1);
              \path [draw, blue, thick]
              (d5) -- (d20) (d20) edge [dashed] (T1-1);
              \path [draw, thick, red]  (d0) edge[bend right = 30] (d20) (d20) edge [dashed] (T2-1);
              \path [draw, dotted, orange!70!black, line width=0.4mm]
              ($(f0n.east) + (-0.5, 0)$) edge[-stealth] (d0) 
              ($(f0n.east) + (-0.5, 0)$) edge[-stealth,bend left = 30] (d20)
              ($(f1n.east) + (-0.5, 0)$) edge[-stealth] (d5) 
              ($(f1n.east) + (-0.5, 0)$) edge[-stealth] (d20) ;
            \end{tikzpicture}}
          
          $\sigma_2$ (after executing line 4)
        \end{center}
      \end{minipage}
      \hfill
      \begin{minipage}[b]{0.33\columnwidth}
        \begin{center}
          \scalebox{0.85}{
            \begin{tikzpicture}[scale=0.6]
              \coordinate (f0) at (-2.5,3.5);
              \coordinate (f1) at (-2.5,1.7);
              \coordinate (f2) at (-2.5,-0.1);
              \coordinate (d0) at (0,3.5);
              \coordinate (d5) at (0,1.7);
              \coordinate (d20) at (2,3.5);
              \coordinate (d25) at (2,1.7);
              
              \draw[thick,->] (0,4) -- (0,-0.3);
              \draw[thick,->] (2,4) -- (2,-0.3);
              
              \coordinate (T3-1) at (4,4);
              \coordinate (T2-1) at (4,2.5);
              \coordinate (T1-1) at (4,1);
              
              \draw (-2.5,4) node[above] {$M$};
              \draw (0,4) node[above] {$d_1$};
              \draw (2,4) node[above] {$d_2$};
              \draw (T3-1) node[above] {$\tidn{3}$};
              \draw (T2-1) node[above] {$\tidn{2}$};
              \draw (T1-1) node[above] {$\tidn{1}$};
              
              \filldraw [black] (d0) circle (2pt) node[right, black, yshift=-1mm] {$0,0$};
              \filldraw [black] (d5) circle (2pt) node[right, black, yshift=-2mm] {$5,1$}; 

              \filldraw [black] (d20) circle (2pt) node[right, black, yshift=-1mm] {$0,0$};
              \filldraw [black] (d25) circle (2pt) node[right, black, yshift=-1mm] {$10,1$}; 

      \node at (f0) (f0n) {$
        \left\{\begin{array}{@{}l@{}}
          f \mapsto 0
        \end{array}\right \}^{\sf RX}
$}; 

      \node at (f1) (f1n) {$
        \left\{\begin{array}{@{}l@{}}
          f \mapsto 1
        \end{array}\right \}^{\sf R}
$}; 

      \node at (f2) (f2n) {$
        \left\{\begin{array}{@{}l@{}}
          f \mapsto 2
        \end{array}\right \}^{\sf RA}
$};

              \filldraw [black] (T1-1) circle (2pt);
              \filldraw [black] (T2-1) circle (2pt);
              \filldraw [black] (T3-1) circle (2pt);              
              \path [draw, green!70!black, thick]
              (d0) -- (d20) (d20) edge [dashed] (T3-1);
              \path [draw, blue, thick]
              (d5) -- (d20) (d20) edge [dashed] (T1-1);
              \path [draw, thick, red]  (d5) -- (d25) (d25) edge [dashed] (T2-1);
              \path [draw, dotted, orange!70!black, line width=0.4mm]
              ($(f0n.east) + (-0.5, 0)$) edge[-stealth] (d0) 
              ($(f0n.east) + (-0.5, 0)$) edge[-stealth,bend left = 30] (d20)
              ($(f1n.east) + (-0.5, 0)$) edge[-stealth] (d5) 
              ($(f1n.east) + (-0.5, 0)$) edge[-stealth] (d20) 
              ($(f2n.east) + (-0.25, 0)$) edge[-stealth] (d5) 
              ($(f2n.east) + (-0.25, 0)$) edge[-stealth] (d25) ;
            \end{tikzpicture}}
          
          $\sigma_4$ (after executing line 10)
        \end{center}
\end{minipage}
      \begin{minipage}[b]{0.32\columnwidth}
        \begin{center}
          \scalebox{0.85}{
            \begin{tikzpicture}[scale=0.6]
              \coordinate (f0) at (-2.5,3.5);
              \coordinate (f1) at (-2.5,1.7);
              \coordinate (f2) at (-2.5,-0.1);
              \coordinate (d0) at (0,3.5);
              \coordinate (d5) at (0,1.7);
              \coordinate (d20) at (2,3.5);
              \coordinate (d25) at (2,1.7);
              
              \draw[thick,->] (0,4) -- (0,-0.3);
              \draw[thick,->] (2,4) -- (2,-0.3);
              
              \coordinate (T3-1) at (4,4);
              \coordinate (T2-1) at (4,2.5);
              \coordinate (T1-1) at (4,1);
              
              \draw (-2.5,4) node[above] {$M$};
              \draw (0,4) node[above] {$d_1$};
              \draw (2,4) node[above] {$d_2$};
              \draw (T3-1) node[above] {$\tidn{3}$};
              \draw (T2-1) node[above] {$\tidn{2}$};
              \draw (T1-1) node[above] {$\tidn{1}$};
              
              \filldraw [black] (d0) circle (2pt) node[right, black, yshift=-1mm] {$0,0$};
              \filldraw [black] (d5) circle (2pt) node[right, black, yshift=-2mm] {$5,1$}; 

              \filldraw [black] (d20) circle (2pt) node[right, black, yshift=-1mm] {$0,0$};
              \filldraw [black] (d25) circle (2pt) node[right, black, yshift=-2mm] {$10,1$}; 

      \node at (f0) (f0n) {$
        \left\{\begin{array}{@{}l@{}}
          f \mapsto 0
        \end{array}\right \}^{\sf RX}
$}; 

      \node at (f1) (f1n) {$
        \left\{\begin{array}{@{}l@{}}
          f \mapsto 1
        \end{array}\right \}^{\sf R}
$}; 

      \node at (f2) (f2n) {$
        \left\{\begin{array}{@{}l@{}}
          f \mapsto 2
        \end{array}\right \}^{\sf RA}
$};

              \filldraw [black] (T1-1) circle (2pt);
              \filldraw [black] (T2-1) circle (2pt);
              \filldraw [black] (T3-1) circle (2pt);              
              \path [draw, green!70!black, thick]
              (d5) edge[bend left = 30] (d25) (d25) edge [dashed] (T3-1);
              \path [draw, blue, thick]
              (d5) -- (d20) (d20) edge [dashed] (T1-1);
              \path [draw, thick, red]  (d5) -- (d25) (d25) edge [dashed] (T2-1);
              \path [draw, dotted, orange!70!black, line width=0.4mm]
              ($(f0n.east) + (-0.5, 0)$) edge[-stealth] (d0) 
              ($(f0n.east) + (-0.5, 0)$) edge[-stealth,bend left = 30] (d20)
              ($(f1n.east) + (-0.5, 0)$) edge[-stealth] (d5) 
              ($(f1n.east) + (-0.5, 0)$) edge[-stealth] (d20) 
              ($(f2n.east) + (-0.25, 0)$) edge[-stealth] (d5) 
              ($(f2n.east) + (-0.25, 0)$) edge[-stealth] (d25) ;
            \end{tikzpicture}}
          
          $\sigma_5$ (after executing line 13)
        \end{center}
\end{minipage}
\subcaption{Views for release-acquire transaction chain
  (\reffig{fig:ra-trans-mem})}
  \label{fig:ra-trans-mem-view}
\end{minipage}
\caption{Views for the transaction-based client program from \reffig{fig:trans-client}}
\Description{Views for the transaction-based client program from \reffig{fig:trans-client}}
\label{fig:trans-mem-views}
\end{figure*}
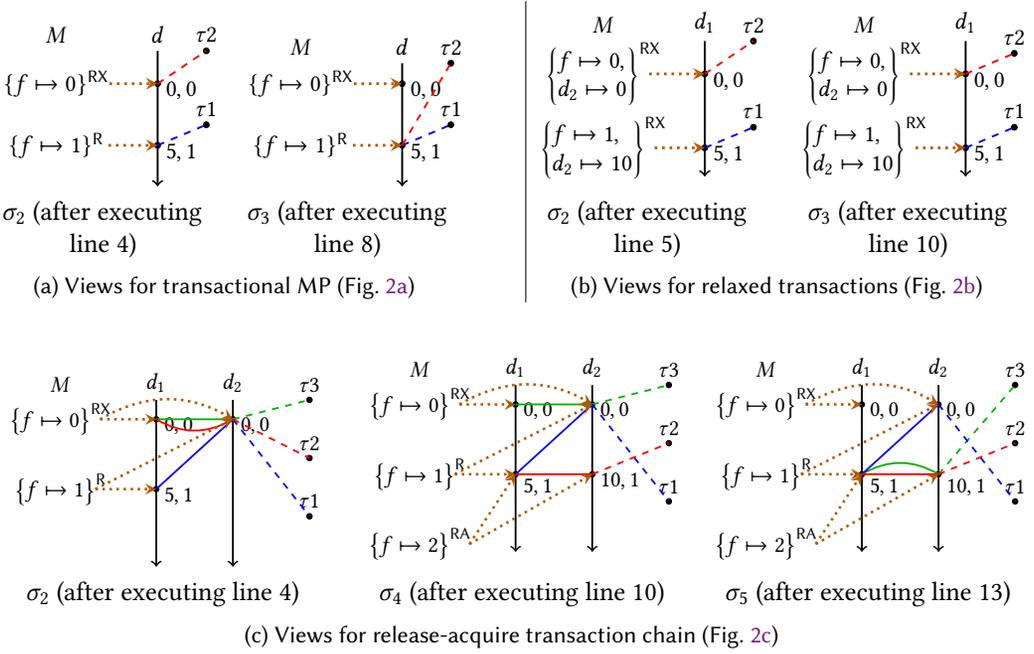

We now discuss the release-acquire extensions of \wtms, as defined by
the I/O automata algorithm in \reffig{fig:wtms}, including the
highlighted components. The key extension of \wtms is its ability to
synchronise client threads, thus allowing it to cope with the examples
in \reffig{fig:trans-client}. Formally, this is achieved by ensuring
\wtms synchronises the {\em thread view} of the client whenever
transactional release-acquire synchronisation occurs.

We introduce two new local variables in transaction~$t$. Namely, {\tt
  synctype}$_t$, which records the type of synchronisation of $t$, and
{\tt seenIdxs}$_t$, which records the set of all memory indices seen
by $t$ that are either releasing or
release-acquiring. We also introduce a new thread local variable 
$\txview$, which records the 
{\it transaction thread view} of $\tau$. The transaction thread view
 is similar to 
{\it thread view} introduced in \refsec{sec:views}. The difference here is 
in the definition of $\View$. In this context the $\View$ is a function that 
maps the threads to memory indexes of $M$. The transaction view of
$\tau$ is the smallest memory in $M$ that can be read by any transaction 
$t$ that was begun by thread $\tau$.
We also introduce two global variables. Namely $S$, which is a
sequence recording the type of each committed writing transaction that
installs each new memory in $M$, and $V$ which is a sequence of
modification views for each new memory in $M$. Thus, in \wtms, memory
$M_i$ has synchronisation type $S_i$, and modification view
$V_i$. 

The transactional operations of TMS2 are modified as follows. In {\tt
  TxBegin}$_t$, we take as input the type of synchronisation
transaction $t$ is to perform, and store this value in {\tt
  synctype}$_t$. Another input to {\tt TxBegin}$_t$ is $m$, an index to
  a visible memory $M$.
  We also initialise {\tt seenIdxs}$_t$ to the empty
set. In {\tt TxReadExt}$_t(x,i)$, \ie a transition for external read
of $x$ from memory index $i$, we record the index $i$ in {\tt
  seenIdxs}$_t$ if the memory $M_i$ is releasing or release-acquiring.

When a transaction ends (for both read-only and writing transactions),
if the transaction is acquiring or release-acquiring and {\tt
  seenIdxs}$_t$ is non-empty, we construct a new view {\tt nv} to be
the maximum modification view for each transaction in {\tt
  seenIdxs}$_t$ using the function $\view$. We use this to synchronise
the client thread's view by updating $\tview_\tid$ to
$\tview_\tid \otimes{}${\tt nv}. For a writing transaction, we record
this new view of the client in $V$ so that any future transactions
that synchronise with this new transaction does so with respect to
this view. Finally, once a transaction ends, it updates $\txview$
to the largest index in $\seenIdxs$.

We demonstrate the interaction of $\wtms$ and a client program by
considering the views of three possible executions of the programs in
\reffig{fig:trans-mem-views}. Unlike the trace considered in
\reffig{fig:po-message-sync-exec}, we only show the most critical
transitions. The memory sequence $M$ of $\wtms$ is clear from the
figures. We represent $S$ by the superscripts on each state of $M$,
and the modification views $V$ by the dotted arrows (\orangearrow)
from each state of $M$.

\reffig{fig:txmp-view} represents part of an execution of the program
in \reffig{fig:mp-tm}. In the execution depicted, we assume that all
of thread $\tidn{1}$ executes before $\tidn{2}$. Here, $\sigma_2$ is
the state after executing line~$4$, where $\tidn{1}$ has introduced a
new write to $d$ and then executed its (releasing) transaction,
introducing a new memory snapshot whose modification view becomes the
new write of $d$ (since $\tidn{1}$'s thread view is at this new
write). Then, when $\tidn{2}$ executes its (acquiring) transaction
that reads $1$ from $f$, it {\em synchronises} with the latest memory
snapshot, causing $\tidn{2}$'s thread view to be the new write of $d$
as well. This is analogous, as required, to the way in which views are
updated in C11 (see \reffig{fig:po-message-sync-exec}).

\reffig{fig:txmp2-view} represents a part execution of the program in
\reffig{fig:mp-tm-2}. Again, we assume a complete execution of thread
$\tidn{1}$ followed by $\tidn{2}$. State $\sigma_2$ is the state after
executing line~5, where $\tidn{1}$ has introduced a new write to
$d_1$. Now consider the state $\sigma_3$ (the state after execution of
line $10$), where $\tidn{1}$ introduces a new memory snapshot in $M$
with annotation {\sf RX} and modification view pointing to the new
write on $d_1$. When $\tidn{2}$ continues execution, its transaction
must be ordered after the latest memory snapshot, but this will {\em
  not} induce a release-acquire synchronisation. This means that
$\tidn{2}$'s view of $d_2$ will not be updated. However, since
$\tidn{2}$'s transaction occurs after $\tidn{1}$'s transaction,
$\tidn{2}$ is guaranteed to read $10$ for $d_2$. Note that the
transaction executed by $\tidn{2}$ is a read-only relaxed transaction,
the view of $\tidn{2}$ of the client variable ($d_1$) is
unchanged. However, $\tidn{2}$'s view of the transactional memory (not
shown in the diagrams) will be updated to the new memory state
$\{f \mapsto 1, d_2 \mapsto 10\}$.

Finally, \reffig{fig:ra-trans-mem-view} represents part of an execution
of \reffig{fig:ra-trans-mem} comprising the complete execution of
$\tidn{1}$, $\tidn{2}$ then $\tidn{3}$ in order.  State $\sigma_2$
represents the state after executing line~4, where the modification
view of the newly installed memory is consistent with the view of the
executing thread $\tidn{1}$. Then, in $\sigma_4$ (the state after
execution of line~10), we have a new write on $d_2$ with value 10 and
a further new snapshot that synchronises with the snapshot
$\{f \mapsto 1\}^{\sf R}$ causing the thread view of $\tidn{2}$ and
the modification view of the new snapshot to be updated to the last
writes of $d_1$ and $d_2$. Next, in $\sigma_5$, when $\tidn{3}$
executes its transaction, this transaction is guaranteed to
synchronise with $\{f \mapsto 2\}^{\sf RA}$, causing $\tidn{3}$'s view
to be updated to the latest writes of $d_1$ and $d_2$, which is
inherited from the modification view of $\{f \mapsto 2\}^{\sf RA}$.


\subsection{Modular operational semantics}

\begin{figure*}[!t]
  \centering  \small
  $\inference[{\sc TxBegin}$_\tid(\sflag, m, \regset)$] {
  \gamma.\status_\txid = {\tt NOTSTARTED} \qquad \gamma.\txn_\tau = \bot \qquad m \in vmems_\tau}
    {\lst, \gamma, \beta\  \strans{~~~~~~~~~\ ~~~~~~~~~~~~}_\tid\  \lst, \gamma\left[
      \begin{array}[c]{@{}l@{}}
        \beginIdx_\txid \asgn m, \seenIdxs_\txid \asgn \emptyset, \rset_\txid \asgn \emptyset, \\
        \wset_\txid \asgn \emptyset, \synctype_\txid \asgn \sflag, \regSet_t := \regset, 
          \\
        \status_\txid = {\tt READY},  \txn_\tid := \txid 
        
      \end{array}\right], \beta}$
    \medskip

    \bigskip
    $\inference[{\sc TxWrite}$_{\tid}(l, v)$] { \gamma.\status_\txid = {\tt READY} \qquad \gamma.\txn_\tid = \txid }
    {\lst, \gamma, \beta\  \strans{~~~~~~~~~~~\ ~~~~~~~~~~}_\tid\  \lst, \gamma\left[
      \begin{array}[c]{@{}l@{}}
        \wset_\txid \asgn \gamma.\wset_\txid \cup \{l \mapsto v\}
      \end{array}\right], \beta}$
      \bigskip

      \bigskip
      $\inference[{\sc TxRead}$_{\tid}(l, r)$] { \gamma.\status_\txid = {\tt READY} \qquad  \gamma.\txn_\tid = \txid \qquad r \in \gamma.\regSet_\txid \\
      (l \in \dom(\gamma.\wset_\txid) \lor (\gamma.\beginIdx_\txid \leq i ~ \land \gamma.\rset_\txid \subseteq \gamma.\memories_i)) \\
      \ensuremath{
        \begin{array}[t]{@{}r@{~}l}
          v = & \kwif\  l \notin  \dom(\gamma.\wset_\txid) ~ \kwthen ~ \gamma.\memories_i(l) ~ \kwelse~ \gamma.\wset_\txid(l) \\[2pt]
      \seenIdxs' = & 
          \kwif\  l \notin  \dom(\gamma.\wset_\txid) \wedge \gamma.S_i \in \{\R, \R\A\} \\
          & \kwthen ~ \gamma.\seenIdxs_\txid \cup \{i\} ~ \kwelse~ \gamma.\seenIdxs_\txid \\[2pt]
            \rset' = & \kwif\  l \notin  \dom(\gamma.\wset_\txid) ~ \kwthen ~ \gamma.\rset_\txid \cup \{l \mapsto v \} ~ \kwelse~ \gamma.\rset_\txid
                               \end{array}}
                       }
      {\lst, \gamma, \beta\  \strans{~~~~~~~~~~\ ~~~~~~~~~~~}_\tid\  \lst[r:= v], \gamma\left[
        \begin{array}[c]{@{}l@{}}
          \rset_\txid \asgn \rset', \seenIdxs_\txid \asgn \seenIdxs'
        \end{array}\right], \beta}$
        \bigskip

        \bigskip
        $\inference[{\sc TxEndRO}$_{\tid}$] { \gamma.\status_\txid = {\tt READY} \qquad \gamma.\txn_\tid = \txid \qquad
        \gamma.\wset_\txid = \emptyset \\
        \tview' =
        \ensuremath{\begin{array}[t]{@{}l}
          \kwif ~ \gamma.\rset_\txid \neq \emptyset \land \gamma.\synctype_\txid 
          \in \{\A,\R\A\} \wedge \gamma.\seenIdxs_\txid \neq \emptyset \\
          \kwthen~ \beta.\tview_\tid \otimes \view(\gamma.\seenIdxs_\txid, \gamma.V) ~\kwelse~ \beta.\tview_\tid
        \end{array}}
      }        {\lst, \gamma, \beta\  \strans{~~~~~~~~~~~\ ~~~~~~~~~~}_\tid\  \lst, \gamma\left[
          \begin{array}[c]{@{}l@{}}
            \status_\txid \asgn {\tt COMMITTED}, \\
             \txview_{\tid}:= \max(\gamma.\seenIdxs_t)
          \end{array}\right], \beta[\tview_\tid \asgn \tview']}$

          \bigskip

          \bigskip

          $\inference[{\sc TxEndWR}$_{\tid}$] { \gamma.\status_\txid = {\tt READY} \qquad \gamma.\txn_\tid = \txid \qquad 
            \gamma.\wset_\txid \neq \emptyset \qquad  i = |\gamma.\memories| \\
          \mrel' = \gamma.\synctype_{\txid} \qquad  mem' = (\last(\gamma.\memories) \oplus \gamma.\wset_\tid)\\
        \tview' = \ensuremath{\begin{array}[t]{@{}l}
          \kwif ~ \gamma.\rset_\txid \neq \emptyset \land \gamma.\synctype_\txid 
          \in \{\A,\R\A\} \wedge \gamma.\seenIdxs_\txid \neq \emptyset \\
          \kwthen~ \beta.\tview_\tid \otimes \view(\gamma.\seenIdxs_\txid, \gamma.V) ~\kwelse~ \beta.\tview_\tid
        \end{array}}
          }        {\lst, \gamma, \beta\  \strans{~~~~~~\ ~~~~~~~~~~~~~~~}_\tid\  \lst, \gamma\left[
            \begin{array}[c]{@{}l@{}}
              \status_\txid \asgn {\tt COMMITTED}, \memories_i \asgn mem' \\
              \mrel_i := \mrel', 
              \mmview_i := \tview' \\
              \txview_{\tid}:= \max(\gamma.\seenIdxs_t)
            \end{array}\right], \beta[\tview_\tid \asgn \tview']}$

            \bigskip

            \bigskip $\inference[{\sc
              TxAbort}$_{\tid}$] {\gamma.\status_\txid = {\tt READY}
              \qquad \gamma.\txn_\tid = \txid \qquad \lst' = \lambda
              r \in \Reg.\ \kwif\ r \in \gamma.\regSet_\txid\ \kwthen\ \bot\ \kwelse\
              \lst(r)} {\lst, \gamma, \beta\ \strans{~~~~~~~~~~\
                ~~~~~~~~~~~}_\tid\ \lst', \gamma\left[
              \begin{array}[c]{@{}l@{}}
                \status_\txid \asgn {\tt ABORTED}\end{array}\right], \beta}$

          \caption{Operational semantics for TMS2-RA}
          \Description{Operational semantics for TMS2-RA}
  \label{fig:tms2-opsem}
\end{figure*}

To reason about clients that use abstract \wtms transactions in a
modular fashion, we use configurations that are triples
$(\lst, \gamma, \beta)$, where $\lst: \Reg \to \Val$ denotes the local
register state, $\gamma$ is the \wtms state (which includes all
transactional variables described in \reffig{fig:wtms}) and $\beta$ is
the C11 state of the client
(see~\cite{ECOOP20,DBLP:journals/corr/abs-2004-02983},
\refsec{sec:oper-semant}).

The transition relation for transactional operations is given in
\reffig{fig:tms2-opsem}. These follow the automata-style description
given in \reffig{fig:wtms}, but make state components that are affect
more precise. The most interesting aspect of these rules is the
interaction between a releasing writing transaction and subsequent
committing reading transaction.

Note that each releasing writing transaction sets $S_i$ (where $i$ is
last index in $M$ at the time of writing) to either $\R$ or
$\R\A$. Additionally, the view of the thread at the time of writing is
recorded in $V_i$.
A later transaction with an acquiring annotation calculates a new view
using the function $\view$ as defined in \reffig{fig:wtms} and
updates, among other components, the executing thread's view in
$\beta$. This means that, as expected, if there is a release-acquire
synchronisation through a transactional memory library, then the
client's view will be updated to match the synchronisation that
occurs.




\section{A C11 STM implementation}
\label{sec:tml-c11}

In this section, we develop a release-acquire version of a
transactional mutex lock, that we call \tmlra, based on an SC
implementation by Dalessandro et
al~\cite{DBLP:conf/europar/DalessandroDSSS10}.  Our algorithm is
provided in \reffig{fig:TML-RA}, where the highlights indicate the
fragments of code that we have introduced or modified. The
\raext{grey} highlights represent code additional to Dalessandro
\etal's original implementation, and the \highlight{blue} highlights
represent the necessary release-acquire synchronisation.

We first discuss the core features of TML (\refsec{sec:tml}), then
discuss the extensions introduced in \tmlra to optimise for C11
release-acquire synchronisation (\refsec{sec:tmlra}). We present the
benchmarking results for both algorithms in
\refsec{sec:benchmarking}. In \refsec{sec:cont-refin-b}, we present a
proof that \tmlra implements \wtms, \ie any observation a client
program makes when it uses \tmlra is a possible observation when it
uses \wtms.

\begin{figure*}[t]
  \centering
  \begin{minipage}[b]{1\linewidth}
    \scalebox{0.99}{$\begin{array}{l@{\qquad}l}
                       \multicolumn{2}{l}{${\bf Init: } ${\tt glb}:=0$ $}\\[1mm]
                       \begin{array}[t]{@{}l@{}}
                         {\tt TxBegin}($\regset$) \\
                         B1\!: {\tt regs} := $ \regset$ \\
                         B2\!: \raext{${\tt hasRead} := \False$}; \\
                         B3\!: \kwdo\ {\tt loc}\ \highlight{$\gets^{\sf A}$}\ {\tt glb}; \\
                         B4\!: \kwuntil\ even({\tt loc}) 
                         \\[2mm]
                         {\tt TxWrite}(x, v) \\
                         W\hspace{-1pt}1\!: \kwif\ even({\tt loc})\ \kwthen \\
                         W\hspace{-1pt}2\!: \quad r_1 \gets \highlight{$\kwcas^{\sf RA}$}({\tt glb}, {\tt loc}, {\tt loc+1})\\
                         W\hspace{-1pt}3\!:  \quad  \kwif\ \neg r_1\ \kwthen\\
                         W\hspace{-1pt}4\!: \qquad \forall s \in {\tt regs}.\, s:=\bot;
                         \kwreturn ; \texttt{// ABORT} \\
                         W\hspace{-1pt}5\!: \quad  \kwelse\ {\tt loc}:={\tt loc}+1; \\
                         W\hspace{-1pt}6\!: x\ \highlight{$:=^{\sf R}$}\ v; \texttt{ // WRITE OK} 
                         \\[2mm]
                         {\tt TxEnd} \\   
                         E1\!: \kwif\ odd({\tt loc})\  \kwthen\  \\
                         E2\!: \quad  {\tt glb}\ \highlight{$:=^{\sf R}$}\ {\tt loc + 1};  
                       \end{array}
                       & 
                       \begin{array}[t]{@{}l@{~}l}
                         \multicolumn{2}{@{}l}{{\tt TMRead}(x, r)} \\  
                         R1\!: &  \kwif\ r \in {\tt regs}\ \kwthen \\
                         R2\!: &  \quad r\ \highlight{$\gets^{\sf A}$}\ x; \\
                         R3\!: &  \quad\kwif\ \raext{${}\neg {\tt hasRead} \wedge even({\tt loc})$}\ \kwthen \\
                         R4\!: &  \quad\quad \raext{$r_1 \gets \highlight{$\kwcas^{\sf RA}$}({\tt glb}, {\tt loc}, {\tt loc})$} \\
                         R5\!: &  \quad\quad \raext{$\kwif\ r_1\ \kwthen$}  \\
                         R6\!: &  \quad\qquad \raext{${\tt hasRead} := true;$}\\
                         R7\!: &  \quad\qquad \raext{$
                                 \begin{array}[t]{@{}l@{}}
                                   \kwreturn ; \ \  \texttt{// READ OK}
                                 \end{array}
$} \\ 
                         R8\!: &  \quad\kwelse \\
                         R9\!: &  \quad\quad r_1 \gets {\tt glb}; \\ 
                         R10\!: &  \quad\quad \kwif\ r_1 = {\tt loc}\ \kwthen  \\
                         R11\!: &  \quad\qquad
                                  \begin{array}[t]{@{}l@{}}
                                    \kwreturn ;\ \  \texttt{// READ OK}
                                  \end{array}
\\
                         R12\!: &  \forall s \in {\tt regs}.\, s:=\bot;  \texttt{// ABORT}
                       \end{array}
                     \end{array}$}
\end{minipage}
\caption{\tmlra: A release-acquire transactional mutex lock. For
  simplicity, the thread id is omitted}
\Description{A release-acquire transactional mutex lock}
\label{fig:TML-RA}
\end{figure*}

\subsection{TML}
\label{sec:tml}
 TML is synchronised using
a single global counter {\tt glb}, initialised to $0$, where {\tt glb}
is even iff no writing transaction is currently executing.

A transaction begins by taking a snapshot of {\tt glb} in local
variable {\tt loc} and only begins if the value read is even.

A write operation checks that {\tt loc} is even and if so, it attempts
to increment {\tt glb} using the $\kwcas$ at line
$W\hspace{-1pt}2$. If this $\kwcas$ succeeds, it increments {\tt loc}
(line $W\hspace{-1pt}5$), then immediately updates the location $\loc$
(line $W\hspace{-1pt}6$). If the $\kwcas$ fails, the transaction
aborts. Note that if {\tt loc} is odd then the current transaction
``owns'' the lock, meaning that lines
$W\hspace{-1pt}2$-$W\hspace{-1pt}5$ can be bypassed.

A read operation (ignoring lines $R3$-$R7$ for now) reads the given
location into the given register $r$ (line $R2$). At lines $R9$ and
$R10$ it checks that {\tt glb} is consistent with {\tt loc}. If so, the
read succeeds, otherwise, the transaction aborts.

A transaction ends by checking whether the current transaction is a
writing transaction. This can be determined by checking whether {\tt
  loc} is odd since a writing transaction must have incremented {\tt
  glb} via the $\kwcas$ at line $W\hspace{-1pt}2$ and {\tt loc} via
the write at line $W\hspace{-1pt}5$ making both their values
odd. Therefore, a writing transaction must increment {\tt glb} to make
it even again.

\subsection{\tmlra}
\label{sec:tmlra}
We now describe the necessary modifications to TML and the
synchronisation induced by \wtms. We assume that transactions in
\tmlra are all release-acquiring and hence we omit the transaction
annotation in {\tt TxBegin}.

We assume all accesses to shared variables are either relaxed (\eg the
read at line $R9$), releasing (\eg the write at line $E2$), acquiring
(\eg the read at line $B2$) or release-acquiring (\eg the $\kwcas$ at
line $R4$). Additionally, we introduce a new local variable {\tt
  hasRead}, initially set to $\False$ and a code path $R3$-$R7$, which
is followed if a transaction performs a read without having previously
performed a read or a write. We explain the purpose of this code path
in more detail
below.  

\paragraph{Transaction synchronisation.}  Recall that
\wtms requires that transactions are consistent w.r.t. a single memory
snapshot and that external reads of a transaction synchronise with
\emph{some} memory snapshot. This may not occur in a relaxed memory
context without adequate synchronisation. In particular, a writing
transaction must perform a releasing write to {\tt glb} at line $E2$ so that if a later transaction reads from this write, it
synchronises with {\em all} of the writes performed by the writing
transaction. To ensure this, we require the read of {\tt glb} at line
$B3$ as well as the $\kwcas$ operations at lines $W\hspace{-1pt}2$ and
$R4$ to be acquiring. Note that this also guarantees release-acquire
{\em client synchronisation}.

The second key synchronisation is between $W\hspace{-1pt}6$ performed
by a writing transaction $t_w$ and $R2$ performed by a (different)
reading transaction $t_r$. Suppose that both $t_w$ and $t_r$ are
live. If $t_r$ happens to read the write written at $W\hspace{-1pt}6$,
it must now abort because $t_r$'s snapshot of {\tt glb} will be
inconsistent with the latest value of {\tt glb} installed by the
$t_w$. The release-acquire synchronisation between $W\hspace{-1pt}6$
and $R1$ ensures that this will happen, \ie $t_r$ will see the new
{\tt glb} written by $t_w$, causing the test at $R10$ to fail and
$t_r$ to abort.



\paragraph{Causal linearizability.}  The design of
\tmlra ensures that all transactions, including read-only transactions
are {\em causally linearizable}~\cite{ifm18}, which is a condition
that additionally guarantees {\em compositionality} (or {\em
  locality}~\cite{HeWi90,DBLP:conf/podc/0001HP21}) of concurrent
objects. This notion of compositionality is that of Herlihy and
Wing~\cite{HeWi90}. In particular, under SC memory, given a history
comprising several concurrent objects, if the history restricted to
each object is linearizable, then the history as a whole is
linearizable. 
In a relaxed memory setting, Doherty et al~\cite{ifm18} have shown
that linearizability alone is insufficient to guarantee
compositionality, and it is necessary to induce a ``happens-before''
relation when a specification induces a particular linearization. 

The happens-before required by causal linearizability is naturally
achieved for writing transactions via the $\kwcas$ at line
$W\hspace{-1pt}2$. For a read-only transaction, we introduce the
$\kwcas$ at line $R4$, which installs a new write to {\tt glb} without
changing its value. All transactions that follow the $\kwcas$ at line
$R4$ will be {\em causally ordered} after the reading
transaction. Such a $\kwcas$ must only be performed once, thus we
introduce a local variable {\tt hasRead}, which is set to true if the
$\kwcas$ succeeds so that later reads from the same transaction can
avoid the code path from $R3$-$R7$.

Note that the conditions necessary to guarantee causal linearizability
(and hence compositionality) could have been introduced at the level
of $\wtms$. However, there are questions about whether the notion of
compositionality introduced by Herlihy and Wing~\cite{HeWi90} are
appropriate in a relaxed memory
context~\cite{DBLP:journals/pacmpl/RaadDRLV19}. Therefore we leave out
the causal linearizability conditions in $\wtms$ to avoid
over-constraining the specification.

 \begin{figure*}[!t]
  \centering
    \scalebox{0.3}{\includegraphics{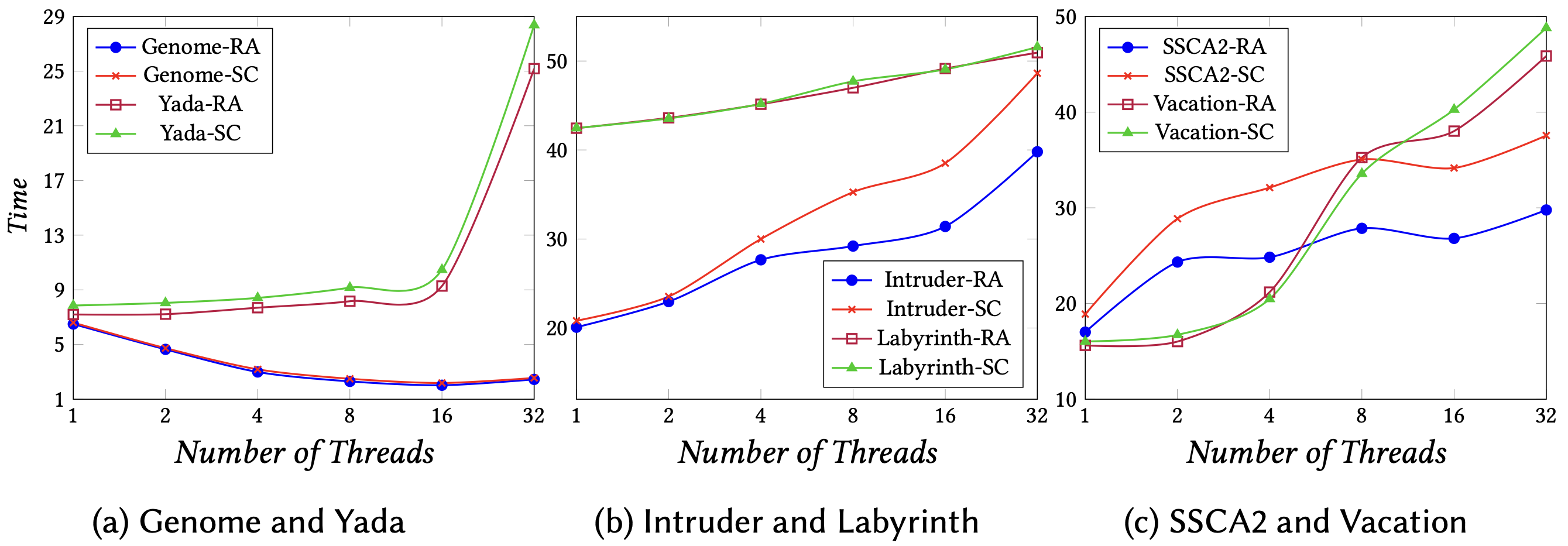}}

\vspace{-1em}
\noindent\rule{\textwidth}{0.5pt}
\vspace{-2em}
    \caption{Results of STAMP benchmarks for \tmlra and \tmlsc}
    \label{fig:stampres}
  \end{figure*}


  \subsection{Benchmarking}
  \label{sec:benchmarking}
  
  We implemented two versions of the TML algorithm: 
  \tmlra (see \reffig{fig:TML-RA}) and \tmlsc (the SC
  counterpart~\cite{DBLP:conf/europar/DalessandroDSSS10}) and
  benchmarked both using the STAMP benchmarking
  suite~\cite{DBLP:conf/iiswc/MinhCKO08}. 
  Each experiment was repeated 20 times to rule out external loads on
  the test machine and an average of these times was taken. The
  results of the six benchmarks that we ran with STAMP are presented
  in \reffig{fig:stampres}. \tmlra is equivalent to or outperforms
  \tmlsc in almost all cases, with a maximum improvement of 20\%. On
  average, \tmlra performs 8.2\% better than \tmlsc.

  Unsurprisingly, since TML optimises read-heavy workloads, its
  performance degrades under high write contention, and this is
  consistent with prior
  results~\cite{DBLP:conf/europar/DalessandroDSSS10}. However, it is
  interesting that the degradation of \tmlra is not as severe as
  \tmlsc for the Intruder and SSCA2 benchmarks. 

  \tmlra theoretically allows more parallelism than \tmlsc since a
   read-only transaction $t_r$ is not forced to abort if a writing transaction
   $t_w$ executes after $t_r$'s first read operation - $t_r$ must only aborts
   if it sees $t_w$'s $glb$ update, or one of $t_w$'s writes.
   Both Intruder and SSCA2 have a large number of short transactions;
    SSCA2 additionally has small read/write sets \cite{DBLP:conf/iiswc/MinhCKO08}.
     Here, \tmlra may be able to exploit the theoretical parallelism.
      In the single-threaded case, \tmlra executes far fewer heavyweight CASs.

  As with prior
  results, we see that for the read-heavy benchmark Genome, the
  performance of both \tmlra and \tmlsc improves as the number of
  threads increases.



\newcommand{\txndv}{\stackrel{S}{=}}
\newcommand{\com}{{\it Com}}

\newcommand{\Assertion}{{\it Assertion}}

\newcommand{\initlabel}{\iota}
\newcommand{\finlabel}{\zeta}

\newcommand{\ina}{{\it in}}
\newcommand{\fina}{{\it fin}}
\newcommand{\ann}{{\it ann}}
\newcommand{\Ann}{{\it Ann}}

\section{\logic: A logic for release-acquire TM}
\label{sec:logic:-logic-release}
The development of view-based operational semantics for various
fragments of
C11~\cite{DBLP:conf/ecoop/KaiserDDLV17,ECOOP20,DBLP:conf/popl/KangHLVD17}
has provided foundations for several logics for reasoning about C11
programs. These include separation
logics~\cite{DBLP:conf/ecoop/KaiserDDLV17,DBLP:conf/esop/SvendsenPDLV18}
and extensions to Owicki-Gries
reasoning~\cite{DBLP:journals/corr/abs-2108-01418,DBLP:conf/icalp/LahavV15,DBLP:journals/corr/abs-2004-02983,ECOOP20}. Our
point of departure is the Owicki-Gries encoding for RC11
RAR~\cite{ECOOP20}, which is the fragment of C11 that we focus on in
this paper.\footnote{These frameworks are based on models that assume
  top-level parallelism only. Therefore, our framework similarly
 re assumes top-level paralellism. This model can be extended to support
  dynamic parallelism, but such extensions are uninteresting for the
  purposes of this paper.}

A key benefit of the logic in~\cite{ECOOP20} is that it enables reuse
of {\em standard} Owicki-Gries proof decomposition rules and
straightforward mechanisation in
Isabelle/HOL~\cite{DBLP:journals/darts/DalvandiDDW20,DBLP:journals/corr/abs-2004-02983}. As
we shall see, we maintain these benefits in the context of C11 with
release-acquire transactions. Our reasoning framework, called \logic,
like Dalvandi \etal~\cite{ECOOP20,DBLP:conf/ppopp/DalvandiD21} uses
view-based assertions to abstractly describe the system state,
allowing reasoning about the {\em current view} of a thread, and {\em
  view transfer} from one thread to another through release-acquire
synchronisation. \logic introduces additional assertions to enable
reasoning about transactional views.

\subsection{View-based assertions}
\label{sec:view-based-assert}
In this section, we discuss the assertions and proof rules of \logic
abstractly. The proof rules can be used to reason syntactically about
a program without having to understand the low-level operational
semantics of the C11 model. Our operational semantics is an extension
of prior
works~\cite{DBLP:conf/ecoop/KaiserDDLV17,ECOOP20,DBLP:conf/popl/KangHLVD17}
that include an encoding of \wtms. 

\begin{figure}[t]
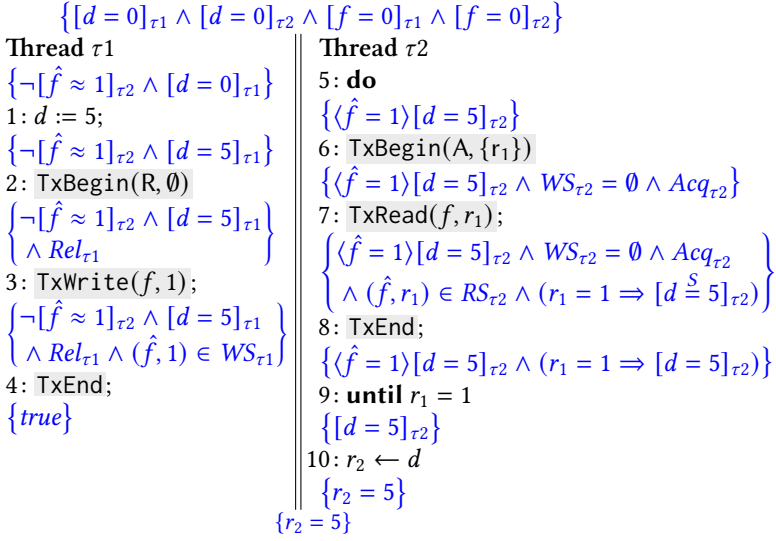

  \centering
  \scalebox{1}{
   $\begin{array}{@{}l@{~}||@{\,}l@{}}
      \multicolumn{2}{@{}l@{}}{
      \qquad \assert{[d = 0]_{\tidn{1}} \wedge [d = 0]_{\tidn{2}} \wedge [f = 0]_{\tidn{1}} \wedge [f = 0]_{\tidn{2}}}} \\
      \begin{array}[t]{@{}l@{}}
        {\bf Thread }\ \tidn{1}\\
        \assert{\neg [\hat{f} \approx 1]_{\tidn{2}} \wedge [d = 0]_{\tidn{1}}}\\
        1\!: d := 5; \\
        \assert{\neg [\hat{f} \approx 1]_{\tidn{2}} \wedge [d = 5]_{\tidn{1}}}\\
        2\!:\raext{${\tt TxBegin}(\sf R, \emptyset)$} \\
        \assert{\neg [\hat{f} \approx 1]_{\tidn{2}} \wedge [d = 5]_{\tidn{1}}  \\ {} \wedge \rel_{\tidn{1}}}\\
        3\!:\raext{${\tt TxWrite}(f, 1)$};\\
        \assert{\neg [\hat{f} \approx 1]_{\tidn{2}} \wedge [d = 5]_{\tidn{1}}  \\ {} \wedge \rel_{\tidn{1}} \wedge (\hat{f}, 1) \in \wrset_{\tidn{1}} }\\
        4\!:\raext{${\tt TxEnd}$};\\
        \assert{\True }\\
      \end{array}
      & 
      \begin{array}[t]{@{}r@{~}l@{}}
        \multicolumn{2}{l}{{\bf Thread }\ \tidn{2}}\\
        5\!:& \kwdo \\
        \multicolumn{2}{l@{}}
        {\assert{\langle \hat{f} = 1\rangle[d = 5]_{\tidn{2}}}} \\
        6\!:& \raext{${\tt TxBegin}(\sf A, \{r_1\})$} \\
        \multicolumn{2}{l@{}}
        {\assert{\langle \hat{f} = 1\rangle[d = 5]_{\tidn{2}} \wedge \wrset_{\tidn{2}} = \emptyset  \wedge \acq_{\tidn{2}}}} \\
        7\!:& \raext{${\tt TxRead}(f, r_1)$};\\
        \multicolumn{2}{l@{}}
        {\assert{\langle \hat{f} = 1\rangle[d = 5]_{\tidn{2}} \wedge \wrset_{\tidn{2}} = \emptyset \wedge \acq_{\tidn{2}}
        \\
        {}  \wedge (\hat{f}, r_1) \in \rdset_{\tidn{2}}
         \wedge (r_1 = 1 \imp 
        [d \txndv  5]_{\tidn{2}})}} \\
        8\!:& \raext{${\tt TxEnd}$};\\
        \multicolumn{2}{l@{}}
        {\assert{\langle \hat{f} = 1\rangle[d = 5]_{\tidn{2}}   \wedge (r_1 = 1 
        \imp  [d =  5]_{\tidn{2}})}} \\
        9\!:& \kwuntil\ r_1 = 1 \\
        \multicolumn{2}{l@{}}
        {\assert{[d =  5]_{\tidn{2}}}} \\
        10\!:& r_2 \gets d\\
        \multicolumn{2}{l@{}}
        {\assert{r_2 = 5}} \\
      \end{array}
   \end{array}$}
   
 \hfill \scalebox{0.9}{\color{blue}  $\{ r_2=5\}$\qquad \qquad \qquad } \hfill{}

  \vspace{-0.5em}
  \caption{Proof outline for transactional MP from \reffig{fig:mp-tm}}
  \Description{Proof outline for transactional MP}
 \label{fig:mp-tm-proof}
 \vspace{-0.5em}
\end{figure}

To motivate \logic, consider the proof outline in
\reffig{fig:mp-tm-proof} for the transactional message passing program
from \reffig{fig:mp-tm}. We use\ `\ $\hat{}$\ '\ to distinguish
transactional locations in a proof. For the program in
\reffig{fig:mp-tm-proof}, we have a transactional location $\hat{f}$.

\subsubsection{View assertions}
The proof outline contains three assertions from~\cite{ECOOP20}
describing the {\em views} that each thread may have of the system
state. Recall (\refsec{sec:views}), that we can define the set of
values that a thread can see in each state using the function $OV$.
\begin{itemize}
\item A {\em definite value} assertion, denoted $[\loc = \val]_\tid$,
  holds iff thread $\tid$ sees the last write to location $\loc$ and
  this write has value $\val$. Thus,
  $[\loc = \val]_\tid \imp OV_\tid(\loc) = \{\val\}$.
\item A {\em possible value} assertion, denoted
  $[\loc \approx \val]_\tid$, which holds iff when $\tid$ can see a
  write to $\loc$ with value $\val$. $[\loc \approx \val]_\tid$ is
  shorthand for $\val \in OV_\tid(\loc)$.
\item A {\em conditional value} assertion, denoted
  $\langle \locb = \valb \rangle[\loc = \val]_\tid$, which holds iff
  an acquiring read of $\locb$ by $\tid$ that returns the value
  $\valb$ is guaranteed to induce a release-acquire synchronisation so
  that $[\loc = \val]_\tid$ holds after this read.
\end{itemize}

\begin{example}
  Consider the third state, \ie $\sigma_2$ in
  \reffig{fig:po-message-sync-exec}. There, we have
  $[d = 5]_{\tidn{1}} \wedge [f = 1]_{\tidn{1}}$ as well as
  $[f \approx 0]_{\tidn{2}} \wedge [f \approx
  1]_{\tidn{2}}$. Moreover, we have
  $\langle f =1 \rangle [d = 5]_{\tidn{2}}$.
\end{example}
We ask the interested reader to consult
\cite{ECOOP20,DBLP:journals/corr/abs-2004-02983} for further details
of these assertions.

\subsubsection{Transactional assertions}

As alluded to above, \logic introduces several new assertions to
describe the transactional state. These assertions are, in general,
local to the transaction being executed, and hence, stable under the
execution of other threads. \reffig{fig:mp-tm-proof} contains the
following transaction local assertions:
\begin{itemize}
\item $\rel_\tid$, which holds iff $\tid$ is executing a releasing or
  release-acquiring transaction.
\item $\acq_\tid$, which holds iff $\tid$ is executing an acquiring
  or release-acquiring transaction.
\item $(\hat{\loc}, \val) \in \wrset_\tid$ (and
  $(\hat{\loc}, \val) \in \rdset_\tid$), which holds iff $\tid$ is
  executing a transaction whose write set (resp. read set) contains a
  write to (resp. read of) $\hat{\loc}$ with value $\val$.
\item $[\loc \txndv \val]_{\tid}$, which holds iff $\tid$ is executing
  a transaction such that committing this transaction results in the
  definite value assertion $[\loc = \val]_{\tid}$ (see above).
\end{itemize}

In addition, we include a number of assertions 
This section provides the formal definition for 
 the assertion language used in the verification
of client programs that use TMS2-RA (See \refsec{sec:view-based-assert}). The assertion language presented 
here is heavily inspired by the view-based assertion 
language presented in \cite{ECOOP20}. 


A memory $i$ is visible to a transaction executed by a thread $\tid$
iff $i$ is greater than the transaction thread view of $\tau$ ($\txview_{\tau}$)
and is less than the maximum index of the memory ($|M| - 1$).  We define the set of visible memories
$\vismem_t$ to be:

\begin{align*}
  \vismem_{\tau} =  \{ n ~|~ n \geq \txview_{\tau} \land n \leq |M|-1 \}  
\end{align*}

\begin{itemize}
\item A {\em transactional definite observation} assertion, denoted
  $[\hat{x} = v]_{\txid}$, holds iff for all memory versions $i$, where
  $i$ is greater than or equal to $\beginIdx_t$, the value of
  $\memories_i(x)$ is $v$. Formally, for a transactional state $\gamma$:
    \begin{align*}
       [\hat{x} = v]_{\tid}(\gamma) \ \  \eqdef\ \  \forall i \in \gamma.\vismem_\tid.\ 
                 \gamma.\memories_i(\hat{x}) = v
    \end{align*}
    These are lifted to client-object states $(\gamma, \beta)$ in the normal manner, e.g.,
    $[\hat{x} = v]_{\tid}(\gamma, \beta) = [\hat{x} = v]_{\tid}(\gamma)$
  \item A {\em transactional possible observation} assertion, denoted
    $[\hat{x} \approx v]_\tid$, holds iff there exists a memory version
    $i$ that has value $v$ for $\hat{x}$. Formally: 
     \begin{align*}
        [\hat{x} \approx v]_{\tid}(\gamma)\ \ \eqdef \ \  \exists i \in \gamma.\vismem_\tid.\ 
                \gamma.\memories_i(\hat{x}) = v
    \end{align*}
  \item A {\em transactional conditional observation} assertion,
    denoted $\langle \hat{\locb} = \valb \rangle[\loc = \val]_\tid$,
    holds iff an acquiring transactional read of $\locb$ by $\tid$
    that returns a value $\valb$ is guaranteed to induce a
    release-acquire synchronisation so that $[\loc \txndv \val]_\tid$ holds
    in the client state after the reading transaction successfully
    commits. Formally: 
\begin{align*}
  &\langle \hat{\locb} = \valb \rangle[\loc = \val]_\tid (\gamma, \beta) \ \  \eqdef\ \
    \begin{array}[t]{@{}l@{}}
      \forall i \in \gamma.\vismem_\tid.\    
      \gamma.\memories_i(\hat{\locb}) = \valb \imp \\ 
      \qquad \gamma.\txnview_i(x) = \beta.last(x) \wedge \valu(\beta.\last(x)) = v 
      \land \gamma.\mrel_i
    \end{array}
\end{align*}
where $last(x)$ is the last write to $\loc$ in the modification
order. It is important to note that this assertion is over two states:
transaction state $\gamma$ and client state $\beta$, explaining
transfer of information across two threads using the transactional
memory. In particular, the transactional view $V_i$ for the memory
index $i$ must see the last write to $x$ in the client
state~$\beta$. This means that the thread that committed the
transaction writing the value $u$ to $\hat{y}$ did so when it saw the
last write to $x$.


%
 \end{itemize}

 \begin{example}
   Returning to our transactional MP example (\reffig{fig:mp-tm-proof}),
the precondition of line 1 contains assertions
$\neg [\hat{f} \approx 1]_{\tidn{2}}$ and $[d = 0]_{\tidn{1}}$, which
ensure that, prior to executing line 1, thread $\tidn{2}$ {\em cannot}
see the value $1$ for $\hat{f}$ and thread $\tidn{1}$ {\em must} see
the value $0$ for $d$, respectively. In the postcondition of line 2,
$[d = 0]_{\tidn{1}}$ changes to $[d = 5]_{\tidn{1}}$ since $\tidn{1}$
performs a write to $d$ with value $5$. The other view-based
assertions in $\tidn{1}$ are similar. We explain the transactional
assertions involving $\rel$ and $\wrset$ below.

Now consider the assertions in thread $\tidn{2}$. The precondition of
line 6 (which is also the precondition of line 5) contains a
conditional value assertion
$\langle \hat{f} = 1\rangle[d = 5]_{\tidn{2}}$. This assertion ensures
that, if $\tidn{2}$ reads the value 1 for $\hat{f}$ via an acquiring
(or release-acquiring) transaction, and this transaction successfuly
commits, then its view is guaranteed to be updated so that
$[d = 5]_{\tidn{2}}$ holds. In a transactional setting, we establish
this fact in three steps.
\begin{enumerate}[label=(\arabic*)]
\item After executing line~7, we use
  $\langle \hat{f} = 1\rangle[d = 5]_{\tidn{2}}$ to establish that
  $r_1 = 1 \imp [d \txndv 5]_{\tidn{2}}$ holds. Note that $r_2$ stores
  the value $1$ returned by a transactional read of $\hat{f}$. Thus,
  $\langle \hat{f} = 1\rangle[d = 5]_{\tidn{2}}$ is transformed into
  an implication after the execution of line~7. The assertion
  $[d \txndv 5]_{\tidn{2}}$ is a new assertion introduced in \logic,
  which states that if the transaction executed by $\tidn{2}$ commits,
  then $[d = 5]_{\tidn{2}}$ holds in the post-state.
\item If the transaction sucessfully commits (line~8), we use
  $r_1 = 1 \imp [d \txndv 5]_{\tidn{2}}$ to establish
  $r_1 = 1 \imp [d = 5]_{\tidn{2}}$ in the postcondition. Recall that
  all registers used by a transaction are set to $\bot$ when a
  transaction aborts, so if $\tidn{2}$ reaches line 9 by aborting the
  transaction, then this assertion is trivially true.
\item We use $r_1 = 1 \imp [d = 5]_{\tidn{2}}$ to establish
  $[d = 5]_{\tidn{2}}$ after the \kwdo-\kwuntil\ loop, using the guard
  $r_1 = 1$ at line~9.
\end{enumerate}
Finally, we use $[d = 5]_{\tidn{2}}$ in the precondition of line~10 to
establish the postcondition $r_2 = 5$. This is because
$[d = 5]_{\tidn{2}}$ guarantees that the only value $\tidn{2}$ can
read for $d$ is $5$.
 \end{example}

\subsection{\logic: Transactional Owicki-Gries Reasoning }
\label{sec:owicki-gries-reas}
Now that we have introduced the assertions used by \logic, we now
review the Owicki-Gries proof obligations. As discussed above, the use
of view-based assertions allows us to use the standard Owicki-Gries
theory. Regardless, we review the theory in the context of our
language, which supports (abstract) TM operations.  Formally, we model
programs as a labelled transition system, given by the syntax in
\reffig{fig:syntax}.

A command (of type $\AComm$) is either a local assignment $r := \Exp$,
a {\em store} to a shared location $\loc :=^{[{\sf R}]} \Exp$, a {\em
  load} from a shared location $r \gets^{[{\sf A}]} \loc$, a
compare-and-swap
$r \gets \kwcas^{\sf [RX] [R] [A]}(\loc, \valb, \val)$, or a
transactional operation. The annotations $\sf RX$, $\sf R$ and $\sf A$
are optional, as indicated by the brackets `[' and `]'. Thus, for
example, both $\loc := e$ and $\loc :=^{\sf R} e$ are valid load
commands; the former is relaxed and the latter is releasing. A
$\kwcas$ may be annotated to be relaxed, or release and/or
acquire. Note that a $\kwcas$ returns a boolean to indicate whether or
not the compare-and-swap has been successful.
$\txbegin{[{\sf RX}] [{\sf R}] [{\sf A}]}$, $\txread{r}{\loc}$,
$\txwrite{\loc}{\val}$ and $\txend$ are transactional operations, as
defined by the \wtms automata in \reffig{fig:wtms}.

\begin{figure}[t]
  \centering
$\begin{array}{@{} l @{}}
	\valb,\val \!\in\! \Val 
  \eqdef \nat
	\qquad
	\loc, \locb,\ldots \!\in\! \Var 
	\qquad
  r, r_1, r_2 \ldots \!\in\! \Reg \qquad 
	\tid, \tidn{1}, \tidn{2},\ldots \!\in\! \TID \eqdef \nat
  \qquad 
	i, j, k,\ldots \!\in\! \Label \\
\begin{array}{@{} r @{\hspace{2pt}} l @{}} 
  e \in \Exp ::= & v \mid r \mid e {+} e \mid \cdots 
  \\
  B \in \BExp ::= & 
                  \True \mid B \land B \mid \cdots \\
  \alpha \in \AComm ::=     &  
                              r := \Exp
                   \mid \loc :=^{\sf [R]} \Exp \mid 
 r \gets^{\sf [A]} \loc \mid r \gets \kwcas^{\sf [RX] [R] [A]}(\loc, \valb, \val)
                   \mid {} 
  \\
                 & \txbegin{[{\sf RX}] [{\sf R}] [{\sf A}], 2^\Reg} \mid \txread{r}{\loc}
                   \mid  \txwrite{\loc}{\val} \mid \txend
  \\[1mm]
  	ls \in \LComm ::= 
  	&  \stepgoto{\alpha}{j} \mid \ifgoto{B}{j}{k} 
  \\
	\prog \in \Prog \eqdef 
	& \TID \times \Label \to \LComm 	
\end{array} 
 \end{array}$
 \vspace{-0.5em}

  \caption{Language syntax}
  \Description{Language syntax}
  \label{fig:syntax}
  \vspace{-1em}
\end{figure}

We use a program counter variable $pc: \TID \to Label$ to model
control flow, and model a program $\prog$ as a function mapping each
pair $(\tid, i)$ of thread identifier and label to the \emph{labelled
  statement} (in $\LComm$) to be executed.  A labelled statement may
be
\begin{enumerate*}[label=(\roman*)]
\item a plain statement of the form $\stepgoto{\alpha}{j}$, comprising
  an atomic statement $\alpha$ to be executed and the label $j$ of the
  next statement; or
\item a conditional statement of the form $\ifgoto{B}{j}{k}$ to
  accommodate branching, which proceeds to label $j$ if $B$ holds and
  to $k$, otherwise.
\end{enumerate*}
We assume a designated label, $\initlabel \in \Label$,
representing the \emph{initial label}; \ie each thread begins
execution with $\pc(\tid) = \initlabel$.  Similarly,
$\finlabel \in \Label$ represents the \emph{final label}.

We let $\Assertion$ be the set of \emph{assertions} that use
view-based expressions.  We model program
annotations via an \emph{annotation function},
$\ann \in \Ann = \TID \times \Label \to \Assertion$, associating each
program point $(\tid, i)$ with its associated assertion.  A {\em proof
  outline} is a tuple $(\ina, \ann, \fina)$, where
$\ina, \fina \in \Assertion$ are the initial and final assertions.

\begin{definition}[Validity]
\label{def:outline}
A proof outline $(\ina, \ann, \fina)$ is {\em valid} for a program
$\prog$ iff each of the following holds:

\begin{description}
\item [Initialisation] For all $\tid \in \TID$,
    $\ina \imp \ann(\tid, \initlabel)$.
  \item[Finalisation] $(\forall \tid \in \TID.\, \ann(\tid, \finlabel)) \imp \fina$
  \item[Local correctness] For all $\tid \in \TID$ and $i \in \Label$, either:
    \begin{itemize}
    \item $\prog(\tid, i) = \stepgoto{\alpha}{j}$ and
      ${\assert{\ann(\tid, i)}}\ \alpha\ {\assert{\ann(\tid, j)}}$; or 
    \item $\prog(\tid, i) = \ifgoto{B}{j}{k}$ and both
      $\ann(\tid, i) \wedge B \imp \ann(\tid,j)$ and
      $\ann(\tid, i) \wedge \neg B \imp \ann(\tid,k)$ hold.
    \end{itemize}
  \item[Stability]     For all $\tid_1, \tid_2 \in \TID$ such that $\tid_1 \neq \tid_2$ and
    $i_1, i_2 \in \Label$ if $\prog(\tid_1, i_1) = \stepgoto{\alpha}{j}$, then

      $\assert{\ann(\tid_2, i_2) \wedge \ann(\tid_1, i_1)}\ \alpha\ \assert{\ann(\tid_2, i_2)}$ 
\end{description}
\end{definition}
\noindent
Intuitively, {\sf Initialisation} (resp. \textsf{Finalisation})
ensures that the initial (resp. final) assertion of each thread holds
at the beginning (resp. end); {\sf Local correctness} establishes
validity for each thread; and {\sf Stability} ensures that each
(local) thread annotation is {\em interference-free} under the
execution of other threads~\cite{DBLP:journals/acta/OwickiG76}.


To support Owicki-Gries reasoning, we have proved a number of
high-level rules, extending those of Dalvandi et
al.~\cite{DBLP:conf/ppopp/DalvandiD21,ECOOP20} to cope with
transactional assertions from \refsec{sec:view-based-assert} and the
transactional commands. For instance, the following rules are used in
the proof of transactional message passing. A number of other rules
are provided as part of our Isabelle/HOL development.
\begin{lemma} \label{lemma:proof-rules-logic}
  Suppose $\tidn{1} \neq \tidn{2}$. Then each of the following holds:
    \begin{align*} 
      & \assert{\True}\ \raext{${\tt TxWrite}_\tid(x, v)$} \ \assert{(\hat{x},v) \in \wrset_\tid} 
      \\
      & \assert{(x,u) \in \wrset_{\tidn{1}} \wedge \rel_{\tidn{1}} \wedge {} {[}\hat{x}\not\approx u]_{\tidn{2}} \land [y = v]_{\tidn{1}}}\ \raext{${\tt TxEnd}_{\tidn{2}}$} \ \assert{ \langle \hat{x} =u \rangle [y = v]_{\tidn{2}}} 
      \\
      & \assert{(x, \_) \not \in \wrset_{\tid} \wedge Acq_\tid {}  {} \wedge \langle \hat{x} =u \rangle [y = v]_\tid }\ \raext{${\tt TxRead}_\tid(x, r)$} \ \assert{ (\hat{x}, r) \in \rdset_\tid \wedge {}  (r = u \imp [y \txndv v]_\tid) } 
      \\
      & \assert{ (\hat{x}, r) \in \rdset_\tid \wedge 
      (r = u \imp [y \txndv v]_\tid) }\ \raext{${\tt TxEnd_\tid}$} \ \assert{ r = u  \Rightarrow [y =m]_\tid} 
  \end{align*}

\end{lemma}
The rules in \reflem{lemma:proof-rules-logic} have been verified in
Isabelle/HOL w.r.t. the operational semantics. 
Once proved, they can be used to show validity of proof outlines such
as those in \reffig{fig:mp-tm-proof} without having to consult the
operational semantics.

\begin{theorem}
  The proof outline in \reffig{fig:mp-tm-proof} is valid.
\end{theorem}
\noindent
This theorem has been verified in Isabelle/HOL, and it makes extensive
use of generic proof rules such as the ones proved in
\reflem{lemma:proof-rules-logic}. In particular, given such lemmas,
like in previous
works~\cite{DBLP:conf/esop/BilaDLRW22,DBLP:journals/corr/abs-2004-02983,DBLP:conf/ppopp/DalvandiD21,ECOOP20},
Isabelle/HOL is automatically able to find and apply the appropriate
proof rule using the built-in {\tt sledgehammer}
tool~\cite{DBLP:conf/cade/BohmeN10}. This automation has been key to
scaling mechanised verification of proof outlines in view-based
logics. 
For example, the proofs of TML-RA (see \refsec{sec:cont-refin-b})
requires verification of complex invariants and proof outlines, and
these proofs make use of the proof rules developed in prior
work~\cite{ECOOP20}. Similarly, TARO can be applied to verify more
complex porgrams that use transactions, for instance if one were to
develop transactional data structures. Interestingly, because
transactions provide isolation guarantees, many of the proofs are
simplified since the stability checks for in-flight transactions
become trivial.

The proof outlines for the programs in Figs.~\ref{fig:mp-tm-2} and
\ref{fig:ra-trans-mem} are provided in
Appendix~\ref{sec:additional-examples}. 


\section{Proving correctness of \tmlra}
\label{sec:cont-refin-b}

We now turn to the question of correctness of \tmlra with respect to
the \wtms specification.


\subsection{Refinement and Simulation for Weak Memory}
\label{sec:refin-simul-weak}

Since we have an operational semantics with an interleaving semantics
over weak memory states, the development of our refinement theory
closely follows the standard approach under
SC~\cite{DBLP:books/cu/RoeverE1998}.
Suppose $P$ is a program with initialisation $\Init$. An
\emph{execution} of $P$ is defined by a possibly infinite sequence
$\Delta_0\, \Delta_1\, \Delta_2\,\dots$ such that
\begin{enumerate}[leftmargin=*]
\item each $\Delta_i$ is a 4-tuple $(P_i, \ls_i, \gamma_i, \beta_i)$
  comprising a program to be executed, local state, global library
  state and global client state, and
\item
  $(P_0, \ls_0, \gamma_0, \beta_0) = (P, \ls_\Init, \gamma_\Init,
  \beta_\Init)$, and
\item for each $i$, we have $\Delta_i \Longrightarrow \Delta_{i+1}$, where
  $\Longrightarrow$ is the transition relation of the program (as
  defined by the operational semantics).
\end{enumerate}
Let $\LVar_P$ be the set of local variables corresponding to a program
$P$. If $P$ is a client, a \emph{client trace} corresponding to an
execution $\Delta_0\,\Delta_1\,\Delta_2\dots$ is a sequence
$\ct \in \Sigma_P^*$ such that
$\ct_i = (\pi_{2}(\Delta_i)_{|P}, \pi_{4}(\Delta_i))$, where $\pi_n$ is a
projection function that extracts the $n$th component of a given tuple
and $\ls_{|P}$ restricts the given local state $\ls$ to the
variables in $\LVar_P$. Thus, each $\ct_i$ is the global client state
component of $\Delta_i$. After such a projection, the concrete
implementation may contain (finite or infinite)
stuttering~\cite{DBLP:books/cu/RoeverE1998}, i.e., consecutive states
in which the client state is unchanged. We let ${\it rem\_stut}(\ct)$
be the function that removes all stuttering from the trace $\ct$,
i.e., each consecutively repeating state is replaced by a single
instance of that state. We let $\Tr(P)$ denote the set of
\emph{stutter-free traces} of a program $P$, i.e., the stutter-free
traces generated from the set of all executions of $P$.

Below we refer to the client that uses the abstract object as the
\emph{abstract client} and the client that uses the object's
implementation as the \emph{concrete client}. The notion of contextual
refinement that we develop ensures that a client is not able to distinguish 
the use of a concrete implementation in place of an abstract
specification. In other words, each thread of the concrete client
should only be able to observe the writes (and updates) in the client
state (i.e., $\gamma$ component) that the thread could already observe
in a corresponding of the client state of the abstract client.
First we define trace refinement for weak memory states. 
\begin{definition}[State and Trace Refinement]
  \label{def:cont-refin-1}
  We say a concrete client state $(\ls, \beta_C)$ is a
  \emph{refinement} of an abstract client state $(\als, \beta_A)$,
  denoted $(\ls, \beta_C) \leq (\als, \beta_A)$ iff
  $\ls = \als$ and 
  for all threads $\tid$ and $x \in \GVar$, we have
  $\beta_C.\OW_\tid(x) \subseteq \beta_A.\OW_\tid(x)$. 
  We say a concrete client trace $\ct$ is a \emph{refinement} of an
  abstract client trace $\at$, denoted $\ct \leq \at$, iff
  $\ct_i \leq \at_i$ for all $i$.
\end{definition}
This now leads to a natural trace-based definition of contextual refinement. 
\begin{definition}[Program Refinement]
  \label{def:prog-ref}
  A concrete program $P_C$ is a \emph{refinement} of an abstract
  program $P_A$, denoted $P_C \leq P_A$, iff for any (stutter-free)
  client trace $\ct \in \Tr(P_C)$ there exists a (stutter-free) client
  trace $\at \in \Tr(P_A)$ such that $\ct \leq \at$.
\end{definition}
Finally, we obtain a notion of contextual refinement for abstract
objects. We let $P[O]$ be the client program calling operations from
object $O$. Note that $O$ may be an abstract object, in which case
execution of each method call follows the abstract object semantics, or a concrete implementation.
\begin{definition}[Contextual refinement]
  \label{def:cref}
  We say a concrete object $CO$ is a \emph{contextual refinement} of
  an abstract object $AO$ 
   iff for any client
  program $P$, we have $P[CO] \leq P[AO]$.
\end{definition}


Here, we use a \emph{simulation-based} proof
method, which is a standard
technique from the literature that establishes refinement between
\wtms and \tmlra. The difference in a relaxed memory setting is that
the refinement relation is between more complex configurations of the
form $(\ls, \gamma, \beta)$, where $\ls$ describes the local state,
$\gamma$ is the client state and $\beta$ is a state of the TM in
question. 
In particular, a {\em simulation relation}, $R$, relates triples
$\Gamma_A \eqdef (\als, \gamma_A, \beta_A)$ of the abstract system
with triples $\Gamma_C \eqdef (\ls, \gamma_C, \beta_C)$ of the
concrete system. 

The definition below assumes a reflexive relation
$\gamma_C \trans{\tview_\tid} \gamma_C'$ for each thread $\tid$ that
arbitrarily advances the thread view of $\tid$ (for one or more
locations).
\begin{definition}[Forward simulation] \label{def:fsim} For an
  abstract object $AO$ and a concrete object $CO$, for a client
  program $P$, we say
  $R(\Gamma_A, \Gamma_C) \eqdef  R_V((\als, \beta_A), (\ls, \beta_C)) \wedge
  R_O((\als_{|AO}, \gamma_A), (\ls_{|CO}, \gamma_C))$ is a forward
  simulation between $A$ and $C$ iff each of the following holds:
  \begin{description}
  \item [Client observation.]\ 
    
    $
    \begin{array}[t]{@{}l}
      R_V((\als, \beta_A), (\ls, \beta_C)) =  \begin{array}[t]{@{}l@{}}
                      \als_{|P} = \ls_{|P} \wedge (\forall \tid\in \TID, \loc \in \Var.\ \ 
                      \beta_A.\tview(t, \loc) \leq \beta_C.\tview(t, \loc))
                    \end{array}
    \end{array}
    $
    
    
  \item [Thread view stability.] For any thread $\tid$,
    
    $\begin{array}[t]{@{}l@{}}
      R_O((\als_{|AO}, \gamma_A), (\ls_{|CO}, \gamma_C)) \wedge
      (\gamma_C \trans{\tview_\tid} \gamma_C') 
       \Rightarrow 
      R_O((\als_{|AO}, \gamma_A), (\ls_{|CO}, \gamma_C'))
    \end{array}$


  \item [Initialisation.] For any 
    concrete initial state $\Gamma_C^0$, there exists an abstract initial state
    $\Gamma_A^0$ such that
    $R(\Gamma_A^0, \Gamma_C^0)$.
\item [Preservation.]For any concrete states $\Gamma_C$, $\Gamma_C'$ such that $C$
  can take an atomic transition from $\Gamma_C$ to $\Gamma_C'$, if
  $\Gamma_A$ is an abstract state such that $R(\Gamma_A, \Gamma_C)$,
  then either
    \begin{itemize}
    \item $R(\Gamma_A, \Gamma_C')$, or \hfill (stuttering step)
    \item there exists a transition of $A$ from $\Gamma_A$ to some
      state $\Gamma_A'$ such that $R(\Gamma_A', \Gamma_C')$.
      \hfill (non-stuttering step)
    \end{itemize}
  \end{description}
\end{definition} {\em Initialisation} and {\em preservation} are
standard components of a forward simulation. {\em Client observation}
is necessary in a relaxed memory context to ensure that the
client-side observations of the concrete system are possible
observations of the abstract system. In particular, if an abstract
object specifies a particular client-side synchronisation, then this
synchronisation must also be present in the concrete implementation
(see~\cite{DBLP:conf/ppopp/DalvandiD21}). {\em Thread view stability}
guarantees that the $R_O$ component of the refinement relation is
preserved when the thread view in the library is shifted forward,
e.g., due to synchronisation within a client. 

Note that \refdef{def:fsim} only guarantees preservation of safety. To
additionally preserve liveness, further progress guarantees are
required in an
implementation~\cite{DBLP:conf/icalp/GotsmanY11,DBLP:conf/icfem/DongolG16}. We
leave liveness preservation through refinement for future work since
notions of fairness and progress of weak memory models is still at the
early stages~\cite{DBLP:journals/pacmpl/LahavNOPV21}.

\begin{restatable}{theorem}{FSimSound}
  If $R$ is a forward simulation between $AO$ and $CO$,
  then for any client $P$ we have $P[CO] \leq P[AO]$.
\end{restatable}


\newcommand{\wc}{\mathtt{wc}}
\newcommand{\odd}{\mathtt{odd}}
\newcommand{\even}{\mathtt{even}}
\newcommand{\grel}{\mathtt{gRel}}
\newcommand{\lastfn}{\mathtt{lastval}}
\newcommand{\txnrel}{\mathtt{txnRel}}
\newcommand{\hasRead}{\mathtt{hasRead}}
\newcommand{\hasWritten}{\mathtt{hasWritten}}
\newcommand{\mviewrel}{\mathtt{mViewRel}}
\newcommand{\cviewrel}{\mathtt{cViewRel}}
\newcommand{\evenwsetempty}{\mathtt{locValWrSetRel}}
\newcommand{\oddwsetempty}{\mathtt{locOddValWrSetRel}}
\newcommand{\hasreadrsetempty}{\mathtt{hasReadRdSetRel}}
\newcommand{\wsetlastval}{\mathtt{wrSetLastValRel}}
\newcommand{\validIdx}{\mathtt{validIdx}}

\subsection{Forward Simulation for  \tmlra}

Perhaps the most technically challenging aspect of this paper is the
proof of \refthm{sec:forw-simul-tmlra} below, which ensures the
correctness of \tmlra w.r.t. \wtms. %

This section describes the simulation relation used to prove
refinement between \tmlra and TMS2-RA. Validity of the forward
simulation itself has been verified using Isabelle/HOL. The refinement
relation
\[
  R((\als, \gamma_A, \beta_A), (ls, \gamma_C, \beta_C)) \ \ \eqdef\ \
  \begin{array}[t]{@{}l@{}}
  R_V((\als, \beta_A), (\ls, \beta_C)) \wedge \refeq{eq:rr0} \wedge {} \\
  (\forall \txid.\ 
  \refeq{rr1} \land \refeq{rr2}  \land \refeq{rr4} \land \refeq{rr5} \land \refeq{rr6} \land \refeq{rr7} \land \refeq{rr8} \land \refeq{rr9})
  \end{array}
\] 

\noindent The first conjunct \refeq{eq:rr0} in the refinement relation
$R$ states that the value of the last write to $\glb$ divided by 2
($\wc(n) \eqdef n \div 2$) is equal to the last version of history
written to $\memories$.
\begin{equation}
  \label{eq:rr0}
  \wc(\gamma_C.\lastfn(glb)) = |\gamma_A.\memories|
\end{equation}
The next conjunct, \refeq{rr1}, states that the last value written to
any location $l$ in $\gamma_C$ is either the value of $l$ in the last
abstract memory index or in the write set of the executing transaction
\begin{equation} \label{rr1}
  \begin{split}
        & \forall l . ~l \neq \glb \Rightarrow
        \begin{array}[t]{@{}l@{}}
          \gamma_C.\lastfn(l) \in \{\gamma_A.\memories_{|\gamma_A.\memories|}(l), 
          \gamma_A.\wset_\txid(l)\}
        \end{array}
\end{split}
\end{equation}

\noindent where $\lastfn(\loc)$ is a function that returns the value of the
last write written to a location $\loc$.


The next conjunct \refeq{rr2} is an on-the-fly simulation relation 
(i.e. the transaction has begun and is not committed or aborted) and
 states that for all threads $\tid$ if transaction $\txid$ ($\gamma_A.\txn_\tid = \txid$) is on-the-fly,
  the value of $\wc(\gamma_C.loc_\tid)$ will be greater
 than or equal to $\beginIdx_\txid$
and the read set of $\txid$ will be consistent with memory verison 
$\wc(\gamma_C.loc_\tid)$:
\begin{equation}\label{rr2}
    \begin{split}
        \gamma_A.\beginIdx_\txid \leq \wc(ls.loc_\txid)  \land \gamma_A.\rset_\txid \subseteq \gamma_A.\memories_{\wc(ls.loc_\txid)}
    \end{split}
\end{equation}

Conjunct \refeq{rr4} states that if the value of $ls.loc_\txid$ is even then write set of $\gamma_A$
must be empty: 
\begin{equation}\label{rr4}
    \begin{split}\even(ls.loc_\txid) \Rightarrow \gamma_A.\wset_\txid  = \emptyset
    \end{split}
\end{equation}

Also if a transaction $\txid$ that has already written to a 
location then the write set of 
$\gamma_A$ is not empty:
\begin{equation}\label{rr5}
    \begin{split}
          ls.\hasWritten_\txid  \Rightarrow \gamma_A.\wset_\txid \noteq \emptyset
    \end{split}
\end{equation}

If a transaction has not read any location yet in the concrete state,
 then the read set of the abstract state should be empty:
\begin{equation}\label{rr6}
    \begin{split}
       \neg ls.\hasRead_\txid \Rightarrow \gamma_A.\rset_\txid = \emptyset
    \end{split}
\end{equation}

If there is a write in the write set of the abstract state, then the value
should match the value of the last write written to that location by the concrete
 implementation:
\begin{equation}\label{rr7}
    \begin{split}
        \forall  l  \in \dom(\gamma_A.\wset_\txid) . \ \ \gamma_A.\wset_\txid(l) = \gamma_C.\lastfn(l)
\end{split}
\end{equation}

The value of a visible write of thread $\tau$ to variable $glb$ divided by two
is a visible memory by thread $\tau$ of the abstract state:
\begin{equation}\label{rr8}
    \begin{split}
\forall w \in \gamma_C.OW_{\tau}(glb) . \ \  \wc(\valu(w)) \in \gamma_A.vmems_{\tau}
\end{split}
\end{equation}

All seen memory indices by the abstract transaction $t$ are less the value of thread view of 
$glb$ for thread $\tau$ divided by 2:
\begin{equation}\label{rr9}
  \begin{split}
\forall i \in \gamma_A.\seenIdxs_{t} . \ \  i \leq \wc(\valu(\gamma_C.\tview_{\tau}(glb))) 
\end{split}
\end{equation}

\begin{theorem}
  \label{sec:forw-simul-tmlra}
  $R$ is a forward simulation between \wtms and \tmlra.
\end{theorem}
\begin{proof}
  This theorem has been verified in Isabelle/HOL.
\end{proof}

\section{Related work}
\label{sec:related-work}
\paragraph{Verifying C11 programs.}  There are now several different
approaches to program verification that support different aspects of
the C11 relaxed memory model using pen-and-paper proofs
(\eg~\cite{DBLP:conf/icalp/LahavV15,DBLP:conf/oopsla/TuronVD14,DBLP:conf/popl/AlglaveC17,DBLP:conf/esop/DokoV17}),
model checking
(\eg~\cite{DBLP:conf/pldi/Kokologiannakis19,DBLP:conf/pldi/AbdullaAAK19}),
specialised tools
(\eg~\cite{DBLP:conf/pldi/TassarottiDV15,DBLP:conf/esop/KrishnaEEJ20,DBLP:conf/esop/SvendsenPDLV18,DBLP:conf/tacas/Summers018}),
and generalist theorem provers (\eg~\cite{ECOOP20}). These cover a
variety of (fragments of) memory models and proceed via exhaustive
state space exploration, separation logics, or Hoare-style calculi. A
related approach to TARO that uses a view-based semantics for
persistent x86-TSO has been developed by
\citet{DBLP:conf/esop/BilaDLRW22}.

Another series of works has focussed on semantics that support the
{\em relaxed dependencies} that are allowed by
C11~\cite{Batty2020,DBLP:conf/popl/KangHLVD17,DBLP:conf/pldi/LeeCPCHLV20,DBLP:journals/pacmpl/JagadeesanJR20}. These
have been followed more recently by logics and verification over this
semantics~\cite{DBLP:journals/corr/abs-2108-01418,DBLP:conf/esop/SvendsenPDLV18}. However,
relaxed dependencies produce high levels of non-det\-er\-mi\-nism,
making verification significantly more complex. We consider a
verification framework that supports relaxed dependencies and STMs to
be a topic for future research.

More recent works include {\em robustness} of C11-style programs,
which aims to show ``adequate synchronisation'' so that the relaxed
memory executions reduce to executions under stronger memory
models~\cite{10.1145/3434285}. Such reductions, although automatic,
are limited to finite state systems, and a small number of
threads. Furthermore, it is currently unclear how they would handle
client-library synchronisation or relaxed (non-SC) specifications.

\paragraph{Correctness conditions under relaxed memory.}
Following the extensive literature on the semantics of relaxed memory
architectures, a natural next question has been the development of
library abstractions for relaxed memory. One aim has been to ensure {\em
  observational refinement} and {\em compositionality} of the
implemented objects. A series of works have considered reforumulations
of {\em
  linearizability}~\cite{DongolJRA18,ifm18,DBLP:journals/pacmpl/RaadDRLV19}
by presenting suitable weakenings fine-tuned to the underlying memory
model. This includes extensions of linearizability, \eg so that it is
defined in terms of axiomatic (aka declarative) relaxed memory
models~\cite{DongolJRA18,DBLP:journals/pacmpl/RaadDRLV19} and those
that are based on the more abstract concept of execution
structures~\cite{ifm18}. Recent works have covered verification of
relaxed memory concurrent data structures that have been developed to
satisfy the conditions described
above~\cite{DBLP:conf/esop/KrishnaEEJ20,DBLP:journals/pacmpl/RaadDRLV19,DBLP:conf/ppopp/DalvandiD21},
but none of these cover transactions. 

\citet{DBLP:conf/esop/KhyzhaL22} have recently developed notions of
abstraction for crash resilient libraries, providing correctness
conditions (extending linearizability) that ensure contextual
refinement for concurrent objects executed over the PSC (persistent
sequential consistency) model. They do so by exposing the internal
synchronisation mechanisms that are used to implement an object in the
history (in addition to the invocations and responses). Our work
differs since we consider transactional memory libraries as opposed to
concurrent objects, use a different memory model and focus on
verification of contextual refinement directly. Nevertheless, in
future work, it would be interesting to see if their methods provide
an alternative method for specifying concurrent object and
transactional memory libraries in C11.

Several papers have revisited transaction semantics in the context of
relaxed memory
models~\cite{DBLP:conf/vmcai/RaadLV19,DBLP:conf/ppopp/DongolJR19,DongolJR18,DBLP:conf/pldi/ChongSW18}. Raad
\etal have considered relaxed memory and \emph{snapshot
  isolation}~\cite{DBLP:conf/esop/RaadLV18,DBLP:conf/vmcai/RaadLV19},
which is a weaker condition that serializability (and hence opacity
and TMS2). The question of whether snapshot isolation can be fully
exploited by implementations in a relaxed memory setting remains a
topic of future research, with most transactional implementations
aiming to support at least
serializability~\cite{DBLP:journals/taco/ZardoshtiZBSS19}. \citet{DongolJR18}
and \citet{DBLP:conf/pldi/ChongSW18} have provided {\em axiomatic}
atransactional semantics integrated with relaxed memory models,
focussing on hardware memory models and hardware
transactions. \citet{DBLP:conf/pldi/ChongSW18} additionally propose a
model for C11 transactions, but these models are focussed on
transactions within the compiler, as opposed to STMs. 
Finally, the axiomatic models proposed
in these earlier works~\cite{DongolJR18,DBLP:conf/pldi/ChongSW18} are
not suitable for operational verification, \eg as supported by \logic,
where we require an operational semantics as provided by \wtms.

Another set of works has focussed on distributed (relaxed)
transactions~\cite{DBLP:conf/ecoop/XiongCRG19,DBLP:conf/esop/BeillahiBE21,DBLP:journals/lmcs/BeillahiBE21}. Although
there are analogues between transactions in distributed systems and
relaxed memory, constraints such as replication consistency and
session order are not factors in shared memory, and hence the
underlying issues are fundamentally different. \citet{DBLP:conf/ecoop/XiongCRG19} describe a
taxonomy of distributed transactional models supported by an
operational semantics. It would be
interesting to investigate whether \logic can be adapted to cope with
client-object systems in their models.

\paragraph{Relaxed memory TM implementations.} There is a set of
recent works on implementing TM algorithms in
C11~\cite{Spear2020,DBLP:journals/taco/ZardoshtiZBSS19}. The focus
here has been real-world implementability of STMs via compiler
support. Since the focus is on benchmarks and real-world workflows,
these works neither consider a formal semantics nor provide a
verification framework. Our work can thus be seen as providing a
formal basis to support to these efforts. In particular, we show how
the serialisability specifications assumed by
\citet{Spear2020,DBLP:journals/taco/ZardoshtiZBSS19} can be relaxed,
without impacting correctness, while improving performance.

\section{Conclusions}
\label{sec:conclusion}

In this paper, we have presented a new approach to release-acquire
transactions for RC11 RAR (a fragment of C11 that supports relaxed as
well as release-acquire atomics). We have developed a new TM
specification, \wtms, that extends TMS2 to a relaxed memory context by
describing the interactions between transactions and their clients. We
implement \wtms by \tmlra, which is an adaptation of an existing eager
algorithm, TML.  We show that \tmlra outperforms \tmlsc using the
STAMP benchmarks.

Our second set of contributions covers the verification of
release-acquire TM implementations. We focus on proofs at two levels:
\begin{enumerate*}[label=(\roman*)]
\item correctness of {\em client programs} that use \wtms, and
\item correctness of {\em implementations} of \wtms.
\end{enumerate*}
For (i), we have developed a logic, \logic,
extending~\cite{DBLP:conf/ppopp/DalvandiD21}, and used this logic to
prove that \wtms does indeed guarantee the desired client-side
synchronisation properties. For (ii), we have applied a simulation
method, simular to~\cite{DBLP:conf/ppopp/DalvandiD21} and proved a
forward simulation between \tmlra and \wtms. All proofs for (i) and
(ii) as well as all meta-level soundness results are fully mechanised
in the Isabelle/HOL proof assistant, providing a high level of
assurance to our results.

Our motivation for using TML as the main implementation case study was
to start with a simple algorithm with an existing proof in
SC~\cite{DBLP:journals/fac/DerrickDDSTW18}. TML performs a global
synchronisation through a CAS on a single location, which degrades
performance on write-heavy workloads. For improved scalability, there
are more sophisticated algorithms like
TL2~\cite{DBLP:conf/wdag/DiceSS06} that offer per-location locking as
well as hybrid TM implementations~\cite{DBLP:conf/asplos/MatveevS15}
that combine hardware and software TM. TMS2 is known to be a
sufficient abstraction for hybrid TMs in
SC~\cite{DBLP:conf/forte/ArmstrongD17}, so it is likely that \wtms
also provides a basis for developing and verifying relaxed and
release-acquire versions of these more sophisticated algorithms.  We
leave such studies for future work.

\begin{acks}
The authors would also like to thank the anonymous referees for
their valuable comments and helpful suggestions. Dalvandi and Dongol are supported by 
 \grantsponsor{TMSurrey}{EPSRC}{https://epsrc.ukri.org/} Grant \grantnum{TMSurrey}{EP/R032556/1}.  
Dongol is additionally supported by 
 \grantsponsor{BD2}{EPSRC}{https://epsrc.ukri.org/} Grant \grantnum{BD2}{EP/V038915/1}, 
 \grantsponsor{BD3}{EPSRC}{https://epsrc.ukri.org/} Grant \grantnum{BD3}{EP/R025134/2}, \grantsponsor{BD4}{ARC}{https://www.arc.gov.au/} Grant \grantnum{BD4}{DP190102142} and VeTSS. 
\end{acks}

\bibliography{references}

\newpage
\appendix



\begin{figure*}[!t]
  \centering \small
  $\inference[{\sc Read}] {a \in \{rd(\loc, \val), rd^\mathsf{A}(\loc,
    \val) \} \qquad w \in \gamma.\OW_\tid(\loc) \qquad
    \xxval(w) = \val \\
    \tview' = \kwif\ \gamma.\relst(w) \wedge a = rd^\mathsf{A}(\loc,
    \val)\ \kwthen\ \gamma.\tview_\tid \otimes\gamma.\mview_{w}\ \kwelse\ \gamma.\tview_\tid[x := w]}
    {\gamma\  \strans{a}_{\tid}\  \gamma[\tview_\tid \asgn \tview']}$
    \bigskip

    \bigskip
    
  $ \inference[{\sc Write}] {
    a \in \{ wr(\loc,\val), wr^{\sf R}(\loc,\val)\} \qquad \wr \in \gamma.\OW_\tid(x) \setminus \gamma.\covered \qquad \fresh_\gamma(\loc, q) \qquad  \xxtst(w) < q \\
    \writes' = \gamma.\writes \cup \{(\loc,\val, q)\} \qquad
    \tview' = \gamma.\tview_\tid[x := (\loc,\val, q)] \qquad  b = (a = wr^{\sf R}(\loc,\val))
    }
    {\gamma\  \strans{a}_{\tid}\  \gamma[\tview_\tid \asgn \tview', \mview_{(\loc,\val, q)} \asgn \tview', \writes \asgn \writes', \relst \asgn \gamma.\relst[w \asgn b]]}$
    \bigskip

        \bigskip
  $
    \inference[{\sc RMW-RA}] {
    a = rmw^{\sf
      RA}(\loc,\valb,\val) 
    \qquad \wr \in \gamma.\OW_\tid(\loc) \setminus \gamma.\covered
    \qquad
    \xxval(\wr) = \valb
    \qquad 
    \fresh_\gamma(\loc, q) \qquad \xxtst(w) < q \\ \writes' = \gamma.\writes \cup \{(\loc, \val, q)\} \qquad
    \covered' = \gamma.\covered \cup \{(\loc, \val, q)\} 
    \\
    \tview' = \kwif\ \gamma.\relst(w) \ \kwthen\ \gamma.\tview_\tid[x := (\loc, \val, q)] \otimes\gamma.\mview_{w}\ \kwelse\ \gamma.\tview_\tid[x := (\loc, \val, q)]
  }
    {\gamma\ \strans{a}_{t}\  \gamma\left[
      \begin{array}[c]{@{}l@{}}
        \tview_\tid \asgn \tview', \mview_{(\loc, \val, q)} \asgn \tview', \\
        \writes \asgn \writes',
        \covered \asgn \covered', \relst := \gamma.\relst[(\loc, \val, \ts) \asgn \True] 
      \end{array}\right]}$

  \caption{Memory semantics for reads, writes and updates, where
    $\fresh_\gamma(\loc, q)$ holds iff there is no write in
    $\gamma.\writes$ on $\loc$ with timestamp $q$, \ie $(\loc, \_, q) \notin \gamma.\writes$}
  \Description{Memory semantics for reads, writes and updates}
  \label{fig:surrey-opsem}
\end{figure*}

\section{RC11 RAR Operational Semantics}
\label{sec:oper-semant}
We now briefly review the C11 operational semantics (introduced by \cite{ECOOP20})
 used in our framework. 

\paragraph{Component State}
Here we detail the C11 state as modelled by the operational semantics. The first state component
is $\writes$ which is the set of all global writes to shared locations.
Each global write is represented by a tuple $(\loc, v, q)$, where 
$\loc$ is a shared location, $q$ is a rational number used as a timestamp, and $v$ is the value written by the
global write. 
The writes to each variable are totally ordered by timestamps.

For any write $w = (\loc, v, q)$, we have $\var(w)= \loc$,
$\tst(w) = q$ and $\valu(w) = v$.  The state needs to record the
writes that are observable by each thread. A function $\tview_t$ is
included in the state to record the viewfront of thread $t$ (i.e. the
latest write that a thread has seen so far). All the writes with a
timestamp greater than or equal to the timestamp of the viewfront of
thread $t$ are $observable$ by the thread.  Another component of the
state is a function $\mview_w$ that records the viewfront of write
$w$, which is set to be the viewfront of the thread that executed $w$
at the time of $w$’s execution.  $\mview_w$ is used to compute a new
value for $\tview_t$ if a thread $t$ synchronizes with $w$. The state
also must record If a global write is $releasing$. This is recorded by
a function $\relst_w$ which returns $\True$ if $w$ is releasing or
$\False$ otherwise.  Finally, the state maintains a variable
$\covered \subseteq \writes$.  The semantics assumes that each update
action occurs in the modification order immediately after the write
that it reads from to preserve the atomicity of updates.  To prevent
any newer write to intervene between any update and the write that it
reads from, we add all the writes read by an update operation to the
the covered set $\covered$ so newer writes should never interact with
covered writes.

 \paragraph{Initialisation} 
 The initial state $\gamma^\Init$ is defined as follows. 
\begin{align*}
  \gamma^\Init.\writes  &\eqdef  \{(\loc,0, 0) \mid \loc \in \Var\} \\
  \gamma^\Init.\covered  &\eqdef \emptyset \\
  \gamma^\Init.\relst  & \eqdef \lambda (\loc, 0, 0).\, \False \\
  \gamma^\Init.\tview_\tid  &\eqdef \lambda \loc \in \Var.\, (\loc, 0, 0) \\
  \gamma^\Init.\mview_{w}  &\eqdef \lambda \loc \in \Var.\, (\loc, 0, 0) 
\end{align*}

\paragraph{Transition semantics.}
The transition relation of our semantics for global reads and writes
is given in \reffig{fig:surrey-opsem}~\cite{ECOOP20}. 

\paragraph{\bf {\sc Read} transition by thread $\tid$.} Assume that $a$ is
either a relaxed or acquiring read to variable $x$, $w$ is a write to
$x$ that $t$ can observe (\ie $(w, q, v) \in \gamma.\OW_\tid(x)$), and the
value read by $a$ is the value written by
$w$. 
Each read causes the viewfront of $t$ to be updated. For an
unsynchronised read, $\tview_t$ is simply updated to include the new
write. A synchronised read causes the executing thread's view of the
executing component and context to be updated. In particular, for each
variable $x$, the new view of $x$ will be the later (in timestamp
order) of either $\tview_t(x)$ or $\mview_w(x)$. 


\paragraph{\bf {\sc Write} transition by thread $t$.} A write
transition must identify the write $(w, v, q)$ after which $a$
occurs. This $w$ must be observable and must {\em not} be covered ---
the second condition preserves the atomicity of read-modify-write (RMW)
updates. We must choose a fresh timestamp $q' \in \rat$ for $a$, which
for a C11 state $\gamma$ is formalised by $\fresh_\gamma(q, q') = q < q' \wedge \forall w' \in \gamma.\writes.\ q < \tst(w') \Rightarrow q' < \tst(w')$. 
That is, $q'$ is a new timestamp for variable $x$ and that
$(a,q', v')$ occurs immediately after $(w, v, q)$. The new write is added to
the set $\writes$. 
 
We update $\gamma.\tview_t$ to include the new
write, which ensures that $t$ can no longer observe any writes prior to
$(a, v', q')$. Moreover, we set the viewfront of $(a, v', q')$ to be the new
viewfront of $t$ in $\gamma$ together with the thread viewfront of the
environment state $\beta$. If some other thread synchronises with this
new write in some later transition, that thread's view will become at
least as recent as $t$'s view at this transition. Since $\mview$ keeps
track of the executing thread's view of both the component being
executed and its context, any synchronisation through this new write
will update views across components.

\paragraph{\bf {\sc Update} (aka RMW) transition by thread $t$.} These
transitions are best understood as a combination of the read and write
transitions. As with a write transition, we must choose a valid fresh
timestamp $q'$, and the state component $\writes$ is updated in the
same way. State component $\mview$ includes information from the new
view of the executing thread $t$.  As discussed earlier, in {\sc
  Update} transitions it is necessary to record that the write that
the update interacts with is now covered, which is achieved by adding
that write to $\covered$. Finally, we must compute a new thread view,
which is similar to a {\sc Read} transition, except that the thread's
new view always includes the new write introduced by the
update. 


\section{TMS2-RA Operational Semantics}
This section presents further details of the operational semantics
from \refsec{sec:wtms-1} of TMS2-RA as formalised in our Isabelle/HOL
development.

\paragraph{Component State} Here we present the transactional state of
TMS2-RA which builds on the earlier semantics of TMS2
\cite{DBLP:journals/fac/DohertyGLM13}.  The state space of TMS2-RA
extends the state space of TMS2 and comprises several components.  The
first component is $\memories$ which is a sequence that keeps track of
memory snapshots (memory states).  The status of each transaction
$\txid$ is stored in $\status_\txid$ which can have any of the
following values: ${\tt NOTSTARTED}$, ${\tt READY}$,
${\tt COMMITTED}$, ${\tt ABORTED}$. Each transaction $\txid$ also has
a write set ($\wset_\txid$) and a read set ($\rset_\txid$) where the
writes and reads performed by the transaction are stored. For each
transaction $\txid$ there is also a $\beginIdx_\txid$ variable which
is set to the most recent memory version when the transaction
begins. We extend the TMS2 state space with a number of new
components.  The $\seenIdxs_\txid$ variable stores the memory version
of all the reads performed by the transaction. $\mmview_i$ is the
{\em memory~modification~view} of a stored memory $i$
($0 \leq i \leq |\memories|-1$). Memory modification view of memory
$i$ is an snapshot of the transaction's client thread view at the time
of commiting taht memory. $\mrel_i$ indicates that if the transaction
that committed memory version $i$ was releasing.  

\paragraph{Initialisation} The initial state $\gamma^\Init$ is defined as follows:
\begin{align*}
  \gamma^\Init.\memories  &\eqdef  \langle (\lambda l\in Loc . ~ 0) \rangle\\
  \gamma^\Init.V  &\eqdef  \langle (\lambda l\in Loc.\ \ \  ~ \beta^\Init.\tview_\txid(l)) \rangle \\
  \gamma^\Init.\mrel &\eqdef \langle \False \rangle\\
  \gamma^\Init.\status_\txid  &\eqdef {\tt NOTSTARTED} \\
  \gamma^\Init.\rset_\txid  &\eqdef \emptyset \\
  \gamma_\init.\wset_\txid  & \eqdef \emptyset \\
  \gamma^\Init.\seenIdxs_\txid &\eqdef \emptyset \\
  \gamma^\Init.\beginIdx_\txid &\eqdef 0 
\end{align*}

\noindent where $\beta^\Init$ is the initial state of the client program.

\paragraph{Transition semantics of TMS2-RA} The transition relation of 
our TMS2-RA semantics for various transactional operations is given in 
\reffig{fig:tms2-opsem} and explained below.

\paragraph{\bf {\sc TxBegin} operation by transaction $\txid$.} A transaction
starts with a {\sc TxBegin} operation and should specify if it is a relaxed
or releasing/acquiring transaction. The {\sc TxBegin} operation takes $a$ and 
$r$ and sets the value of $\synctype_\txid$ accordingly. It also sets the
value $\beginIdx_\txid$ to be the latest memory version at the time of begining
the transaction and changes the $\status_\txid$ to ${\tt READY}$. Values of $\seenIdxs_\txid$,
$\rset_\txid$, and $\wset_\txid$ are set to empty.

\paragraph{\bf {\sc TxWrite} operation by transaction $\txid$.} Assuming that
$l$ is a location and $v$ is a value and $\status_\txid = {\tt READY}$, {\sc TxWrite}
operation of transaction $\txid$ adds the $l\mapsto v$ pair to the write set
of transaction $\txid$ (i.e. $\wset_\txid$). Other state components remain unchanged.

\paragraph{\bf {\sc TxRead} operation by transaction $\txid$.} The transactional read
operation ({\sc TxRead}) reads the value ($v$) of location $l$ from the transaction $\txid$'s
write set ($\wset_\txid$) if the transaction $\txid$ has previously wrote to the location $l$ 
($l \in \dom(\wset_\txid)$). In this case the state will remain unchanged.

If $\txid$ has not previously written to $l$ 
($l \notin \dom(\wset_\txid)$) then the value of $l$ will be read from a memory
version $i$ where $i$ is greater than or equal to $\beginIdx_\txid$ and 
the $\rset_\txid$ is consistent with that memory version (i.e. $\rset_\txid \subseteq \memories_i$).
In this case version $i$ will be added to $\seenIdxs_\txid$ and the transaction's 
read set ($\rset_\txid$) is also updated to include the $l\mapsto v$. $\seenIdxs$
is particularly important for synchronisation if the transaction is acquiring.

\paragraph{\bf {\sc TxEndRO} operation by transaction $\txid$.} 
A read-only transaction is committed by {\sc TxEndRO}.  If a
transaction is read-only ($\wset_\txid = \emptyset$) then it updates
$\status_\txid$ of the transaction state ($\gamma$) to be
${\tt COMMITTED}$ and will leave the rest of the state unchanged. If
the transaction was set to be acquiring by {\sc TxBegin}
($\isAcq_\txid = True$) and the transaction's read set is not empty
($\rset_\txid \neq \emptyset$) then the client's thread view
($\beta.\tview_\txid$) may also get synchronised as well.

 
  
\paragraph{\bf {\sc TxEndWR} operation by transaction $\txid$.} 
A writer transaction is committed by {\sc TxEndWR}. 
Simillar to a read-only transaction, a writer transaction 
($\wset_\txid \neq \emptyset$) will set $status_\txid$ to ${\tt COMMITTED}$. 
A writer transaction also add a new memory version $i$ to $\memories$
where $i$ is equal to the size of $\gamma.\memories$ meaning that 
the new memory version is added to the end of $\memories$ sequence.
The new memory version is obtained by overrwriting  the latest memory 
version in $\gamma.\memories$ with the transaction's write
set: $mem' = \last(\gamma.\memories) \oplus{} \wset_\txid$). 
The function $last$ is defined as $last(m) \eqdef m_{|m|-1}$.
The memory modification view ($V_i$) of the new memory version $i$ 
and the updated thread view of the clinet state ($\tview_\txid$) is going to 
be the same and determined by $\tview'$. The value of 
$\tview'$ is determined in the same way as explained for
{\sc TxEndRO}.

\paragraph{\bf {\sc TxAbort} operation by transaction $\txid$.}
If transaction $\txid$ aborts, {\sc TxAbort} will set $\status_\txid$
to ${\tt ABORTED}$.  At this point, other transaction operations
(except {\sc TxBegin}) cannot be executed. The client state $\beta$
will remain unchanged.












\newcommand{\ft}{{\it ft}}
\newcommand{\aft}{{\it aft}}

\section{Forward Simulation implies Observational Refinement}
\label{sec:forw-simul-impl}
We have already shown that there exists a forward simulation
between \wtms and \tmlra. In this section, we show that if there exists
a forward simulation between an abstract TM specification $AO$ and
a concrete implementation $CO$ as defined in \refdef{def:fsim}, then
for any client $P$, $P[CO]$ is a contextual refinement of $P[AO]$
($P[CO]\leq P[AO]$).

\FSimSound*
\begin{proof}[Proof]
  Assume $\ft$ is a full trace of $P[CO]$, where for each $i$, $\ft_i$
  is a triple is of the form $(\ls_i, \gamma_i, \beta_i)$. We show
  that there exists a full trace $\aft$ of $P[AO]$ such that
  $\xi(\ft) \leq \xi(\aft)$ (see \refdef{def:cont-refin-1}), where $\xi(\ft)$
  projects the full trace to the client trace, i.e., for each $i$,
  $\xi(\ft)_i = (\ls_i{_{|P}}, \beta_i)$ (and similarly $\xi(\aft)$),
  and additionally removes any stuttering.

  The proof is by induction over prefixes $\ft'$ of $\ft$.

  For the base case,
  $\ft' = \langle (\ls^\Init, \gamma_C^\Init, \beta_C^\Init) \rangle$
  is a trace containing just the initial state of $C$. By
  ``Initialisation'' of \refdef{def:fsim} there exists a
  $(\als^\Init, \gamma_A^\Init, \beta_A^\Init)$ such that
  \[
    R ((\als^\Init, \gamma_A^\Init, \beta_A^\Init), (\ls^\Init,
    \gamma_C^\Init, \beta_C^\Init))\]
  Moreover, by ``Client
  observation''  of \refdef{def:fsim}, we
  have $\als^\Init_{|P} = \ls^\Init_{|P}$ and for all threads
  $\tid$ and locations $x$,
  $\beta^\Init_C.\OW(\tid, x) \subseteq \beta^\Init_A.\OW(\tid,
  x)$. Thus, for $\xi(\ft)$ there exists an $\aft$ such that
  $\xi(\ft) \leq \xi(\aft)$.

  For our inductive hypothesis, assume the result holds for $\ft'$,
  i.e., there exists an abstrace prefix $\aft'$ of an abstract trace
  $\aft$ of $P[AO]$ such that $\xi(\ft') \leq \xi(\aft')$. Moreover,
  we assume that $R(\last(\aft'), \last(\ft'))$.  Suppose
  $\ft'' = \ft' \cdot \langle (\ls, \gamma_C, \beta_C) \rangle $ where
  $(\ls, \gamma_C, \beta_C)$ is generated by a concrete step from
  $\last(\ft')$. There are three possibilites based on the step taken
  by the concrete program.
\begin{itemize}
\item The first case is when the concrete program takes a library step
  and $\ft''$ is the program trace after the execution of the library
  step, because we already proved $CO \leq AO$ then by preservation
  rule of \refdef{def:fsim} we know that there exists an abstract
  trace
  $\aft'' = \aft' \cdot \langle (\als, \gamma_A, \beta_A) \rangle $
  such that $R((\als, \gamma_A, \beta_A), (\ls, \gamma_C,
  \beta_C))$. Moreover using ``Client observation'' of \refdef{def:fsim}, we have $\xi(\ft'') \leq \xi(\aft'')$.

\item The next case is when the concrete program takes a {\em
    non-synchronising} client step. Since the abstract and concrete
  takes the {\em same} non-synchronising step, the abstract and
  concrete library states remain unchanged, thus
  $R_O((\als_{|AO}, \gamma_A), (\ls_{|CO}, \gamma_C))$. Moreover, the
  thread view and local state of both abstract and concrete programs
  will be updated in the same way, thus we also have
  $R_V((\als, \beta_A), (\ls, \beta_C))$.  Finally, by the client
  observation property and the refinement relation $R$ are
  preserved. Since we match the concrete step at the abstract level,
  we have $\xi(\ft'') \leq \xi(\aft'')$.

\item The last case is when the concrete program takes a {\em
    synchronising} client step.  Like the unsynchronising case, the
  abstract and concrete client steps are identical and they both
  update the thread view and local state of the client in the same
  way, thus $R_V((\als, \beta_A), (\ls, \beta_C))$. Thus, we have
  $\xi(\ft'') \leq \xi(\aft'')$.  We must now show that the refinement
  relation holds in the post-state.  As opposed to the previous case,
  the synchronising client step potentially updates the library state
  by advancing the library thread view for the executing thread,
  $\tid$. However, this is equivalent to moving the library view
  forward using the ``Thread view stability'' property of
  \refdef{def:fsim}. Thus, we have
  $R_O((\als_{|AO}, \gamma_A), (\ls_{|CO}, \gamma_C))$.
\end{itemize}




\end{proof}













\section{Selection of TARO Proof Rules}
\label{sec:proof-rules}

We present a selection of TARO proof rules that we use in our
examples. The full development~\cite{Artifact} contains many other
rules, but we do not present all of these, since they are less
interesting. Note that $z$ below is a client variable.

\subsection{Assertions over {\tt TxBegin}}

\begin{enumerate}\addtolength{\itemsep}{5pt}
\item $\assert{\neg[\hat{x}   \approx   u]_{\tau}}~{\tt TxBegin}_{\_}(\_, \_)~\assert{\neg[\hat{x}   \approx   u]_{\tau}}$
\item $\assert{
    \langle\hat{x}   =   u\rangle[z   =   v]_{\tau} }~{\tt TxBegin}_{\_}(\_, \_)~\assert{\langle\hat{x}   =   u\rangle[z   =   v]_{\tau}   
  }$ 
\item 
  
  $\assert{[z = u]_{\tau}}~ {\tt TxBegin}_{\_}(\_, \_)~\assert{[z = u]_{\tau} }$  
   \item $\assert{
       \True 
   }~{\tt TxBegin}_\tau({\sf R}, \_)~\assert{Rel_\tau
      }$

   \item $\assert{\True}~{\tt TxBegin}_\tau({\sf A}, \_)~
   ~\assert{Acq_{\tau}
      }$
      
    \item $\assert{
        (\hat{y},v) \in \wrset_{\tau'} \wedge \tau   \neq      \tau'
      }~{\tt TxBegin}_{\tau}(\_,\_)
      ~\assert{(\hat{y},v) \in \wrset_{\tau'}   
      }$ 

    \end{enumerate}

\subsection{Assertions over {\tt TxRead}}

\begin{enumerate} \addtolength{\itemsep}{5pt}
    \item $\assert{\neg[\hat{x}   \approx   u]_{\tau}}~  {\tt TxRead}_{\_}(\_, \_) ~ \assert{\neg[\hat{x}   \approx   u]_{\tau}}$ 
    \item $\assert{
    [\hat{x}   =   u]_{\tau'}}~{\tt TxRead}_{\tau}(\_,\_)~\assert{status_{\tau}   =  {\tt READY} \imp [\hat{x}   =   u]_{\tau'}   
       }$ 
    \item $\assert{
    \wrset_\tid = \emptyset 
   }~{\tt TxRead}_{\tau}(\hat{x}, r)~\assert{(\hat x,r)   \in   \rdset_{\tid}   
      }$ 

    \item $\assert{
     %
    (\hat x, r)  \not\in  \wrset_\tau 
  }     ~ {\tt TxRead}_\tau(\hat x, r) ~ \assert{[\hat{x}   \approx   r]_\tau   
       }$

     \item 

       $\assert{
    [\hat{x}   =   u]_\tau
    \wedge \hat{x} \notin \dom(\wrset_{\tau})
    }~{\tt TxRead}_{\tau}(\hat{x}, r)~\assert{r   =   u
       }$


     \item $\assert{
         (\hat{y},v) \in \wrset_{\tau}   
       }~{\tt TxRead}_{\_}(\_, \_)~\assert{(\hat{y},v) \in \wrset_{\tau}   
       }$

     \item $\assert{
         (\hat y,v) \in \rdset_{\tau'} \wedge \tau   \neq      \tau'
       }     ~ {\tt TxRead}_{\tau}(\_, \_) ~ \assert{(\hat y,v) \in \rdset_{\tau'}   
       }$

\end{enumerate}

\subsection{Assertions over {\tt TxWrite}}

\begin{enumerate}\addtolength{\itemsep}{5pt}
\item $\assert{
    \neg[\hat{x}   \approx   u]_{\tau}} ~ {\tt TxWrite}_{\_}(\_ ,\_)~ 
  \assert{\neg[\hat{x}   \approx   u]_{\tau} }$

\item $\assert{[z   =   u]_{\tau}}~ {\tt TxWrite}_{\_}(\_ , \_)~\assert{[z   =   u]_{\tau}}$ 

\item $\assert{\True 
  }~ {\tt TxWrite}_\tau(\hat{y} ,  v)~\assert{(\hat{y},v) \in \wrset_\tid}$ 

     \item 
       $\assert{ \langle\hat{x} = u\rangle[z = v]_{\tau} \wedge w \neq
         u}~{\tt TxWrite}_{\tau}(\hat y, w) ~ \assert{\langle\hat{x} =
         u\rangle[z = v]_{\tau} }$
       
    \item $\assert{
        \hat x   \neq      \hat y
    \wedge (\hat{y},v) \in \wrset_{\tau}   
    }~ {\tt TxWrite}_{\_}(\hat x, \_) ~ \assert{(\hat{y},v) \in \wrset_{\tau}   
       }$

\end{enumerate}

\subsection{Assertions over {\tt TxEnd}}

\begin{enumerate} \addtolength{\itemsep}{5pt}
     \item $\assert{
       \langle\hat{x}   =   u\rangle[z   =   v]_{\tau'} \wedge \wrset_\tid = \emptyset 
     } ~{\tt TxEnd}_{\tau} ~ 
   \assert{\langle\hat{x}   =   u\rangle[z   =   v]_{\tau'}}$

 \item $\assert{
    \begin{array}[c]{@{}l@{}}
      (\hat{x}, u) \in \wrset_\tid 
      \wedge  \neg[\hat{x}\approx   u]_{\tau'} 
      \wedge {} \\ {[}z = v]_{\tau} 
      \wedge Rel_{\tau} 
      \wedge \tau \neq \tau'
    \end{array}
  }~    {\tt TxEnd}_{\tau} ~
  \assert{status_{\tid}   =   {\tt COMMITTED}  \imp \langle\hat{x}   =   u\rangle[z   =   v]_{\tau'}}$

   \item $\assert{\langle\hat{x}   =   u\rangle[z   =   v]_{\tau}   
   \wedge (\hat{x}, u) \notin \wrset_\tid   
   \wedge Rel_\tid
   }~{\tt TxEnd}_{\tau} ~\assert{\langle\hat{x}   =   u\rangle[z   =   v]_{\tau}   
      }$

    \item   $\assert{
    \langle\hat{x}   =   u\rangle[z   =   v]_{\tau'} 
   }~{\tt TxEnd}_{\tau} ~ \assert{status_{\tau}   =   {\tt ABORTED} \Rightarrow \langle\hat{x}   =   u\rangle[z   =   v]_{\tau'}   
      }$

   \item $\assert{
       \begin{array}[c]{@{}l@{}}
   \langle\hat{x}   =   u\rangle[z   =   v]_{\tau}   
   \wedge (\hat{x}, u) \in \rdset_\tau 
   \wedge  {} \\
    Acq_{\tau}
   \wedge \wrset_\tid = \emptyset 
       \end{array}
}~{\tt TxEnd}_{\tau}~\assert{status_{\tid}   = {\tt COMMITTED} \imp [z = v]_{\tau}}$


   \item $\assert{
   \neg[\hat{x}   \approx   u]_{\tau'}   
   \wedge   (\hat{x}, u) \notin \wrset_{\tid} 
 }~ 
    {\tt TxEnd}_{\tau} ~ 
   \assert{\neg[\hat{x}   \approx   u]_{\tau'}   
      }$
    \item $\assert{[y = v]_{\tau}} ~{\tt TxEnd}_{\tau} ~\assert{[y = v]_{\tau}}$

    \item 
      $\assert{
        [x   \approx   u]_{\tau'} \wedge \neg   Acq_{\tau}
      }~{\tt TxEnd}_{\tau}~\assert{[x   \approx   u]_{\tau'}   
      }$

     \item $\assert{(\hat{y},v) \in \wrset_{\tau'}}~ {\tt TxEnd}_{\tau} ~\assert{(\hat{y},v) \in \wrset_{\tau'}}$
       
    \end{enumerate}

\subsection{Assertions over client actions}

\begin{enumerate} \addtolength{\itemsep}{5pt}
   \item $\assert{
   \langle\hat{x}   =   u\rangle[z   =   v]_{\tau'}}  ~ r   \leftarrow_\tau   z' ~ \assert{\langle\hat{x}   =   u\rangle[z   =   v]_{\tau'}   
      }$

    \item $\assert{
    \langle\hat{x}   =   u\rangle[z   =   v]_{\tau}   
    \wedge
     z   \neq  z'}~z'   :=_{\tau}   m ~ \assert{\langle\hat{x}   =   u\rangle[z   =   v]_{\tau}   
       }$ 



\end{enumerate}

\section{Additional examples}
\label{sec:additional-examples}
This section provides the full proof outline for two additional examples.
The proof outline presented in \reffig{fig:rel-trans-po} includes two new 
assertions that were not introduced previously:
\begin{itemize}
  \item {\bf Memory-value assertion}: A memory-value assertion,
   denoted as $\memories[\hat x, v]_i$, holds iff the value of 
   location $x$ in memory version $i$ is $v$.
   \begin{align*}
    & \memories[\hat x, v]_i  \eqdef i \in \dom(\memories) \land \memories_i(\hat x) = v
 \end{align*}

 \item {\bf Never-written assertion}: A never-written assertion, 
 denoted as $NW[\hat x, v]$, holds iff none of the recorded memories in $\memories$
 has the value $v$ for location $x$.
 \begin{align*}
  & NW[\hat x, v] \eqdef \forall  i \in \dom(\memories).\  \neg \memories[\hat x, v]_i
  \end{align*}

\end{itemize}

A selection of proof rules over memory-value and never-written predicates are given below.

\begin{enumerate} \addtolength{\itemsep}{5pt}
\item $\assert{\memories[\hat{x},u]_i}
  ~{\tt TxBegin}_\tau(\_,\_)\assert{\memories[\hat{x}, u]_i}$

\item $\assert{\memories[\hat{x}, u]_i}
  ~{\tt TxRead}_\tau(\_,\_)\assert{\memories[\hat{x}, u]_i}$

\item $\assert{\memories[\hat{x}, u]_i}
  ~{\tt TxWrite}_\tau(\_,\_)\assert{\memories[\hat{x}, u]_i}$ 
  
   \item $\assert{
   (\hat{x}, u) \in \wrset_{\tau}   
   \wedge |\memories|  =   n
   }~{\tt TxEnd}_{\tau}~\assert{status_{\tau}   = {\tt COMMITTED} \imp \memories[\hat{x}   =   u]_n   
      }$

    \item $\assert{
        \memories[\hat{x}   =   v]_i  
        \wedge i   <   |\memories| -1  
      }~{\tt TxEnd}_\tau      ~\assert{\memories[\hat{x}   =   v]_i   
          }$
   \item $\assert{
  \memories[\hat{x}   =   u]_{i-1}   
   \wedge \memories[\hat{x}   =   v]_i   
   \wedge {} \\ \memories[\hat{y}   =   l]_{i-1}   
   \wedge \memories[\hat{y}   =   n]_i  \wedge {} 
   \\  u   \neq      v
   \wedge l   \neq      n \wedge i   =   |\memories|-1 \wedge{} \\
    \txn_\tid = t \wedge beginIndex_t  = i-1
    \wedge {} \\
    (\hat x,v)   \in   \rdset_{\tau}   
   \wedge \wrset_{\tau} = \emptyset   
 }~{\tt TxRead}_{\tau}(\hat y, r)~\assert{status_{\tau}   =   {\tt READY} \imp
   r   =   n
      }$ 
  
\item $\assert{
    NW[\hat x,   v]
  }~{\tt TxBegin}_\tau(\_,\_)~\assert{\neg[\hat{x}   \approx   v]_{\tau'}   
  }$

\item $\assert{
    NW[\hat x,   v]  
  }~{\tt TxBegin}_\tau(\_,\_)~\assert{NW[x,   v]   
  }$

\item $\assert{
    NW[\hat x,   v] 
  }~{\tt TxRead}_\tau(\_,  \_)~\assert{NW[x,   v]   
  }$ 

\item $\assert{
    NW[\hat x,   v]  
  }~{\tt TxWrite}_\tau(\_, \_)~\assert{NW[\hat x,   v]   
  }$ 
\item $\assert{
    NW[\hat x,   v]  
    \wedge  (\hat x,v) \not\in \wrset_\tau
  }~{\tt TxEnd}_\tau~\assert{NW[\hat x,   v]   
  }$
  
\end{enumerate}

\begin{theorem}
  The proof outlines in \reffig{fig:ra-trans-mem-po} and \reffig{fig:rel-trans-po} are valid.
\end{theorem}
\begin{proof}
  This theorem has been verified in Isabelle/HOL.
\end{proof}

\begin{figure*}[b]
  \small
$\begin{array}{@{}l@{~~~~}||@{~~~~}l@{~~~~}||@{~~~~}l}
   \multicolumn{3}{@{}l}{
 \qquad \qquad \qquad \qquad \assert{\forall \tid \in \{\tidn{1}, \tidn{2}, \tidn{3}\}.\ [\hat{f} = 0]_\tid \wedge [d_1 = 0]_\tid \wedge [d_2 = 0]_\tid}} \\
      \begin{array}[t]{@{}l@{}}
        {\bf Thread }\ \tidn{1}\\
        \assert{[\hat f \not \approx 1]_{\tidn{2}} \wedge
                [d_1 = 0]_{\tidn{1}} \wedge \neg {\it tf}_1} \\
        1\!: d_1 := 5; \\
        \assert{[\hat f \not \approx 1]_{\tidn{2}} \wedge
                [d_1 = 5]_{\tidn{1}}  \wedge \neg {\it tf}_1} \\
        2\!:\raext{${\tt TxBegin}(\sf R, \emptyset)$} \\
        \assert{[\hat f \not \approx 1]_{\tidn{2}}  \wedge
                [d_1 = 5]_{\tidn{1}} \wedge {} \\ \rel_{\tidn{1}} \wedge
                 \neg {\it tf}_1} \\
        3\!:\raext{${\tt TxWrite}(f, 1)$};\\
        \assert{
                [\hat f \not \approx 1]_{\tidn{2}} \wedge [d_1 = 5]_{\tidn{1}} \wedge {} \\ 
                 \rel_{\tidn{1}} \wedge (\hat f, 1) \in \wrset_{\tidn{1}} \wedge \neg {\it tf}_1} \\
        4\!:\raext{$\langle {\tt TxEnd}, {\it tf}_1 := \True \rangle $};\\
        \assert{\langle\hat f = 1 \rangle [d_1 = 5]_{\tidn{2}} \wedge {\it tf}_1}
      \end{array}
      & 
      \begin{array}[t]{@{}r@{~}l@{}}
        \multicolumn{2}{@{}l}{{\bf Thread }\ \tidn{2}}\\
        \multicolumn{2}{@{}l}{\assert{P(0) \wedge 
        ([\hat f \approx 1]_{\tidn{2}}  \imp {\it tf}_1)}} \\
        5\!:& d_2 := 10; \\
        \multicolumn{2}{@{}l}{\assert{P(10) \wedge 
        ([\hat f \approx 1]_{\tidn{2}}  \imp {\it tf}_1)}}
        \\
        6\!:& \raext{${\tt TxBegin}(\sf RA, \{r_2\})$} 
        \\
        \multicolumn{2}{@{}l}{\assert{P(10) \wedge ([\hat f \approx 1]_{\tidn{2}}  \imp {\it tf}_1)  \wedge {} \\
        \acq_{\tidn{2}} \wedge \rel_{\tidn{2}} \wedge \wrset_{\tidn{2}} = \emptyset 
        }}
        \\
        7\!: & \raext{${\tt TxRead}(f, r_2)$}; \\
        \multicolumn{2}{@{}l}{
        \assert{P(10) \wedge {} \\
        (r_2 = 1 \imp [d_1 \txndv 5]_{\tidn{2}} \wedge {\it tf}_1) \wedge {} \\
        \acq_{\tidn{2}} \wedge \rel_{\tidn{2}} \wedge {} \\
        \wrset_{\tidn{2}} = \emptyset \wedge (\hat f, r_2) \in \rdset_{\tidn{2}}}}
        \\
        8\!:& \raext{$\kwif\ r_2 = 1\ \kwthen$}  \\
        \multicolumn{2}{@{}l}{
        \quad \assert{P(10) \wedge [d_1 \txndv 5]_{\tidn{2}} \wedge {\it tf}_1 \wedge {} \\
         r_2 = 1 \wedge \acq_{\tidn{2}} \wedge \rel_{\tidn{2}} \wedge {} \\
        \wrset_{\tidn{2}} = \emptyset \wedge (\hat f, r_2) \in \rdset_{\tidn{2}}}}
        \\
        9\!:& \raext{$\quad {\tt TxWrite}(f, 2)$}
        \\
        \multicolumn{2}{@{}l}{
        \quad \assert{P(10) \wedge [d_1 \txndv 5]_{\tidn{2}} \wedge {\it tf}_1 \wedge {} \\
        r_2 = 1 \wedge \acq_{\tidn{2}} \wedge \rel_{\tidn{2}} \wedge 
          {} \\
        \wrset_{\tidn{2}} = \{(\hat f, 2)\} \wedge {} \\
        (\hat f, r_2) \in \rdset_{\tidn{2}} }}
        \\
        \multicolumn{2}{@{}l}{
        \assert{P(10) \wedge \acq_{\tidn{2}} \wedge \rel_{\tidn{2}} \wedge {} \\
        (r_2 \neq 1 \imp \wrset_{\tidn{2}} = \emptyset) \wedge {} \\
        \left(
        \begin{array}[c]{@{}l@{}}
          r_2 = 1 \imp [d_1 \txndv 5]_{\tidn{2}}    \wedge {} \\
          \qquad {\it tf}_1\wedge (\hat f, 2) \in \wrset_{\tidn{2}}
        \end{array}
        \right) \wedge {} \\
        (\hat f, r_2) \in \rdset_{\tidn{2}}}}
        \\
        10\!:&\raext{${\tt TxEnd}$};\\
        \multicolumn{2}{@{}l}{\assert{
        \langle \hat f = 2 \rangle [d_2 = 10]_{\tidn{3}} \wedge {} \\
        \left(
        \begin{array}[c]{@{}l@{}}
          r_2 \neq 1 \imp \\
          \quad [\hat f \not \approx 2]_{\tidn{1}}  \wedge [\hat f \not \approx 2]_{\tidn{3}}
        \end{array}
\right) \wedge {} \\
        \left(
        \begin{array}[c]{@{}l@{}}
          r_2  = 1 \imp\\
          \quad {\it tf}_1  \wedge \langle \hat f = 2 \rangle [d_1 =_{\tidn{3}} 5]
        \end{array}
\right)
        }}
      \end{array}
      & 
      \begin{array}[t]{@{}r@{~}l@{}}
        \multicolumn{2}{@{}l}{{\bf Thread }\ \tidn{3}}
        \\
        \multicolumn{2}{@{}l}{
        \assert{\langle \hat f = 2 \rangle [d_1 = 5]_{\tidn{3}} \wedge {} \\
        \langle \hat f = 2 \rangle [d_2 = 10]_{\tidn{3}} \wedge {} \\
          (\neg {\it tf}_1 \imp  [\hat f \not \approx 1]_{\tidn{3}} \wedge [\hat f \not \approx 2]_{\tidn{3}})}} 
        \\
        11\!:& \raext{${\tt TxBegin}(\sf A, \{r_3\})$} \\
        \multicolumn{2}{@{}l}{
        \assert{\langle \hat f = 2 \rangle [d_1 = 5]_{\tidn{3}} \wedge  {} \\
        \langle \hat f = 2 \rangle [d_2 = 10]_{\tidn{3}} \wedge {} \\
        (\neg {\it tf}_1 \imp  [\hat f \not \approx 1]_{\tidn{3}} \wedge [\hat f \not \approx 2]_{\tidn{3}}) {} \\
        \wedge {} 
        \wrset_{\tidn{3}} = \emptyset \wedge \acq_{\tidn{3}} \\
        }} 
        \\
        12\!: & \raext{${\tt TxRead}(f, r_3)$};\\
        \multicolumn{2}{@{}l}{
        \assert{
        \langle \hat f = 2 \rangle [d_1 = 5]_{\tidn{3}} \wedge {} \\
        \langle \hat f = 2 \rangle [d_2 = 10]_{\tidn{3}} \wedge {} \\ 
        \left(
        \begin{array}[c]{@{}l@{}}
          \neg {\it tf}_1 \imp \\
          \qquad [\hat f \not \approx 1]_{\tidn{3}} \wedge [\hat f \not \approx 2]_{\tidn{3}}
        \end{array}
        \right) \wedge {} \\
        \wrset_{\tidn{3}} = \emptyset \wedge {} \\
        (\hat f, r_3) \in \rdset_{\tidn{3}}   \wedge  \acq_{\tidn{3}} \wedge {} \\
        \left(
        \begin{array}[c]{@{}l@{}}
          r_3 = 2 \imp \\
          \qquad [d_1 \txndv 5]_{\tidn{3}}  \wedge [d_2 \txndv 10]_{\tidn{3}}
        \end{array}
        \right)
        }} 
        \\
        13\!:& \raext{${\tt TxEnd}$};\\
        \multicolumn{2}{@{}l}{
        \assert{
        r_3 = 2 \imp [d_1 = 5]_{\tidn{3}} \wedge [d_2 = 10]_{\tidn{3}}
        }}  \\
        14\!:& \kwif\ r_3 = 2\  \kwthen \\
        \multicolumn{2}{@{}l}{
        \qquad \assert{
        [d_1 = 5]_{\tidn{3}} \wedge [d_2 = 10]_{\tidn{3}}
        }}  \\
        15\!:& \ \ \ \ s_1 \gets d_1\\
        \multicolumn{2}{@{}l}{
        \qquad \assert{
        s_1 = 5 \wedge [d_2 = 10]_{\tidn{3}}
        }}  \\
        16\!:& \ \  \ \ s_2 \gets d_2\\
        \multicolumn{2}{@{}l}{
        \qquad \assert{
        s_1 = 5 \wedge s_2 = 10
        }}  \\
        \multicolumn{2}{@{}l}{
        \assert{
        r_3 = 2 \imp s_1 = 5 \wedge s_2 = 10
        }}
      \end{array}
   \end{array}$ \smallskip
   
 \hfill {\color{blue} $\{r_3 = 2 \imp s_1=5 \wedge s_2=10\}$}\hfill {}

 \vspace{-0.5em}
 \caption{Proof outline for extended transactional
   MP, where  $P(k) \eqdef [d_2 = k]_{\tidn{2}} \wedge
        \langle \hat f = 1\rangle[d_1 = 42]_{\tidn{2}} \wedge 
        {[} \hat f \not \approx 2]_{\tidn{1}} 
        \wedge {}  {[}\hat f \not \approx 2]_{\tidn{3}}$ 
 }
 \Description{Proof outline for extended transactional
   MP} \label{fig:ra-trans-mem-po} 
\end{figure*}


\begin{figure*}[b]
  \small
$\begin{array}{@{}l@{~~~~}||@{~~~~}l@{}}
    \begin{array}[t]{@{}l@{}}
      {\bf Thread }\ \tidn{1}\\
      \assert{ NW[\hat f, 1]   \land NW[\hat d2, 10]  \\
       \land [d1 = 0]_{\tidn{1}}  \land |\memories|  = 0} \\
      1\!: d_1 := 5; \\
      \assert{ NW[\hat f, 1]   \land NW[\hat d2, 10] \\
       \land [d1 = 5]_{\tidn{1}}  \land |\memories|  = 0} \\
      2\!:\raext{${\tt TxBegin}(\sf RX,\emptyset)$}; \\
      \assert{\neg[\hat f \approx 1]_{\tidn{2}}   \land \neg[\hat d2 \approx 10]_{\tidn{2}}
      \\
       \land [d1 = 5]_{\tidn{1}}    \land \neg\rel_{\tidn{1}}
      \\ 
       \land ~\status_{\tidn{1}} = R \\
       \land \neg \acq_{\tidn{1}}   \land |\memories| = 0}\\
      3\!:\raext{${\tt TxWrite}(d_2,10)$}; \\
      \assert{(\hat d_2, 10)\in \wrset_{t}  \land \neg[\hat f \approx 1]_{\tidn{2}} \\
       \land \neg[\hat d_2 \approx 10]_{\tidn{2}} 
       \land [d1 = 5]_{\tidn{1}}  \\
       \land \neg\rel_{\tidn{1}}  \land ~\status_{\tidn{1}} = R  \land \\
      \neg \acq_{\tidn{1}}   \land |\memories| = 0}\\
      4\!:\raext{${\tt TxWrite}(f, 1)$};\\
      \assert{(\hat f, 1)\in \wrset_{\tidn{1}}  \land (\hat d_2, 10)\in \wrset_{\tidn{1}} \\
       \land \neg[\hat f \approx 1]_{\tidn{2}}   \land
      \neg[\hat d2 \approx 10]_{\tidn{2}}\\
       \land [d1 = 5]_{\tidn{1}}    \land \neg\rel_{\tidn{1}}  \land \\ 
      ~\status_{\tidn{1}} = R  \land \neg \acq_{\tidn{1}}   \land |\memories| = 0}\\
   5\!:\raext{${\tt TxEnd}$};\\
      \assert{\True}
    \end{array}
    & 
    \begin{array}[t]{@{}r@{~}l@{}}
      \multicolumn{2}{l}{{\bf Thread }\ \tidn{2}} \\
      \multicolumn{2}{l}{
      \assert{\memories[\hat f = 0]_{0}  \land \memories[\hat d2 = 0]_{0}  \land \\
      (\status_{\tidn{1}} = C \Rightarrow \memories[\hat f = 1]_1   \land \memories[\hat d2 = 10]_1 )  \land \\
      (\status_{\tidn{1}} = C \Rightarrow ([\hat{f} \approx 0]_t  \lor [\hat{f} \approx 1]_t ))  \land \\
      (\status_{\tidn{1}} = C \Rightarrow ([\hat{d2} \approx 0]_t  \lor [\hat{d2} \approx 10]_t ))  \land \\
      (\forall v . v\not\in \{0,5\} \Rightarrow \neg [d1 \approx_{\tau2} v])  \land \\
      (\status_{\tidn{1}} \neq C \Rightarrow [\hat{f} = 0]_t   \land [\hat{d2} = 0]_t )  \land\\
      WS_t = \emptyset  \land \\
      (\forall i . i\neq 0  \land i \leq |\memories| \Rightarrow \neg \memories[\hat d2 = 0]_i )  \land \\
      (\forall i . i\neq 1  \land i \leq |\memories| \Rightarrow \neg \memories[\hat f = 1]_i ) }} \\
      6\!:& \raext{${\tt TxBegin}(\sf RX, \{r_1, r_2\})$} \\
      \multicolumn{2}{l}{\assert{\memories[\hat{f} = 0]_0   \land \memories[\hat{d2} = 0]_0   \land \\
      (\status_{\tidn{1}} = C \Rightarrow \memories[\hat{f} = 1]_1   \land \memories[\hat{d2} = 10]_1 )  \land \\
      (\status_{\tidn{1}} = C \Rightarrow ([\hat{f} \approx 0]_{\tau1}  \lor [\hat{f} \approx 1]_{\tau1} ))  \land \\
      (\status_{\tidn{1}} = C \Rightarrow ([\hat{d2} \approx 0]_{\tau1}  \lor [\hat{d2} \approx 10]_{\tau1} ))  \land \\
      (\forall v . v\not\in{0,5} \Rightarrow \neg[d1 \approx_{\tau2} v])  \land \\
      (\status_{\tidn{1}} \neq C \Rightarrow [\hat{f} = 0]_{\tau1}   \land [\hat{d2} = 0]_{\tau1} )   \land \\
       WS_{\tau1} = \emptyset  \land \\
                (\forall i . i\neq0  \land i \leq |M|  \Rightarrow \neg\memories[\hat{d2} = 0]_i )  \land \\
      (\forall i . i\neq1  \land i \leq |M|  \Rightarrow \neg\memories[\hat{f} = 1]_i )   \land \neg Acq_{\tau1} \\
      }}
       \\
      7\!:& \raext{${\tt TxRead}(f, r_1)$};\\
      \multicolumn{2}{l}{\assert{\memories[\hat{f} = 0]_0   \land \memories[\hat{d2} = 0]_0   \land \\
      (r1 = 1 \Rightarrow \status_{\tidn{1}} = C) \land \\
      (r1 = 1 \Rightarrow (\hat{f},1) \in RS_{\tau1} )  \land \\
      (r1 = 1 \Rightarrow \memories[\hat{f} = 1]_1   \land \memories[\hat{d2} = 10]_1 )   \land \\
      (\forall v . v\not\in{0,5} \Rightarrow \neg[d1 \approx_{\tau2} v] )  \land \\
      WS_{\tau1} = \emptyset  \land \\
                (\forall i . i\neq0  \land i \leq |M|  \Rightarrow \neg\memories[\hat{d2} = 0]_i )  \land \\
(\forall i . i\neq1  \land i \leq |M|  \Rightarrow \neg\memories[\hat{f} = 1]_i )   \land \neg Acq_{\tau1}
        }} \\
      8\!:& \raext{\kwif\ $r_1 = 1$\ \kwthen\ }\\
      \multicolumn{2}{l}{\assert{
        \memories[\hat{f} = 0]_0   \land \memories[\hat{d2} = 0]_0   \\
         \land \memories[\hat{f} = 1]_1   \land \memories[\hat{d2} = 10]_1   \\
         \land r1 = 1  \land (\hat{f},1) \in RS_{\tau1}   \land \\
        (\forall v . v\not\in{0,5} \Rightarrow \neg[d1 \approx_{\tau2} v] )  \land \\
        WS_{\tau1} = \emptyset  \land \\
                  (\forall i . i\neq0  \land i \leq |M|  \Rightarrow \neg\memories[\hat{d2} = 0]_i )  \land \\
(\forall i . i\neq1  \land i \leq |M|  \Rightarrow \neg\memories[\hat{f} = 1]_i )   \land \neg Acq_{\tau1}  \land
\status_{\tidn{1}} = C 
            }}\\
      9\!:& \raext{\quad${\tt TxRead}(d_2, r_2)$};\\
      \multicolumn{2}{l}{\assert{ (r1 = 1 \Rightarrow r2 = 10)  \land \\
      (\forall v . v\not\in{0,5} \Rightarrow \neg[d1 \approx_{\tau2} v] )  \land \\
      WS_{\tau1} = \emptyset  \land \neg Acq_{\tau1}
}} \\
      10\!:& \raext{${\tt TxEnd}$};\\
      \multicolumn{2}{l}{\assert{ (\forall v . v\not\in{0,5} \Rightarrow \neg[d1 \approx_{\tau2} v])  \land \\
      (\status_{\tidn{2}} = C  \land r1 = 1 \Rightarrow r2 = 10)  \\

      }}\\
      11\!:& r_3\gets d_1; \\
      \multicolumn{2}{l}{\assert{(\status_{\tidn{2}} = C  \land r_1 = 1 \Rightarrow r_2 = 10  \land r_3 \in\{0,5\})}}\\
    \end{array}
 \end{array}$ \medskip
   
 \hfill {\color{blue} $\{\status_{\tidn{2}} = C  \land r_1 = 1 \Rightarrow r_2 = 10  \land r_3 \in\{0,5\}\}$}\hfill {}

 \vspace{-0.5em}
 \caption{Proof outline for relaxed transactions where $C$ and $R$ are shorthand for ${\tt COMMITTED}$ and ${\tt READY}$, respectively.}
 \Description{Proof outline for relaxed transactions}
 \label{fig:rel-trans-po} 
\end{figure*}


\end{document}